\documentclass[12pt]{preprint}
\usepackage[full]{textcomp}
\usepackage[osf]{newtxtext} 
\usepackage{colortbl}
\usepackage{upgreek}
\usepackage{comment}
\usepackage{amssymb}
\usepackage{lmodern}
\usepackage{mathtools}

\usepackage{hyperref}
\usepackage{breakurl}
\usepackage{aligned-overset}
\usepackage{mhenvs}
\usepackage{mhequ} 
\newcommand{\be}{\begin{equation*}}
\newcommand{\ee}{\end{equation*}}
\usepackage{mhsymb}
\usepackage{booktabs}
\usepackage{tcolorbox}
\usepackage[utf8]{inputenc}
\usepackage{longtable}
\usepackage{wrapfig}
\usepackage{subcaption}
\usepackage{epsfig}
\usepackage{microtype}
\usepackage[dvipsnames]{xcolor}

\usepackage{centernot}
\usepackage{enumitem}
\usepackage{bm}
\usepackage{stackrel}
\usepackage{graphicx}

\makeatletter
\newcommand{\globalcolor}[1]{%
  \color{#1}\global\let\default@color\current@color
}
\makeatother

\usetikzlibrary{calc}
\usetikzlibrary{decorations}
\usetikzlibrary{positioning}
\usetikzlibrary{shapes}
\usetikzlibrary{shapes.misc}
\usetikzlibrary{arrows.meta, bending}

\tikzset{cross/.style={cross out, draw=black, fill=none, minimum size=2*(#1-\pgflinewidth), inner sep=0pt, outer sep=0pt}, cross/.default={2pt}}

\definecolor{blush}{rgb}{0.87, 0.36, 0.51}
	\definecolor{brightcerulean}{rgb}{0.11, 0.67, 0.84}
	\definecolor{greenryb}{rgb}{0.4, 0.69, 0.2}

\newif\ifdark
\darkfalse

\ifdark
\definecolor{darkred}{rgb}{0.9,0.2,0.2}
\definecolor{darkblue}{rgb}{0.7,0.3,1}
\definecolor{darkgreen}{rgb}{0.1,0.9,0.1}
\definecolor{franck}{rgb}{0,0.8,1}
\definecolor{pagebackground}{rgb}{.15,.21,.18}
\definecolor{pageforeground}{rgb}{.84,.84,.85}
\pagecolor{pagebackground}
\AtBeginDocument{\globalcolor{pageforeground}}
\definecolor{symbols}{rgb}{0,0.7,1}
\colorlet{connection}{red!80!black}
\colorlet{boxcolor}{blue!50}

\else

\definecolor{darkred}{rgb}{0.7,0.1,0.1}
\definecolor{darkblue}{rgb}{0.4,0.1,0.8}
\definecolor{darkgreen}{rgb}{0.1,0.7,0.1}
\definecolor{franck}{rgb}{0,0,1}
\definecolor{pagebackground}{rgb}{1,1,1}
\definecolor{pageforeground}{rgb}{0,0,0}
\colorlet{symbols}{blue!90!black}
\colorlet{connection}{red!30!black}
\colorlet{boxcolor}{blue!50!black}

\fi

\def\slash{\leavevmode\unskip\kern0.18em/\penalty\exhyphenpenalty\kern0.18em}
\def\dash{\leavevmode\unskip\kern0.18em--\penalty\exhyphenpenalty\kern0.18em}

\DeclareMathAlphabet{\mathbbm}{U}{bbm}{m}{n}

\DeclareFontFamily{U}{BOONDOX-calo}{\skewchar\font=45 }
\DeclareFontShape{U}{BOONDOX-calo}{m}{n}{
  <-> s*[1.05] BOONDOX-r-calo}{}
\DeclareFontShape{U}{BOONDOX-calo}{b}{n}{
  <-> s*[1.05] BOONDOX-b-calo}{}
\DeclareMathAlphabet{\mcb}{U}{BOONDOX-calo}{m}{n}
\SetMathAlphabet{\mcb}{bold}{U}{BOONDOX-calo}{b}{n}
\setlist{noitemsep,topsep=4pt,leftmargin=1.5em}

\DeclareMathAlphabet{\mathbbm}{U}{bbm}{m}{n}

\DeclareMathAlphabet{\mcb}{U}{BOONDOX-calo}{m}{n}
\SetMathAlphabet{\mcb}{bold}{U}{BOONDOX-calo}{b}{n}
\DeclareFontFamily{U}{mathx}{\hyphenchar\font45}
\DeclareFontShape{U}{mathx}{m}{n}{
      <5> <6> <7> <8> <9> <10>
      <10.95> <12> <14.4> <17.28> <20.74> <24.88>
      mathx10
      }{}
\DeclareSymbolFont{mathx}{U}{mathx}{m}{n}
\DeclareMathSymbol{\bigtimes}{1}{mathx}{"91}

\providecommand{\figures}{false}
{ \ifthenelse{\equal{\figures}{false}} {#1}{\[ {\rm Figure \ missing !} \]} }{}

\def\graft{\rightarrow}

\def\CE{\mathcal{E}}

\tikzstyle{tinydots}=[dash pattern=on \pgflinewidth off \pgflinewidth]
\tikzstyle{superdense}=[dash pattern=on 4pt off 1pt]

\newcommand{\mcE}{\mathcal{E}}

\newcommand{\mcV}{\mathcal{V}}

\newcommand{\mcS}{\CS}

\newcommand{\mcG}{\mathcal{G}}

\newcommand{\beq}{\begin{equation}}
\newcommand{\eeq}{\end{equation}}

\usepackage{empheq}

\newcommand{\mfF}{\mathfrak{F}}

\newcommand{\mfn}{\mathfrak{n}}

\newcommand{\mfe}{\mathfrak{e}}

\newcommand{\mfd}{\mathfrak{d}}
\newcommand{\mfa}{\mathfrak{a}}
\newcommand{\mfc}{\mathfrak{c}}

\newcommand{\mft}{\mathfrak{t}}

\newcommand{\mfS}{\mathfrak{S}}

\newcommand{\mfb}{\mathfrak{b}}

\newcommand{\mfG}{\mathfrak{G}}
\newcommand{\mfg}{\mathfrak{g}}

\newcommand{\Lab}{\mathfrak{L}}

\def\x{\boldsymbol{x}}

\def\n{\boldsymbol{n}}

\newcommand{\p}{\boldsymbol{p}}

\DeclareSymbolFont{rmlargesymbols}{OMX}{mdbch}{m}{n}
\DeclareMathSymbol{\rmintop}{\mathop}{rmlargesymbols}{82}

\newcommand{\D}{\partial}
\newcommand{\dint}{\mathrm{d}}
\newcommand\Item[1][]{%
  \ifx\relax#1\relax  \item \else \item[#1] \fi
  \abovedisplayskip=0pt\abovedisplayshortskip=0pt~\vspace*{-\baselineskip}}

\newcommand{\out}{\mathrm{out}}

\newenvironment{DIFnomarkup}{}{} % see man latexdiff

\theorembodyfont{\rmfamily}

\newfont{\indic}{bbmss12}

\makeatletter
\pgfdeclareshape{crosscircle}
{
  \inheritsavedanchors[from=circle] % this is nearly a circle
  \inheritanchorborder[from=circle]
  \inheritanchor[from=circle]{north}
  \inheritanchor[from=circle]{north west}
  \inheritanchor[from=circle]{north east}
  \inheritanchor[from=circle]{center}
  \inheritanchor[from=circle]{west}
  \inheritanchor[from=circle]{east}
  \inheritanchor[from=circle]{mid}
  \inheritanchor[from=circle]{mid west}
  \inheritanchor[from=circle]{mid east}
  \inheritanchor[from=circle]{base}
  \inheritanchor[from=circle]{base west}
  \inheritanchor[from=circle]{base east}
  \inheritanchor[from=circle]{south}
  \inheritanchor[from=circle]{south west}
  \inheritanchor[from=circle]{south east}
  \inheritbackgroundpath[from=circle]
  \foregroundpath{
    \centerpoint%
    \pgf@xc=\pgf@x%
    \pgf@yc=\pgf@y%
    \pgfutil@tempdima=\radius%
    \pgfmathsetlength{\pgf@xb}{\pgfkeysvalueof{/pgf/outer xsep}}%  
    \pgfmathsetlength{\pgf@yb}{\pgfkeysvalueof{/pgf/outer ysep}}%  
    \ifdim\pgf@xb<\pgf@yb%
      \advance\pgfutil@tempdima by-\pgf@yb%
    \else%
      \advance\pgfutil@tempdima by-\pgf@xb%
    \fi%
    \pgfpathmoveto{\pgfpointadd{\pgfqpoint{\pgf@xc}{\pgf@yc}}{\pgfqpoint{-0.707107\pgfutil@tempdima}{-0.707107\pgfutil@tempdima}}}
    \pgfpathlineto{\pgfpointadd{\pgfqpoint{\pgf@xc}{\pgf@yc}}{\pgfqpoint{0.707107\pgfutil@tempdima}{0.707107\pgfutil@tempdima}}}
    \pgfpathmoveto{\pgfpointadd{\pgfqpoint{\pgf@xc}{\pgf@yc}}{\pgfqpoint{-0.707107\pgfutil@tempdima}{0.707107\pgfutil@tempdima}}}
    \pgfpathlineto{\pgfpointadd{\pgfqpoint{\pgf@xc}{\pgf@yc}}{\pgfqpoint{0.707107\pgfutil@tempdima}{-0.707107\pgfutil@tempdima}}}
  }
}
\makeatother

\def\decorate#1#2{
        \ifnum#2>0
    		\foreach \count in {1,...,#2}{
	       	let
				\p1 = (sourcenode.center),
                \p2 = (sourcenode.east),
				\n1 = {\x2-\x1},
				\n2 = {1mm},
				\n3 = {(1.3+0.6*(\count-1))*\n1},
				\n4 = {0.7*\n1}
			in 
        		node[rectangle,fill=symbols,rotate=30,inner sep=0pt,minimum width=0.2*\n2,minimum height=\n2] at ($(sourcenode.center) + (\n3,\n4)$) {}
				}
		\fi
        \ifnum#1>0
    		\foreach \count in {1,...,#1}{
	       	let
				\p1 = (sourcenode.center),
                \p2 = (sourcenode.east),
				\n1 = {\x2-\x1},
				\n2 = {1mm},
				\n3 = {(1.3+0.6*(\count-1))*\n1},
				\n4 = {0.7*\n1}
			in 
        		node[rectangle,fill=symbols,rotate=-30,inner sep=0pt,minimum width=0.2*\n2,minimum height=\n2] at ($(sourcenode.center) + (-\n3,\n4)$) {}
				}
		\fi
}

\tikzset{
    dectriangle/.style 2 args={
        triangle,
        alias=sourcenode,
        append after command={\decorate{#1}{#2}}
    },
    dectriangle/.default={0}{0},
}

\tikzset{
	dot/.style={circle,fill=black,inner sep=0pt, minimum size=1mm},
	graydot/.style={circle,draw=gray,inner sep=0pt, minimum size=1mm},
	loopnode/.style={circle,draw=black,inner sep=0pt, minimum size=1.5mm},
	treenode/.style={circle,fill=black,inner sep=0pt, minimum size=1.5mm},
	charge/.style={circle,draw=black,inner sep=0pt, minimum size=3mm},
	loopline/.style={->,semithick,shorten >=1pt,shorten <=1pt},
	treeline/.style={semithick, densely dashed, shorten >=1pt,shorten <=1pt},
	homoge/.style={font=\scriptsize,draw=black,fill=white,inner sep=2pt},
	homo/.style={font=\scriptsize},
	KK/.style={thick,shorten >=1pt,shorten <=1pt, densely dashed, >=latex},
kerAlg/.style={thick},
kerAlg2/.style={semithick,dashed},
kerP/.style={very thick},
arrow/.style={->,thick,shorten >=3pt,shorten <=3pt},
renorm/.style={ultra thick},
	kernel/.style={semithick,shorten >=1pt,shorten <=1pt},
	}

\tikzset{
	cross/.style={path picture={ 
  		\draw[symbols]
			(path picture bounding box.south east) -- (path picture bounding box.north west) (path picture bounding box.south west) -- (path picture bounding box.north east);
		}},
root/.style={circle,fill=green!50!black,inner sep=0pt, minimum size=1.2mm},
        dot/.style={circle,fill=pageforeground,inner sep=0pt, minimum size=1mm},
        dotred/.style={circle,fill=pageforeground!50!pagebackground,inner sep=0pt, minimum size=2mm},
        var/.style={circle,fill=pageforeground!10!pagebackground,draw=pageforeground,inner sep=0pt, minimum size=3mm},
        kernel/.style={semithick,shorten >=2pt,shorten <=2pt},
        kernels/.style={snake=zigzag,shorten >=2pt,shorten <=2pt,segment amplitude=1pt,segment length=4pt,line before snake=2pt,line after snake=5pt,},
        rho/.style={densely dashed,semithick,shorten >=2pt,shorten <=2pt},
           testfcn/.style={dotted,semithick,shorten >=2pt,shorten <=2pt},
        renorm/.style={shape=circle,fill=pagebackground,inner sep=1pt},
        labl/.style={shape=rectangle,fill=pagebackground,inner sep=1pt},
        xic/.style={very thin,circle,draw=symbols,fill=symbols,inner sep=0pt,minimum size=1.2mm},
        g/.style={very thin,rectangle,draw=symbols,fill=symbols!10!pagebackground,inner sep=0pt,minimum width=2.5mm,minimum height=1.2mm},
        xi/.style={very thin,circle,draw=symbols,fill=symbols!10!pagebackground,inner sep=0pt,minimum size=1.2mm},
	xies/.style={very thin,rectangle,fill=green!50!black!25,draw=symbols,inner sep=0pt,minimum size=1.1mm},
	xiesf/.style={very thin,rectangle,fill=green!50!black,draw=symbols,inner sep=0pt,minimum size=1.1mm},
        xix/.style={very thin,crosscircle,fill=symbols!10!pagebackground,draw=symbols,inner sep=0pt,minimum size=1.2mm},
        X/.style={very thin,cross,rectangle,fill=pagebackground,draw=symbols,inner sep=0pt,minimum size=1.2mm},
	xib/.style={thin,circle,fill=symbols!10!pagebackground,draw=symbols,inner sep=0pt,minimum size=1.6mm},
	xie/.style={thin,circle,fill=green!50!black,draw=symbols,inner sep=0pt,minimum size=1.6mm},
	xid/.style={thin,circle,fill=symbols,draw=symbols,inner sep=0pt,minimum size=1.6mm},
	xibx/.style={thin,crosscircle,fill=symbols!10!pagebackground,draw=symbols,inner sep=0pt,minimum size=1.6mm},
	kernels2/.style={very thick,draw=connection,segment length=12pt},
	keps/.style={thin,draw=symbols,->},
	kepspr/.style={thick,draw=connection,->},
	krho/.style={thin,draw=symbols,superdense,->},
	krhopr/.style={thick,draw=connection,superdense},
	triangle/.style = { regular polygon, regular polygon sides=3},
	not/.style={thin,circle,draw=connection,fill=connection,inner sep=0pt,minimum size=0.5mm},
	diff/.style = {very thin,draw=symbols,triangle,fill=red!50!black,inner sep=0pt,minimum size=1.6mm},
	diff1/.style = {very thin,dectriangle={1}{0},fill=red!50!black,draw=symbols,inner sep=0pt,minimum size=1.6mm},
	diff2/.style = {very thin,dectriangle={1}{1},fill=red!50!black,draw=symbols,inner sep=0pt,minimum size=1.6mm},
		diffmini/.style = {very thin,rectangle,fill=black,draw=black,inner sep=0pt,minimum size=0.75mm},
	 kernelsmod/.style={very thick,draw=connection,segment length=12pt},
	 rec/.style = {very thin,rectangle,fill=black,draw=black,inner sep=0pt,minimum size=2mm},
	cerc/.style={very thin,circle,draw=black,fill=symbols,inner sep=0pt,minimum size=2mm},
	stars/.style={very thin,star,star points=6,star point ratio=0.5, draw=black,fill=red,inner sep=0pt,minimum size=0.7mm},
	>=stealth,
        }
        \tikzset{
root/.style={circle,fill=black!50,inner sep=0pt, minimum size=3mm},
        circ/.style={circle,fill=white,draw=black,very thin,inner sep=.5pt, minimum size=1.2mm},
        round1/.style={fill=white,outer sep = 0,inner sep=2pt,rounded corners=1mm,draw,text=black,thin,minimum size=1.2mm},
          circ1/.style={circle,fill=red!10,draw=red,very thin,inner sep=.5pt, minimum size=1.2mm},
        rect/.style={fill=white,outer sep = 0,inner sep=2pt,rectangle,draw,text=black,thin,minimum size=1.2mm},
        rect1/.style={fill=white,outer sep = 0,inner sep=2pt,rectangle,draw,text=black,thin,minimum size=1.2mm},
        round2/.style={fill=red!10,outer sep = 0,inner sep=2pt,rounded corners=1mm,draw,text=black,thin,minimum size=1.2mm},
       round3/.style={fill=blue!10,outer sep = 0,inner sep=2pt,rounded corners=1mm,draw,text=black,thin,minimum size=1.2mm}, 
        rect2/.style={fill=black!10,outer sep = 0,inner sep=2pt,rectangle,draw,text=black,thin,minimum size=1.2mm},
        dot/.style={circle,fill=black,inner sep=0pt, minimum size=1.2mm},
        dotred/.style={circle,fill=black!50,inner sep=0pt, minimum size=2mm},
        var/.style={circle,fill=black!10,draw=black,inner sep=0pt, minimum size=3mm},
        kernel/.style={semithick,shorten >=2pt,shorten <=2pt},
         diag/.style={thin,shorten >=4pt,shorten <=4pt},
        kernel1/.style={thick},
        arrow/.style={->,thick},
        kernels/.style={snake=zigzag,shorten >=2pt,shorten <=2pt,segment amplitude=1pt,segment length=4pt,line before snake=2pt,line after snake=5pt,},
		kernels1/.style={snake=zigzag,segment amplitude=0.5pt,segment length=2pt},
		rho1/.style={densely dotted,semithick},
        rho/.style={densely dashed,semithick,shorten >=2pt,shorten <=2pt},
           testfcn/.style={dotted,semithick,shorten >=2pt,shorten <=2pt},
           visible/.style={draw, circle, fill, inner sep=0.25ex},
        renorm/.style={shape=circle,fill=white,inner sep=1pt},
        labl/.style={shape=rectangle,fill=white,inner sep=1pt},
        xic/.style={very thin,circle,fill=symbols,draw=black,inner sep=0pt,minimum size=1.2mm},
        xi/.style={very thin,circle,fill=blue!10,draw=black,inner sep=0pt,minimum size=1.2mm},
	xib/.style={very thin,circle,fill=blue!10,draw=black,inner sep=0pt,minimum size=1.6mm},
	xie/.style={very thin,circle,fill=green!50!black,draw=black,inner sep=0pt,minimum size=1mm},
	xid/.style={very thin,circle,fill=symbols,draw=black,inner sep=0pt,minimum size=1.6mm},
	edgetype/.style={very thin,circle,draw=black,inner sep=0pt,minimum size=5mm},
	nodetype/.style={very thick,circle,draw=black,inner sep=0pt,minimum size=5mm},
	kernels2/.style={very thick,draw=connection,segment length=12pt},
clean/.style={thin,circle,fill=black,inner sep=0pt,minimum size=1mm},	not/.style={thin,circle,fill=symbols,draw=connection,fill=connection,inner sep=0pt,minimum size=0.8mm},
	>=stealth,
        }

\def\one{\mathbf{1}}
\def\eps{\varepsilon}
\newcommand{\snnorm}[1]{{\big\vert\kern-0.25ex\big\vert\kern-0.25ex\big\vert #1 
    \big\vert\kern-0.25ex\big\vert\kern-0.25ex\big\vert}}
\newcommand{\nnorm}[1]{{\Big\vert\kern-0.25ex\Big\vert\kern-0.25ex\Big\vert #1 
    \Big\vert\kern-0.25ex\Big\vert\kern-0.25ex\Big\vert}}
\newcommand{\bnnorm}[1]{{\Biggl\vert\kern-0.25ex\Biggl\vert\kern-0.25ex\Biggl\vert #1 
    \Biggl\vert\kern-0.25ex\Biggl\vert\kern-0.25ex\Biggl\vert}}
\DeclareMathAlphabet{\mathpzc}{OT1}{pzc}{m}{it}

\let\eps\varepsilon

\def\eqref#1{(\ref{#1})}

\makeatletter % Stolen from the internet to make a fat \cdot which isn't as fat as a \bullet
\newcommand*{\bigcdot}{}% Check if undefined
\DeclareRobustCommand*{\bigcdot}{%
  \mathbin{\mathpalette\bigcdot@{}}%
}
\newcommand*{\bigcdot@scalefactor}{.5}
\newcommand*{\bigcdot@widthfactor}{1.15}
\newcommand*{\bigcdot@}[2]{%
  \sbox0{$#1\vcenter{}$}% math axis
  \sbox2{$#1\cdot\m@th$}%
  \hbox to \bigcdot@widthfactor\wd2{%
    \hfil
    \raise\ht0\hbox{%
      \scalebox{\bigcdot@scalefactor}{%
        \lower\ht0\hbox{$#1\bullet\m@th$}%
      }%
    }%
    \hfil
  }%
}
\makeatother

\tcbset
{colframe=boxcolor,colback=symbols!7!pagebackground,coltext=pageforeground,
fonttitle=\bfseries,nobeforeafter,center title,size=fbox,boxsep=1.5pt,
top=0mm,bottom=0mm,boxsep=0mm,tcbox raise base}

\def\two{{\<generic>\kern0.05em\<genericb>}}
\def\twoI{{\<Ito>\kern0.05em\<Itob>}}

\def\mail#1{\burlalt{#1}{mailto:#1}}

\usepackage{thmtools}

\usepackage[margin=6em]{geometry}

\renewcommand{\leq}{\leqslant}
\renewcommand{\geq}{\geqslant}

\begin{document}

\title{On the equivalence between the Polchinski flow and the Connes-Kreimer approaches to perturbative renormalisation}

\author{Y. Bruned, P. Laubie and A. Minguella}
\institute{ 
Universite de Lorraine, CNRS, IECL, F-54000 Nancy, France
  \\
Email:\ \begin{minipage}[t]{\linewidth}
\mail{yvain.bruned@univ-lorraine.fr}
\\
\mail{paul.laubie@univ-lorraine.fr}
\\
\mail{aurelien.minguella@univ-lorraine.fr}
\end{minipage}}

\maketitle 

\begin{abstract}
	We prove a correspondence between the Polchinski flow and the Connes-Kreimer approaches to perturbative renormalisation, in the sense that the first yields the same renormalisation as the latter. More precisely, we show that an ansatz based on decorated graphs (Feynman diagrams), with a combinatorial renormalisation procedure solves the Polchinski equation. This result holds for very general Euclidean quantum field theories. We are able to derive from this the form of the renormalised potential which is, to the best of our knowledge, the most general one in the literature. The main ingredients of the proof are multiple morphism properties with respect to the renormalisation, as well as a novel duality formula for forests of decorated graphs.
\end{abstract}
\setcounter{tocdepth}{2}
\setcounter{secnumdepth}{4}
\tableofcontents
%\newpage

\section{Introduction}

We aim in the present work to unify two different approaches to perturbative renormalisation of Euclidean quantum field theories (QFT), namely the Polchinski flow approach introduced in the 80's \cite{P84}, that takes its inspiration from Wilsonian renormalisation groups ideas \cite{Wil71}, and the Connes-Kreimer approach \cite{CK1,CK2} that has an algebraic and combinatorial flavour. Even though the link between these techniques has been hinted by many physicists, there seems to be a gap in the literature that we rigorously close in the present work.
 In his book, Costello, while reviewing the various approaches to perturbative QFT makes the following remark \cite[Sec. 11.5]{Costello} page 28,
\begin{quote}
	``{\it{{Let me finally mention an approach to perturbative renormalization developed initially by Connes and Kreimer. \textnormal{[...]} In this book, however, counterterms have no intrinsic importance: they are simply a technical tools used to prove the main results. Thus, it is not clear to me if there is any relationship between Connes-Kreimer Hopf algebra and the results of this book}.}}"
\end{quote}
This book follows the Lagrangian formalism and uses renormalisation techniques similar to the Polchinski flow (the reader may for example compare \cite[Figure 3]{Costello} therein to some of the graphical examples of this paper). It is thus our belief that this link might not been known in every community, giving further interest to close this gap. With our main results, we provide a partial answer to Costello's interrogation. Let us provide some mathematical framework. We say that a QFT is the formal Gibbs measure that reads
\begin{equs}
  \dint \mu(\phi)=\frac{1}{Z}\exp(-S[\phi])\,\dint\phi\,,
\end{equs}
where $\dint \phi$ is the (non-existing) Lebesgue measure on the vector space of functions $\phi\colon\R^d\to\R^N$, and $Z$ is a renormalisation (of course, in the sense of Probability theory and not in the sense of QFT) constant that makes $\mu$ a probability measure. $S$ is an action of the form
\begin{equs}
  S[\phi]=\int_{\T^d}\frac{1}{2}\phi(m^2-\Delta)\phi+V[\phi]\,\dint x\,.
\end{equs}
$V[\phi](x)=V\big(\phi(x),\nabla\phi(x)\big)$ is a functional of $\phi$ that involves $\phi$ and its first order derivatives. Note that we chose to set the integral on the $d$-dimensional torus, which is simply a trick to not mention infinite volume problems. We provide several examples of such theories to motivate the generality of our work. We start with one of the most common theories, namely the $\phi^4$ theory. It is a scalar theory ($N=1$), whose potential takes the form
\begin{equs}
  V[\phi](x)=\alpha\phi(x)^4\,,
\end{equs}
with $\alpha$ a coupling constant. It will be the basis of most of our examples along the paper. This model finds a generalisation with the $P(\phi)$ model, with
\begin{equs}\label{eq:pphi}
  V[\phi](x)=\sum_{i=1}^{2k}\alpha_i\phi(x)^i\,,
\end{equs}
with leading coefficent $\alpha_{2k}>0$. Our framework also accomodates vector-valued theories. A common one, whose potential takes a praticularly clean form is the quartic $O(N)$ model, given by
\begin{equs}
  V[\phi](x)=\frac{1}{N}\sum_{i,j=1}^N\phi_i(x)^2\phi_j(x)^2\,.
\end{equs}
Our final example is the celebrated Yang-Mills theory, which has been the object of intensive research. Consider a Lie group $G$ with Lie algebra $\mfg$, as well as $E$, a finite-dimensional real representation of $G$. Let two functions $A$ which is $\mfg$-valued, and $\phi$ which is $E$-valued. The Yang-Mills action reads
\begin{equs}
  S[A,\phi]=\int_{\T^d} -\frac{1}{2}A\cdot\Delta_A A+\frac{1}{2}\phi\cdot(m^2-\Delta)\phi+V[A,\phi]\,\dint x\,,
\end{equs}
where $(\Delta_AA)_i=\Delta A_i+\D_i\D_jA_j$, and
\begin{equs}
  V[A,\phi]=\langle\D_iA_j-\D_jA_i,[A_i,A_j]\rangle_\mfg+\frac{1}{2}|[A_i,A_j]|_\mfg^2+2\langle\D_i\phi,A_i\phi\rangle_E+|A_i\phi|_V+\frac{1}{2}|\phi|_E^4\,.
\end{equs}
Our work also takes into account much more general theories, with an action given by
\begin{equs}
  S[\phi]=\int_{\T^d}-\phi P(\Delta)\phi+V\big(\phi,\nabla\phi,\dots,\nabla^q\phi\big)\,\dint x\,,
\end{equs}
where $P$ is a polynomial with non-negative coefficients of order $p\in\N$ and $q<2p$. However, when $q\geq2$ the theory looses some of its essential physical properties, even if it can be the object of a rigorous mathematical analysis.

It is a famous problem that such measures do not make sense as so. The first problem is the non-existence of the formal Lebesgue measure. The common way out of this is to change the reference measure. More precisely, one can rigorously construct the quadratic part of the action as an infinite-dimensional Gaussian measure, called the massive Gaussian free field, which formally reads
\begin{equs}
  \dint\gamma_{\scriptscriptstyle\mathrm{GFF}}(\phi)\propto\exp\Bigg(-\frac{1}{2}\int_{\T^d}\phi(m^2-\Delta)\phi\,\dint x\Bigg)\,\dint\phi\,.
\end{equs}
However, this measure only gives weight to distributions of negative regularity, which makes the non-linear terms in the potentials very ill-defined. To construct these theories, one has to carry out a renormalisation procedure, which can be realised with a very broad range of techniques. For the $\phi^4$ theory, one has to define a renormalised potential that depends on a regularisation parameter $\eps$
\begin{equs}
  V_0[\phi]=\int_\Lambda\big(\alpha\phi(x)^4+\mfa_\eps\phi(x)^2+\mfb_\eps\big)\,\dint x\,,
\end{equs}
where $\mfa_\eps$ and $\mfb_\eps$ are well-chosen diverging constants. Only with this procedure it is possible to obtain a non-trivial limit, which is then the definition of the associated theory. Nonetheless, in this paper, we will not talk about the rigorous construction of such measures, but what physicists call perturbative renormalisation, \ie the construction of the renormalised potential as a (generally divergent) formal series, but for which we have meaningful estimates on all the terms. We give more details about this in the coming lines. 

We present the first method tackled in this paper, the Polchinski flow. One defines a renormalised potential that depends on a scale parameter $\lambda\in[0,+\infty]$, going from small scales to large scales, 
\begin{equs}
  V_\lambda[\phi]=-\log\E_{G-G_\lambda}\big[e^{-V_0[\phi+\zeta]}\big]\,.
\end{equs}
One can then show that $V_\lambda$ satisfies the Polchinski equation
\begin{equs}\label{eq:poleqintro}
  \D_\lambda V_\lambda+\frac{1}{2}\mathrm{Tr}(\dot G_\lambda \mathrm{D}^2 V_\lambda)+\frac{1}{2}\mathrm{D} V_\lambda \dot G_\lambda \mathrm{D} V_\lambda=0\,.
\end{equs}
where $D$ is the Fréchet derivative. More precisely, $\dot G_\lambda$ is seen as a function of two variables, and in each of the two terms, the Fréchet derivatives are applied to each of these two variables. Solving this equation requires to set an ansatz, \ie guessing the form of the solution as a formal series and obtain a hierarchical system of equations for the terms of the series. The renormalisation can then be implemented inductively when enforcing a boundary condition (generally called the BPHZ condition) at the boundary $\lambda=\infty$. For the $\phi^4$ theory, this ansatz can be written under the form
\begin{equs}
  V_\lambda[\phi]=\sum_{i=0}^\infty\sum_{m=0}^{3i}\alpha^i \int_{(\T^d)^m}V_\lambda^{i,m}(x_1,\dots,x_m)\phi(x_1)\dots\phi(x_m)\,\dint x_1\dots \dint x_m\,.
\end{equs}
In that case, the simplicity of the potential allows a simple combinatorial description of the terms (here a couple). For this example, it is known that the series actually diverges. We show how this method is implemented on the $\phi^4$ theory in Section \ref{sec:analytic}. 

The Connes-Kreimer approach, also sometimes referred to as the BPHZ approach, after the seminal series of papers \cite{BP57,KH69,WZ69} takes antoher point of view and bypasses the inductive aspect. Instead, one directly sets an ansatz
\begin{equs}
  V_\lambda[\phi]=\sum_{\Gamma\in F^\star}\beta(\Gamma)(\hat\Pi_\lambda\Gamma)[\phi]\,,
\end{equs}
where the sum runs on a set of decorated graphs, called Feynman diagrams. These diagrams are interated integrals that involve a generally large number of propagators, which are, in a more mathematical language, the Green functions of the linear operator $m^2-\Delta$, or another operator depending on the theory. $\beta(\Gamma)$ is a real factor that we will make explicit later. For a given Feynman diagram $\Gamma$, $\hat{\Pi}_\lambda \Gamma$ is a renormalised amplitude that depends on the scale parameter $\lambda$. The renormalisation takes the form
\begin{equs}
  \hat\Pi_\lambda\Gamma=(g_{\scriptscriptstyle\mathrm{BPHZ}}\otimes\Pi_\lambda)\Delta\Gamma\,.
\end{equs}
where $ \Delta $ is a Connes-Kreimer type coproduct that extracts subdivergent diagrams that are vacuum diagrams, \ie Feynman diagrams without legs. The map $\Pi_\lambda$ gives the amplitude of the Feynman diagrams without renormalisation. The map $g_{\scriptscriptstyle\mathrm{BPHZ}}$ is a character associating to each vaccum diagram a renormalisation constant. The BPHZ choice of renormalisation is such that
\begin{equs}\label{eq:bphzchoice}
  \hat{\Pi}_{\infty} (\Gamma, \mathfrak{n})=0\,,
\end{equs}
for every divergent vacuum diagram $(\Gamma,\mfn)$ with vertex decoration $\mathfrak{n}\colon\mathcal{\mathcal{V_\star}} \rightarrow \N^{d} $. Here, $\mathcal{V}_\star$ is the set of internal vertices of $ \Gamma $ and this decoration corresponds to monomials. By divergent, we mean that the integral that $(\Gamma,\mfn)$ represents is a diverging quantity. This is given at the combinatorial level by the non-positivity of a notion of degree of this graph. The celebrated BPHZ theorem states that $\hat\Pi_\infty\Gamma$ is a well-defined distribution for every $\Gamma$. See for instance Theorem 3.1 in \cite{BPHZ_theorem}.

\subsection{Main results}

We are now ready to state the main results of this paper.

\begin{theorem}\label{maintheorem}
The ansatz
  \begin{equation}
    V_\lambda[\phi]=\sum_{\Gamma\in F^\star}\frac{\alpha(\Gamma)\Gamma!}{S(\Gamma)}(\hat\Pi_\lambda\Gamma)[\phi]\,,
  \end{equation}
  with a BPHZ renormalisation procedure formally solves the Polchinski equation \eqref{eq:poleqintro} with BPHZ boundary conditions at $\lambda=\infty$.
\end{theorem}
Here, $F^\star$ is a set of Feynman diagrams given by a notion of rule that depends on the form of the potential $V$. $\alpha(\Gamma)$ is a real factor, that is a generalisation of the factors appearing in the potential. This notation matches the coefficients of the potential \eqref{eq:pphi} for clarity. $\Gamma!$ and $S(\Gamma)$ are combinatorial factors. All these notions are defined in Section \ref{sec:addnotions}. From this theorem, one gets an interesting subsequent result, namely an explicit form of the renormalised potential.

\begin{corollary}\label{corollary:renormmeasure}
  The renormalised potential is given by
  \begin{equation}\label{eq:ansatzpotential}
    V_0[\phi]=\sum_{\deg(\Gamma)+|\mfn|\leq0}\frac{\alpha(\Gamma)\Gamma!g_{\scriptscriptstyle\mathrm{BPHZ}}\big(\mathrm{vac}(\Gamma),\mfn\big)}{S(\Gamma,\mfn)}\sum_{\pi\ell=\mfn}\frac{\mfn!}{\ell!}\big(\Pi_0\mathrm{res}_\ell(\Gamma)\big)[\phi]\,.
  \end{equation}
\end{corollary}
Note that in the formula above, the non-positive diagrams also contain the elementary diagrams, that do not require a renormalisation procedure, since they do not contain any subdivergence. If one writes the unrenormalised potential under the form
\begin{equs}
  V[\phi]=\sum_{\Gamma\in E^\star}\alpha(\Gamma)(\Pi_0\Gamma)[\phi]\,,
\end{equs}
where $E^\star$ is the set of elementary diagrams, \ie the diagrams that represent the monomials in $\phi$ and its derivatives appearing in the definition of $V$. One could also rewrite,
\begin{equs}
  V_0[\phi]=V[\phi]+\sum_{\substack{\deg(\Gamma)+|\mfn|\leq0\\\Gamma\notin E^\star}}\frac{\alpha(\Gamma)\Gamma!g_{\scriptscriptstyle\mathrm{BPHZ}}\big(\mathrm{vac}(\Gamma),\mfn\big)}{S(\Gamma,\mfn)}\sum_{\pi\ell=\mfn}\frac{\mfn!}{\ell!}\big(\Pi_0\mathrm{res}_\ell(\Gamma)\big)[\phi]\,.
\end{equs}
This corollary is a generalisation of the main result of \cite{BH25}, since it allows Feynman diagrams with legs. In addition, it does not require to use multi-indices to get an expression of the renormalised potential, and gives a more precise combinatorial description (namely graphs). It is to the best of our knowledge the most general expression of the renormalised potential in the literature.

As a matter of example, the negative diagrams of the $\phi^4$ theory are the quartic diagrams with $0$ or $2$ legs. Moreover, for the subcritical theory, their degree is always greater than $-2$, and the terms involving the first order of the Taylor expansion disappear by invariance by translation. The renormalised potential thus reads
\begin{equs}
  V_0[\phi]=\int_{\T^d}\big(\alpha\phi(x)^4+\mfa_\eps\phi(x)^2+\mfb_\eps\big)\,\dint x\,,
\end{equs}
with renormalisation constants
\begin{equs}
  \mfa_\eps=\sum_i \alpha^i \mfa_\eps^i\,,\mathrm{and}\quad\mfb_\eps=\sum_i\alpha^i \mfb_\eps^i\,,
\end{equs}
where $\mfa_\eps^i$ is the sum over all the negative quartic diagrams $\Gamma$ with $i$ nodes and $2$ legs of $\Gamma!g_{\scriptscriptstyle\mathrm{BPHZ}}\big(\mathrm{vac}(\Gamma)\big)/S(\Gamma)$, and $\mfb_\eps^i$ the same quantity with $2$ legs replaced by $0$ legs.

\subsection{Literature}

The Polchinski flow approach has been the subject of intensive research in the past decades. For a more extensive overview, we recommand the monographs \cite{Kop07,M03} that provide a more complete list of references. Some of the best successes is the proof of the perturbative renormalisability of the $\phi_4^4$ and of the Yang-Mills theories. In a similar way, renormalisation group techniques also have been used to construct quantum fields non-perturbatively, for which one of the main diffuculties is to get a proper notion of convergence for the perturbative expensions. We refer the reader to the book \cite{Riv91}, that gives a complete overview of the question. We would also like to mention more recent works involving a rigorous analysis of the Polchinski flow. The papers \cite{GM24,Mey26} construct the sine-Gordon measure via a forward-backward SDE. There is also a series of works that uses the Polchinski flow to show log-Sobolev inequalities for quantum fields and systems of spins in statistical mechanics. See \cite{BB21,BBD24} and the references therein.

On the other side, the algebraic approach to perturbative renormalisation finds its roots in \cite{CK1,CK2} where  Hopf algebras were proposed to encode the BPHZ algorithm described in the series of papers \cite{BP57,KH69,WZ69}. We based most of the framework of the present paper on the proceeding \cite{BPHZ_theorem}, that gives a rigorous and complete statement and proof of the BPHZ theorem. We also refer to \cite{FMRV85}.

The Hopf algebraic approach of perturbative renormalisation of Connes and Kreimer \cite{CK1,CK2,VS2007} inspired the Hopf algebras at play in the context of singular stochastic partial differential equations (SPDEs).  Indeed, the work \cite{BHZ} lays down the algebraic foundations of the theory of regularity structures invented by Hairer \cite{reg} with two Hopf algebras in co-interaction, one for the recentering of stochastic iterated integrals and the other for the BPHZ renormalisation in the spirit of the one for Feynman diagrams. Then, \cite{CH16} proves the convergence of these renormalised iterated integrals and \cite{BCCH} provides the renormalised equation for a large class of singular SPDEs. Recently, Duch introduced a novel solution theory based on the Polchinski flow \cite{Duc21,Duc22} which, thanks to the analysis of a flow equation that is inspired from the Polchinski equation, manages to give an inductive construction of the stochastic terms, thus bypassing the tedious analysis of \cite{CH16}. The author shows that the cumulants of the iterated stochastic integrals satisfy a flow equation that is in every point similar to the one satisfied by the coefficients of the perturbative expension of the renormalised potential. Let us mention that a discrete version was investigated in previous works \cite{K16,KM17}. In \cite{CF24a}, the authors extend the approach to non-polynomial linearities by considering the generalised KPZ equation. In \cite{BM25}, one gets a general ansatz to Duch's flow approach and the equivalence with the BPHZ renormalisation defined in \cite{BHZ}. One of the goals of the present work is to repeat this equivalence directly at the level of perturbative renormalisation using the same ideas developed in the paper. One has a correspondence between the main results of the two papers. Indeed, Theorems 5.3 and 5.5 in \cite{BM25} correspond respectively to Theorem \ref{maintheorem} and Corollary \ref{corollary:renormmeasure}. Theorem 5.15 in \cite{BM25} that shows that the renormalisation used in the flow approach correspond to the BHZ character introduced in \cite{BHZ}. One has a similar statement in Proposition \ref{prop:gbphz}. 

\subsection*{Acknowledgments}

{\small
   Y.B., P.L., and A.M. acknowledge funding support from the European Research Council (ERC) through the ERC Starting Grant Low Regularity Dynamics via Decorated Trees (LoRDeT), grant agreement No.\ 101075208 is acknowledged. Views and opinions expressed are however those of the author(s) only and do not necessarily reflect those of the European Union or the European Research Council. Neither the European Union nor the granting authority can be held responsible for them. 
} 

\section{Setting}

\subsection{Euclidean quantum fields and the Polchinski flow}
Let $C$ a covariance operator on $L^2(\Lambda)$. We denote by $\gamma_C$ the mean-zero Gaussian measure with covariance $C$, and $\E_C$ its expectation. More precisely it is characterized by the fact that, for every $f,g\in L^2(\Lambda)$,
\begin{equation}\label{eq:covariance}
  \E_C\big[\langle\zeta,f\rangle_{L^2(\Lambda)}\langle\zeta,g\rangle_{L^2(\Lambda)}\big]=\langle f,Cg\rangle_{L^2(\Lambda)}\,,
\end{equation}
where the expectation is on $\zeta$. We have let for conciseness $\Lambda=\T^d$. We will consider for the rest of this work a propagator $G(x,y)=G(x-y)$ where $G$ satisfies for every $k\in\N^d$ and $x\in\Lambda$,
\begin{equation}
  |\D^k G(x)|\lesssim|x|^{\sigma-d-|k|}\,.
\end{equation}
This can be realised for instance as the Green function of the fractional operator $(m^2-\Delta)^{\sigma/2}$. We then would like to consider as a reference Gaussian fields whose covariance is given by the kernel operator of kernel $G$. It is a common fact that such a random variable does not belong to $L^2(\Lambda)$, but rather to a space of distributions because of the divergence of $G$ in $0$. In fact, the expression \eqref{eq:covariance} can be made sense of by putting $\mcS'(\Lambda)\times\mcS(\Lambda)$ duality brackets on the left-hand side. To make everything rigorous, we need to introduce a mollification. Let $\rho$ a smooth function compactly supported on the unit ball. We define $\rho_\eps(x)=\eps^{-d}\rho(\eps^{-1}x)$. We may then consider a mollification of the (fractional) GFF by taking its push-forward by the convolution with $\rho_\eps$. This yields a mollified covariance kernel $G^\eps$ given by
\begin{equation}
  G^\eps(x-y)=\iint_{\Lambda^2} \rho_\eps(x-z)G(z-w)\rho_\eps(w-y)\,\dint z\,\dint w\,.
\end{equation}
We will drop the dependence in $\eps$ in the rest of the paper for conciseness. With this at hand, we need to introduce a covariance decomposition 
\begin{equation}
  G_\lambda=\int_\lambda^\infty \dot G_s\,\dint s\,,
\end{equation}
and such that $G=G_0$. Such a choice of decomposition is not unique, and many choices can be found in the literature. Some of them are, for the classical massive Laplace operator,
\begin{itemize}
  \item Heat kernel decomposition
    \begin{equation}
      G_\lambda(x)=\frac{1}{(4\pi)^{d/2}}\int_{\lambda}^\infty s^{-d/2}e^{-m^2s}e^{-\frac{|x|^2}{4s}}\dint s\,,
    \end{equation}
    %which reads in Fourier
    %\begin{equation}
    %  G_{t,T}(x)=\frac{1}{(2\pi)^d}\sum_{k\in\Z^d}\frac{e^{\i k\cdot x}}{|k|^2+m^2}\left(e^{-\frac{|k|^2+m^2}{\lambda^2}}-e^{-\frac{|k|^2+m^2}{\lambda^2}}\right)\,.
    %\end{equation}
  %The parameter $\lambda$ here corresponds to the small scale cutoff $\eps$.
  \item Pauli-Villars decomposition
    \begin{equation}
      G_\lambda(x)=(-\Delta+m^2)^{-1}-(-\Delta+m^2+1/\lambda)^{-1}\,.
    \end{equation}
  \item Compactly supported decomposition
    \begin{equation}
      G_\lambda(x)=\chi(|x|^2/\lambda)G(x)\,.
    \end{equation}
  where $\chi$ is a smooth function with good support properties, see \eqref{eq:kerneldecomposition}.
\end{itemize}
Note that we will carry out the analysis in physical space, whereas physicists generally use the Fourier space. Let $V_0$ a (renormalised) potential. Our goal is to study the perturbative renormalisability of the QFT given by the measure $\mu$ defined by
\begin{equation}
  \mathbb{E}_{\mu}[F]\propto\E_{G}\big[e^{-V_0[\zeta]}F(\zeta)\big]\,,
\end{equation}
for every bounded functional $F$.

\begin{definition}
   Let $\phi\colon\Lambda\to\R$, we define the renormalised potential by
  \begin{equation}
    V_\lambda[\phi]=-\log\E_{G-G_\lambda}\big[e^{-V_0[\phi+\zeta]}\big]\,,
  \end{equation}
\end{definition}
From this definition, one can show that this potential satifies an infinite-dimensional PDE. A proof of this statement can be found in \cite[Proposition 3.4]{BBD24}.
\begin{proposition}
The potential  $V_\lambda$ satisfies the Polchinski equation
  \begin{equation}\label{eq:poleq}
    \D_\lambda V_\lambda+\frac{1}{2}\mathrm{Tr}(\dot G_\lambda \mathrm{D}^2 V_\lambda)+\frac{1}{2}\mathrm{D} V_\lambda \dot G_\lambda \mathrm{D} V_\lambda=0\,.
  \end{equation}
\end{proposition}
The definition of $V_\lambda$ seems a bit tautological at the moment since we actually want to guess the form of $V_0$ that should be $V_0=V+\mathrm{counterterms}$. We will actually show later that we will impose some boundary conditions at $\lambda=\infty$ for the Polchinski equation, namely the BPHZ choice of renormalisation \eqref{eq:bphzchoice}. In the physics literature, $V_0$ is generally given first, and it is then shown that it is the right one. This is what is called a top-down approach. Here, we do the opposite, \ie we want to find $V_0$ back, wihch is a bottom-up approach.

\subsection{The Connes-Kreimer approach to perturbative renormalisation}

The framework of this section is freely borrowed to \cite{BPHZ_theorem}, and modified for our purposes.

\begin{definition}
  \begin{enumerate}
    \item A \textit{Feynman diagram} is a finite directed connected multigraph (\ie multiple edges are authorized) $\Gamma = (\CV,\CE)$ endowed with the following additional data:
    \begin{claim}
      \item An ordered set of distinct vertices $\CL = \{[1],\ldots,[k]\} \subset \CV$ such that each vertex in $\CL$ has exactly one outgoing edge (called a ``leg'') and no incoming edge, and such that each connected component of $\Gamma$ contains at least one leg. We will frequently use the notation $\CV_\star = \CV \setminus \CL$, as well as $\CE_\star \subset \CE$ for the edges that are not legs.
      \item A decoration $\mfe\colon \CE \to \N^d\times\N^d$ of the edges of $\Gamma$, that we denote $\mfe=(\mfe_-,\mfe_+)$.
      \item A decoration $\mft\colon\CL\to\{1,\dots,N\}$.
      \item A decoration $\mfd\colon\CE_\star\to \{0,1\}$.
    \end{claim}
    We denote by $F$ this set.
    \item A \textit{forest (of Feynman diagrams)} is defined similarly but without the connectedness assumption. It is in particular a collection of Feynman diagrams in which the order does not matter. This set is endowed with a multiplication $\bullet$ that yields the forest composed of the two forests on both sides of the operation. Every forest $\mathbf{\Gamma}$ can be written in a canonical form $\mathbf{\Gamma}=\Gamma_1\dots\Gamma_n$ for some integer $n$ and Feynman diagrams $\Gamma_i$. We denote by $\mfF$ this set.
  \end{enumerate}
\end{definition}
In addition, we define $F_0$ and $F_1$ the subsets of $F$ of Feynman diagrams containing respectively $0$ and $1$ internal edge with a derivative in $\lambda$. We denote $\langle F_0\rangle$ and $\langle F_1\rangle$ their linear spans. Every Feynman diagram represents an analytical quantity given by the following $\lambda$-dependent translation map.

\begin{equation}\label{eq:evmap}
  (\Pi_\lambda\Gamma)[\phi]=\int_{\Lambda^{\CV_\star}}\prod_{e\in\CE_\star} \D_\lambda^{\mfd(e)}\D^{\mfe_+(e)}_{x_{e_+}}\D^{\mfe_-(e)}_{x_{e_-}}(G-G_\lambda)(x_{e_+},x_{e_-})\D^{l_1}\phi_{\mft(v_1)}(x_{v_1})\dots\D^{l_k}\phi_{\mft(v_k)}(x_{v_k})\,\dint x\,.
\end{equation}
This map is extended to forests by imposing a morphism property with respect to $\bullet$. In our notations, $e_-$ and $e_+$ are respectively the source and the target of an edge $e$ ($e_-\to e_+$). Since the legs have only one outgoing edge, we choose arbitrarily that only the decoration $\mfe_+$ matters. These are the numbers $l_i$ in the formula above. In \eqref{eq:evmap}, we have implicitly supposed that the kernels associated to the edges linking legs to internal vertices ($e$ say) are the $\mfe_+(e)$ derivatives of a Dirac $\delta$. We also define a notion of half-edge. A half-edge is the data of a couple $(v,e)$, with $e\in\mcE$ and $v\in\{e_-,e_+\}$. For $e\in\mcE$, we denote $e_\leftarrow=(e,e_-)$ and $e_\rightarrow=(e,e_+)$. 

Since we supposed our linear operator invariant by translation, the propagator takes the form $G(x,y)=G(x-y)$. We thus have the simplification $\D^{\mfe_+(e)}_{x_{e_+}}\D^{\mfe_-(e)}_{x_{e_-}}(G-G_\lambda)(x_{e_+},x_{e_-})=(-1)^{\mfe_-(e)}\D^{\mfe_+(e)+\mfe_-(e)}(G-G_\lambda)(x_{e_+}-x_{e_-})$. Contrarily to \cite{BPHZ_theorem}, we did not implement this simplification directly on the combinatorial side. The reason for that is that we will need to cut edges in Section \ref{sec:forestduality} and it is thus much simpler to keep it as it is. Moreover, it allows us to not write powers of $-1$ everywhere. The decorations $\mft$ and $\mfd$ are also absent in \cite{BPHZ_theorem}. The first one appears because we treat vector valued theories, which is a generalisation that brings no conceptual difficulty and the second one is proper to the Polchinski flow as it represents a derivative in $\lambda$.

\begin{example}\label{ex:feynmandiagram}
We give below an example of Feynman diagram.
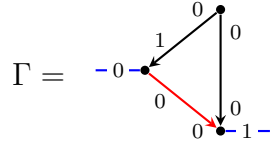
\begin{equs}
\Gamma=\begin{tikzpicture}[scale=0.2,baseline=0.6cm]
			\node at (0,0)  [dot,label= {[label distance=-0.2em]below: \scriptsize  $      $} ] (root) {};
			\node at (-5,4)  [dot,label= {[label distance=-0.2em]right: \scriptsize  $     $} ] (center) {};
			\node at (0,8)  [dot,label= {[label distance=-0.2em]right: \scriptsize  $     $} ] (centerc) {};
			\node at (-9,4)  [,label= {[label distance=-0.2em]right: \scriptsize  $     $} ] (left) {};
			\node at (4,0)  [,label= {[label distance=-0.2em]right: \scriptsize  $     $} ] (right) {};
			\draw[arrow,color=red] (center) to
			node [sloped,below] {\small }     (root);
			\draw[arrow] (centerc) to
			node [sloped,below] {\small }     (center);
			\draw[arrow] (centerc) to
			node [sloped,below] {\small }     (root);
			\draw[kernel1,color=blue] (center) to
			node [sloped,below] {\small }     (left);
			\draw[kernel1,color=blue] (root) to
			node [sloped,below] {\small }     (right);
			\node at (1,1.5) [label={[label distance=0em]center: \scriptsize  $0$} ] () {};
      \node at (1,6.5) [label={[label distance=0em]center: \scriptsize  $0$} ] () {};
      \node at (-1.5,8) [label={[label distance=0em]center: \scriptsize  $0$} ] () {};
			\node at (-1.5,0) [label={[label distance=0em]center: \scriptsize  $0$} ] () {};
      \node at (-4,2) [label={[label distance=0em]center: \scriptsize  $0$} ] () {};
			\node at (-4,6) [label={[label distance=0em]center: \scriptsize  $1$} ] () {};
			\node at (-6.75,4) [fill=white,label={[label distance=0em]center: \scriptsize  $0$} ] () {};
			\node at (1.75,0) [fill=white,label={[label distance=0em]center: \scriptsize  $1$} ] () {};
\end{tikzpicture}
\end{equs}
We represented the legs in blue. The interior edges having decoration $\mfd=1$ are drawn in red, where the remaining ones are in black. For the examples, we will consider scalar theories in dimension $1$, unless specified otherwise. This will lighten considerably the notations. Moreover, from now on, we will not write the null decorations. This diagram represents the analytic quantity
\begin{equs}
  (\Pi_\lambda\Gamma)[\phi]=\int_{\Lambda^3}\D_{x_1}(G-G_\lambda)(x_1,x_2)(G-G_\lambda)(x_2,x_3)(-\dot G_\lambda)(x_3,x_1)\phi(x_1)\D\phi(x_3)\,\dint x_1\dint x_2\dint x_3\,.
\end{equs}
For a much simpler example, we have
\begin{equs}
  \Big(\Pi_\lambda\begin{tikzpicture}[scale=0.15,baseline=-0.1cm, trim right=0.4cm, trim left=-0.4cm]
			\node at (0,0)  [dot,label= {[label distance=-0.2em]below: \scriptsize  $      $} ] (center) {};
			\node at (-3,3)  [label= {[label distance=-0.2em]right: \scriptsize  $     $} ] (topleft) {};
			\node at (-3,-3)  [label= {[label distance=-0.2em]right: \scriptsize  $     $} ] (botleft) {};
			\node at (3,3)  [label= {[label distance=-0.2em]right: \scriptsize  $     $} ] (topright) {};
      \node at (3,-3)  [label= {[label distance=-0.2em]right: \scriptsize  $     $} ] (botright) {};
			\draw[kernel1,color=blue] (center) to
			node [sloped,below] {\small }     (topright);
      \draw[kernel1,color=blue] (center) to
			node [sloped,below] {\small }     (topleft);
      \draw[kernel1,color=blue] (center) to
			node [sloped,below] {\small }     (botright);
      \draw[kernel1,color=blue] (center) to
			node [sloped,below] {\small }     (botleft);
    \end{tikzpicture}\Big)[\phi]=\int_\Lambda\phi(x)^4\,\dint x\,.
\end{equs}
Note that for such pure legs diagrams, the paramter $\lambda$ does not matter.
\end{example}

\begin{remark}
  Since we are in the translation invariant case, one could observe that the self-loops simplify, in the sense that, on a precise example
  \begin{equs}
    \Pi_\lambda\begin{tikzpicture}[scale=0.15,baseline=-0.1cm,trim right=1.1cm, trim left=-0.6cm]
			\node at (0,0)  [dot,label= {[label distance=-0.2em]below: \scriptsize  $      $} ] (center) {};
			\node at (-4,0)  [label= {[label distance=-0.2em]right: \scriptsize  $     $} ] (topleft) {};
			\node at (0,4)  [label= {[label distance=-0.2em]right: \scriptsize  $     $} ] (botleft) {};
			\node at (0,-4)  [label= {[label distance=-0.2em]right: \scriptsize  $     $} ] (topright) {};
      \node at (4,0)  [dot,label= {[label distance=-0.2em]right: \scriptsize  $     $} ] (centerright) {};
      \node at (4,-4)  [label= {[label distance=-0.2em]right: \scriptsize  $     $} ] (leg2) {};
      \draw[arrow] (center) to
			node [sloped,below] {\small }     (centerright);
      \draw[arrow, color=red] (centerright) .. controls (4,5) and (9,0) .. (centerright);
			\draw[kernel1,color=blue] (center) to
			node [sloped,below] {\small }     (topright);
      \draw[kernel1,color=blue] (center) to
			node [sloped,below] {\small }     (topleft);
      \draw[kernel1,color=blue] (center) to
			node [sloped,below] {\small }     (botleft);
      \draw[kernel1,color=blue] (centerright) to
			node [sloped,below] {\small }     (leg2);
    \end{tikzpicture}=-\dot G_\lambda(0)\Pi_\lambda\begin{tikzpicture}[scale=0.15,baseline=-0.1cm,trim right=1.1cm, trim left=-0.6cm]
			\node at (0,0)  [dot,label= {[label distance=-0.2em]below: \scriptsize  $      $} ] (center) {};
			\node at (-4,0)  [label= {[label distance=-0.2em]right: \scriptsize  $     $} ] (topleft) {};
			\node at (0,4)  [label= {[label distance=-0.2em]right: \scriptsize  $     $} ] (botleft) {};
			\node at (0,-4)  [label= {[label distance=-0.2em]right: \scriptsize  $     $} ] (topright) {};
      \node at (4,0)  [dot,label= {[label distance=-0.2em]right: \scriptsize  $     $} ] (centerright) {};
      \node at (8,0)  [label= {[label distance=-0.2em]right: \scriptsize  $     $} ] (leg2) {};
      \draw[arrow] (center) to
			node [sloped,below] {\small }     (centerright);
			\draw[kernel1,color=blue] (center) to
			node [sloped,below] {\small }     (topright);
      \draw[kernel1,color=blue] (center) to
			node [sloped,below] {\small }     (topleft);
      \draw[kernel1,color=blue] (center) to
			node [sloped,below] {\small }     (botleft);
      \draw[kernel1,color=blue] (centerright) to
			node [sloped,below] {\small }     (leg2);
    \end{tikzpicture}\,.
  \end{equs}
  However, implementing this simplification on the combinatorial side would break the duality formula of Lemma \ref{lemma:duality}.
\end{remark}

\begin{definition}
  A \textit{vacuum diagram} is a Feynman diagram $\Gamma\in F$ whose leg set is empty and embedded with an additional decoration $\mfn\colon\mcV_\star\to\N^d$, and with one distinguished vertex $v_\star$.
\end{definition}
As for normal Feynman diagrams, the vacuum diagrams represent an analytic quantity given by
\begin{equation}
  \Pi_\lambda(\Gamma,v_\star,\mfn)=\int_{\Lambda^{\mcV_\star\backslash \{v_\star\}}}\prod_{e\in\mcE_\star}\D_\lambda^{\mfd(e)}\D^{\mfe_+(e)}_{x_{e_+}}\D^{\mfe_-(e)}_{x_{e_-}}(G-G_\lambda)(x_{e_+},x_{e_-})\prod_{v\in\mcV_\star}(x_v-x_{v_\star})^{\mfn(v)}\,\dint x\,.
\end{equation}
Note that this time the analytic quantity is not field-dependent.
\begin{example}\label{ex:vacuum}
  We give a simple example, for which we have represented the distinguished vertex in teal.
  \begin{equs}
    \Pi_\lambda\begin{tikzpicture}[scale=0.2,baseline=-0.1cm,trim right=-0.2cm]
      \node at (-2.5,0)  [dot,label= {[label distance=-0.2em]right: \scriptsize  $k$} ] (centerleft) {};
      \node at (-5,2)  [dot,label= {[label distance=-0.2em]left: \scriptsize  $j$} ] (topleft) {};
			\node at (-5,-2)  [dot,color=teal,label= {[label distance=-0.2em]left: \scriptsize  $i$} ] (botleft) {};
      \draw[arrow] (centerleft) to
			node [sloped,below] {\small }     (topleft);
      \draw[arrow] (centerleft) to
			node [sloped,below] {\small }     (botleft);
      \draw[arrow] (topleft) to
			node [sloped,below] {\small }     (botleft);
    \end{tikzpicture}&=\one_{i=0}\int_{\Lambda^2}G(x_1-x_2)G(x_1-x_3)G(x_2-x_3)(x_2-x_1)^j(x_3-x_1)^k\,\dint x_2\,\dint x_3\\
    &=\one_{i=0}\int_{\Lambda^2}G(x_2)G(x_3)G(x_2-x_3)x_2^jx_3^k\,\dint x_2\,\dint x_3\,.
  \end{equs}
  We supposed for this example that $G$ is symmetric. It will be important later to notice that, when $i=j=k=0$, then
  \begin{equs}
    \Pi_\lambda\begin{tikzpicture}[scale=0.2,baseline=-0.1cm,trim right=-0.4cm]
      \node at (-2.5,0)  [dot] (centerleft) {};
      \node at (-5,2)  [dot] (topleft) {};
			\node at (-5,-2)  [dot,color=teal] (botleft) {};
      \draw[arrow] (centerleft) to
			node [sloped,below] {\small }     (topleft);
      \draw[arrow] (centerleft) to
			node [sloped,below] {\small }     (botleft);
      \draw[arrow] (topleft) to
			node [sloped,below] {\small }     (botleft);
    \end{tikzpicture}=\Pi_\lambda\begin{tikzpicture}[scale=0.2,baseline=-0.1cm,trim right=-0.4cm]
      \node at (-2.5,0)  [dot] (centerleft) {};
      \node at (-5,2)  [dot,color=teal] (topleft) {};
			\node at (-5,-2)  [dot] (botleft) {};
      \draw[arrow] (centerleft) to
			node [sloped,below] {\small }     (topleft);
      \draw[arrow] (centerleft) to
			node [sloped,below] {\small }     (botleft);
      \draw[arrow] (topleft) to
			node [sloped,below] {\small }     (botleft);
    \end{tikzpicture}=\Pi_\lambda\begin{tikzpicture}[scale=0.2,baseline=-0.1cm,trim right=-0.4cm]
      \node at (-2.5,0)  [dot,color=teal] (centerleft) {};
      \node at (-5,2)  [dot] (topleft) {};
			\node at (-5,-2)  [dot] (botleft) {};
      \draw[arrow] (centerleft) to
			node [sloped,below] {\small }     (topleft);
      \draw[arrow] (centerleft) to
			node [sloped,below] {\small }     (botleft);
      \draw[arrow] (topleft) to
			node [sloped,below] {\small }     (botleft);
    \end{tikzpicture}\,.
  \end{equs}
  For vacuum diagrams with node decoration set to $0$, the base point does not matter. This is due to the translation invariance of the propagator that allows to perform a change of variable and therefore to move the base point.
\end{example}

\begin{remark}
  In this paper, we choose to not use the equivalence class introduced in \cite{BPHZ_theorem} that allows to identify Feynman diagrams that can be linked to each other by integration by parts or by change of base-point for the vacuum diagrams. However, we believe that all the operations introduced in the rest of the paper are compatible with such a framework. 
\end{remark}

\begin{definition}\label{def:degree}
	We define the degree of a Feynman diagram $\Gamma$ to be
	\begin{equation}
		\deg(\Gamma,\mfe,\mfd)=\sum_{e\in\mcE_\star}\big(\sigma-d-|\mfe(e)|-\mfd(e)\big)+d(|\mcV|-1)\,,
	\end{equation}
  where $|\mfe(e)|=\mfe_+(e)+\mfe_-(e)$. We also extend this notion to vacuum diagrams with
  \begin{equation}
    \deg(\Gamma,\mfe,\mfd,\mfn)=\deg(\Gamma,\mfe,\mfd)+\sum_{v\in\mcV_\star}|\mfn(v)|\,.
  \end{equation}
\end{definition}
We denote $F_-$ the set of vacuum diagrams of non-positive degree, and $\mfF_-$ the set of forests of such graphs.

Given a Feynman diagram $\Gamma$, we can extract subgraphs from it, formed from a subset $\overline{\mcE}\subset\mcE_\star$ of its edges with their neighbouring vertices, as well as the decoration restricted to this subset. We impose that the multi-edges between two nodes are always extracted all at once. In particular, this subgraph is a vacuum diagram with null node decoration, so that we do not need to specify the base point, as we mentionned it in Example \ref{ex:feynmandiagram}. For such a subgraph $\overline{\Gamma}\subset\Gamma$, we define $\D \overline{\Gamma}$ for the set of all half-edges adjacent to $\overline{\Gamma}$. Legs can also be part of $\D \overline{\Gamma}$, but they can not be part of $\overline{\Gamma}$. Note also that the edge joining the two vertices at the top appears as two distinct half-edges in $\D \overline{\Gamma}$. Given furthermore a map $\ell \colon\D \overline{\Gamma} \to \N^d$ (canonically extended to vanish on all other half-edges of $\Gamma$), we then define the following two objects.
\begin{enumerate}
  \item A vacuum diagram $(\overline{\Gamma},\pi\ell)$ which consists of the graph $\overline{\Gamma}$ endowed with the edge decoration inherited from $\Gamma$, as well as the node decoration $\mfn=\pi\ell$ given by $(\pi\ell)(v) =\sum_{e\colon(e,v)\in\D\overline{\Gamma}}\ell(e,v)$.
  \item A Feynman diagram $\Gamma/(\overline{\Gamma}, \ell)$ obtained by contracting the connected components of $\overline{\Gamma}$ to nodes and applying $\ell$ to the resulting diagram in the sense that, for edges $e \in \mcE \backslash \overline{\mcE}$ adjacent to $\D \overline{\Gamma}$, their new decoration is given by $\big(\mfe_+(e)+\ell(e_\rightarrow),\mfe_-(e)+\ell(e_\leftarrow)\big)$.
\end{enumerate}
\begin{example}
  We provide an example to make these notions clear (it is freely borrowed to \cite{BPHZ_theorem}). Let
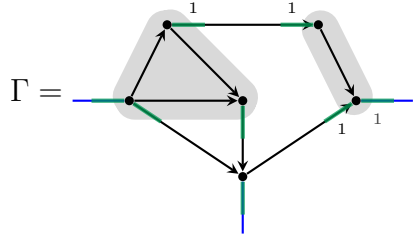
  \begin{equs}
    \Gamma=\begin{tikzpicture}[scale=1,baseline=-0.0cm]
      \node[dot] (l) at (0,0) {};
      \node[dot] (r) at (3,0) {};
      \node[dot] (ul) at (0.5,1) {};
      \node[dot] (ur) at (2.5,1) {};
      \node[dot] (d) at (1.5,-1) {};
      \node[dot] (c) at (1.5,0) {};
      \draw[arrow] (l) -- (ul);
      \draw[arrow] (ul) -- (ur) node[pos=0.15,above=-0.0] {$\scriptscriptstyle1$} node[pos=0.85,above=-0.0] {$\scriptscriptstyle1$};
      \draw[arrow] (ur) -- (r); 
      \draw[arrow] (l) -- (c);
      \draw[arrow] (ul) -- (c);
      \draw[arrow] (l) -- (d);
      \draw[arrow] (c) -- (d);
      \draw[arrow] (d) -- (r) node[pos=0.8,below right=-0.1] {$\scriptscriptstyle1$};
      \draw[thick,blue] (l) -- ++(180:0.75);
      \draw[thick,blue] (r) -- ++(0:0.75);
      \draw[thick,blue] (d) -- ++(-90:0.75);
      \draw[line width=0.5cm,draw opacity=0.15, line cap=round] (l) -- (ul) -- (c) -- (l) -- (0.75,0.25) -- (0.5,0.5);
      \draw[line width=0.5cm,draw opacity=0.15, line cap=round] (ur) -- (r);
      \draw[ultra thick,ForestGreen,nearly opaque] (l) -- ++(180:0.5);
      \draw[ultra thick,ForestGreen,nearly opaque] (r) -- ++(0:0.5) node[pos=0.5,below=-0.0,black] {$\scriptscriptstyle1$};
      \draw[ultra thick,ForestGreen,nearly opaque] (d) -- ++(-90:0.5);
      \draw[ultra thick,ForestGreen,nearly opaque] (ul) -- ++(0:0.5);
      \draw[ultra thick,ForestGreen,nearly opaque] (ur) -- ++(180:0.5);
      \draw[ultra thick,ForestGreen,nearly opaque] (c) -- ++(-90:0.5);
      \draw[ultra thick,ForestGreen,nearly opaque] (l) -- ++(-33.4:0.5);
      \draw[ultra thick,ForestGreen,nearly opaque] (r) -- ++(-146.6:0.5);
    \end{tikzpicture} 
  \end{equs}
  We represented over $\Gamma$ the subgraphs we extract in light grey, as well as the half-legs in dark green. We also indicated the values of the half-legs. With such a choice of subgraph, we get
  \begin{equs}
    (\bar \Gamma, \ell) = 
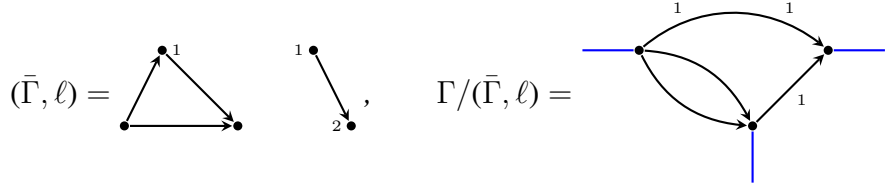
    \begin{tikzpicture}[style={thick},baseline=0.3cm]
      \node[dot] (l) at (0,0) {};
      \node[dot,label={[shift={(0.1,0)}]left:{$\scriptscriptstyle2$}}] (r) at (3,0) {};
      \node[dot,label={[shift={(-0.1,0)}]right:{$\scriptscriptstyle1$}}] (ul) at (0.5,1) {};
      \node[dot,label={[shift={(0.1,0)}]left:{$\scriptscriptstyle1$}}] (ur) at (2.5,1) {};
      \node[dot] (c) at (1.5,0) {};
      \draw[->] (l) -- (ul);
      \draw[->] (ur) -- (r); 
      \draw[->] (l) -- (c);
      \draw[->] (ul) -- (c);
      \end{tikzpicture}\;,\qquad
      \Gamma / (\bar \Gamma, \ell)
      = 
      \begin{tikzpicture}[style={thick},baseline=-0.7cm]
      \node[dot] (l) at (0,0) {};
      \node[dot] (r) at (2.5,0) {};
      \node[dot] (d) at (1.5,-1) {};
      \draw[->] (l) to[bend left=40] node[pos=0.2,above=-0.0] {$\scriptscriptstyle1$} node[pos=0.8,above=-0.0] {$\scriptscriptstyle1$} (r);
      \draw[->] (l) to[bend left=30] (d);
      \draw[->] (l) to[bend right=30] (d);
      \draw[->] (d) -- (r) node[midway,below right=-0.1] {$\scriptscriptstyle1$};
      \draw[thick,blue] (l) -- ++(180:0.75);
      \draw[thick,blue] (r) -- ++(0:0.75);
      \draw[thick,blue] (d) -- ++(-90:0.75);
    \end{tikzpicture}\,.
  \end{equs}
\end{example}
\begin{definition}
We define a coproduct $\Delta\colon \langle F\rangle\to \langle \mfF_-\rangle\otimes\langle F\rangle$ by setting for a single diagram $\Gamma\in F$,
\begin{equation}
  \Delta\Gamma=\sum_{\overline{\Gamma}\subset\Gamma}\sum_{\ell\colon\D\overline{\Gamma}\to\N^d}\frac{1}{\ell!}(\overline{\Gamma},\pi\ell)\otimes\Gamma/(\overline{\Gamma},\ell)\,,
\end{equation}
and by extending it naturally by linearity. The sum runs over the subdiagrams $\overline{\Gamma}$ such that $\deg(\overline{\Gamma},\pi\ell)\leq0$. We forbid the extraction of edges such that $\mfd(e)=1$. We use the convention $\ell!=\prod_{e\in\D\overline{\Gamma}}\ell(e)!$. We extend this definition to vacuum diagrams by defining another coproduct $\Delta^-\colon \langle \mfF_-\rangle\to\langle \mfF_-\rangle\otimes \langle \mfF_-\rangle$.
\begin{equation}
  \Delta^-(\Gamma,\mfn)=\sum_{\overline{\Gamma}\subset\Gamma}\sum_{\scriptscriptstyle{\substack{\overline{\ell}\colon\D\overline{\Gamma}\to\mathbf{N}^d\\\overline{\mfn}\colon\overline{\mcV}\to\mathbf{N}^d}}}\frac{1}{\overline{\ell}!}\binom{\mfn}{\overline\mfn}(\overline{\Gamma},\overline{\mfn}+\pi\overline{\ell})\otimes(\Gamma,\mfn-\overline{\mfn})/(\overline{\Gamma},\ell)\,.
\end{equation}
The sum runs over the subdiagrams $\overline{\Gamma}$ such that $\deg(\overline{\Gamma},\overline{\mfn}+\pi\overline{\ell})\leq0$.
\end{definition}

\begin{definition}
  We define the \textit{renormalised evaluation map} by
  \begin{equation}
    \hat\Pi_\lambda\Gamma=(g\otimes\Pi_\lambda)\Delta\Gamma\,.
  \end{equation}
  where $g\colon\langle \mfF_-\rangle \rightarrow \R$ is a character with respect to the forest product $\bullet$. We do not keep track of the dependency on $g$ in the notation, for conciseness. We extend this notion to vacuum diagrams by setting
  \begin{equation}
    \hat\Pi_\lambda(\Gamma,\mfn)=(g\otimes\Pi_\lambda)\Delta^-(\Gamma,\mfn)\,.
  \end{equation}
\end{definition}
In particular, we will show in Proposition \ref{prop:gbphz} that one has to take the BPHZ character $g_{\scriptscriptstyle\mathrm{BPHZ}}$ to recover the renormalisation constants given by the Polchinski flow.

\begin{remark}
  Even if we do not actually use this fact directly, we would like to mention that the vector space $\langle \mfF_-\rangle$ of non-positive vacuum diagrams endowed with the coproduct $\Delta^-$, the forest product $\bullet$, a unit given by the empty forest, and a counit given to vanish on every non-empty diagram is a graded Hopf algebra. With this algebraic structure, the BPHZ character has a closed formula given by $g_{\scriptscriptstyle\mathrm{BPHZ}}(\Gamma,\mfn)=\Pi_\infty\hat{\mathcal{A}}(\Gamma,\mfn)$, where $\hat{\mathcal{A}}$ is the so-called twisted antipode. The only fact we will need is that it is characterized to vanish on every non-positive negative diagrams. For more details, see \cite[Section 2.4]{BPHZ_theorem}.
\end{remark}

\begin{remark}
  Note that one could have tried to work with elementary differentials, as in \cite{BM25}, and obtain, for a graph \textit{without legs}, but with null vertex decoration,
  \begin{equs}
    (\Pi_\lambda\Gamma)[\phi]=\int_{\Lambda^{\CV_\star}} \prod_{e \in \mathcal{E}_{\star}} \D_\lambda^{\mfd(e)}\D^{\mfe(e)}(G-G_\lambda)(x_{e_+}-x_{e_-})\Upsilon[\Gamma][\phi](x_{v_1},\dots,x_{v_k}) \,\dint x\,,
  \end{equs}
  where 
  \begin{equs}
    \Upsilon[\Gamma,\mfn][\phi](x_{v_1},\dots,x_{v_k})=\prod_{v\in\mcV_\star}\prod_{e\sim v}\D^{\mfn(v)}_{x_v}\mathrm{D}_{\D^{\mfe_\pm(e)}\phi}V[\phi](x_{v})\,.
  \end{equs}
  The second product runs on the half-edges that are neighboring $v$, and $\mfe_\pm(e)$ is their decoration, whether they are incoming or outgoing. $\mathrm{D}$ is a Fréchet derivative, and $V$ is the \textit{unrenormalised} potential. With this definition, the ansatz would read
  \begin{equs}
    V_{t}[\varphi] = \sum_{\Gamma} \frac{1}{S(\Gamma)} (\Pi_\lambda \Gamma)[\varphi]\,.
  \end{equs}
   However, as we lack an inductive construction of the combinatorial objects as with rooted trees, the definition of the renormalised evaluation map is quite unclear in this framework. We believe though that it is possible to show that the ansatz of equation \eqref{eq:ansatzpotential} reads in this formalism
   \begin{equs}
    V_0[\phi]=\sum_{\deg(\Gamma,\mfn)\leq0}\frac{g_{\scriptscriptstyle\mathrm{BPHZ}}(\Gamma,\mfn)}{S(\Gamma,\mfn)}\Upsilon[\Gamma,\mfn][\phi]\,,
   \end{equs}
   which is exactly of the flavour of what is done in the SPDE literature. Moreover, for QFT, tackling non polynomial potentials is not as relevant as in SPDEs, that find applications outside this field strictly speaking. That is why we stick with the framework where the legs are part of the diagrams. This ansatz is very close to the one given on decorated trees in \cite{BM25}, in the context of SPDEs. It is also very similar to \cite{BH25} where one uses multi-indices meaning that the combinatorics focus on nodes and their arity. With multi-indices, one has to consider all the potential pairings between half-edges that lead to the same diagram. This formulation already appeared under the name of pre-Feynman diagrams in \cite{pre_FD} and a similar construction has been used in SPDEs for the first time in \cite{OSSW}.
\end{remark}

\section{A proof of perturbative renormalisation with the Polchinski flow}\label{sec:analytic}

We present in this section an example of how the Polchinski flow can be used to get relevant estimates on the terms of the perturbative expansion of Euclidean quantum fields. We focus on a rather simple model, namely the $\phi^4$ model in the sub-critical (super-renormalisable) regime. The goal of this section is twofold: we show the mechanism of the Polchinski flow in order to give the reader some intuition for the next section and highlight some essential phenomena, and we use a version that is close to the one introduced by Duch \cite{Duc22,Duc23}. More precisely, this version relies on analysis in the physical space, and not the Fourier space as it is usual in physics. Moreover, we use some arguments found by Duch in the treatment of elliptic SPDEs, that lead to substantial simplifications. For example, we do not need to introduce the so-called generalised force coefficients. We let the potential
\begin{equation}
  V_0[\phi]=\int_\Lambda\big(\alpha\phi(x)^4+\mfa_\eps\phi(x)^2+\mfb_\eps\big)\,\dint x\,,
\end{equation}
with renormalisation constants
\begin{equation}
  \mfa_\eps=\sum_{i\leq\sigma/(2\sigma-d)} \alpha^i \mfa_\eps^i \,,\mathrm{and}\quad\mfb_\eps=\sum_{i\leq d/(2\sigma-d)}\alpha^i \mfb_\eps^i\,.
\end{equation}
For this section, the subcriticalily condition reads $d/3<\sigma<d/2$. We make the following formal ansatz on the renormalised potential.
\begin{equation}\label{eq:ansatzphi4}
  V_\lambda[\phi]=\sum_{i=0}^\infty\sum_{m=0}^{3i}\alpha^i \int_{\Lambda^m}V_\lambda^{i,m}(x_1,\dots,x_m)\phi(x_1)\dots\phi(x_m)\,\dint x_{1:m}\,,
\end{equation}
started with the initial conditions
\begin{equation}
  V_\infty^{i,0}=\mfb_\eps^i,\quad\mathrm{and}\quad V_\infty^{i,2}=\mfa_\eps^i\,.
\end{equation}
We recall that we do not claim anything on the convergence of the series \eqref{eq:ansatzphi4}. We see it as a formal power series in $(i,m)$. We are now ready to define a flow of scale of the Green function $G$. We set
\begin{equation}\label{eq:kerneldecomposition}
  G_\lambda(x)=\chi(|x|^\sigma/\lambda)G(x)\,.
\end{equation}
$\chi\colon[0,\infty)\to\R$ is a smooth function such that $\chi(r)=0$ for $r\leq1/4$ and $\chi(r)=1$ for $r\geq1/2$. We denote $\dot G_\lambda=\D_\lambda G_\lambda$. This kernel has the interesting property that its support is included in $\{x\in\Lambda\colon \lambda/4\leq|x|^\sigma\leq\lambda/2\}$. In particular it vanishes for large $\lambda$. We have the following estimate for $p\in[1,\infty]$ and small $\lambda$.
\begin{equation}\label{eq:estimateLp}
	\|\dot G_\lambda\|_{L^p(\Lambda)}\lesssim \lambda^{-\frac{d}{\sigma}\left(1-\frac{1}{p}\right)}
\end{equation}
To state properly the bounds on the coefficients, we define a norm given by
\begin{equation}
\|V\|_{\CV^m}=\sup_{x_1\in\Lambda}\int_{\Lambda^{m-1}}|V(x_{1:m})|\,\dint x_{1:m}\,.
\end{equation}
Note that this norm respects the invariance by translation. Moreover, as in this section the coefficients are symmetric in their arguments, the choice of putting the supremum on the first variable is purely arbitrary. We consequently call $\mcV^m$ the completion of the space of symmetric test functions in $m$ variables with respect to this norm. We can now define another norm on $\mcV^m$ that depends on the flow parameter $\lambda$. 
\begin{equation}
  \snnorm{V}_\lambda=\|K_\lambda^{\otimes m}*V\|_{\CV^m}\,.
\end{equation}
Note that we kept the dependence on $m$ implicit. We have set the scaled test function $K_\lambda=(1-\lambda^{-\sigma}\Delta)^{-1}$.
We define some grafting operators, that (at least graphically) merge renormalised Feynman diagrams together in order to inductively construct bigger renormalised diagrams.
\begin{definition}
	We set, for $U\in\mcV^n$ and $V\in\mcV^m$, 
  \begin{equation}
    B_\lambda(U,V)(x_{1:m_1+m_2+1})=\sum_{i=1}^{m_1}\sum_{j=m_1+1}^{m_1+m_2+1}U(x_{1:m})
    \dot G_\lambda(x_i-x_j)V(x_{m_1+1:m_1+m_2+1})\,,
  \end{equation}
  and for $V\in\mcV^m$, 
  \begin{equation}
    A_\lambda(V)(x_{1:m})=\sum_{1\leq i,j\leq m}V(x_{1:m})\dot G_\lambda(x_i-x_j)\,.
  \end{equation}
\end{definition}
With these objects, one can see that the flow equation \eqref{eq:poleq} reads on the coefficients
\begin{equation}\label{flowcoeff}
  \D_\lambda V_\lambda^{i,m}(x_{1:m})=\frac{1}{2}A_\lambda(V_\lambda^{i,m+2})(x_{1:m})+\frac{1}{2}\sum_{\scriptscriptstyle{\substack{i_1+i_2=i\\m_1+m_2=m+2}}}B_\lambda(V_\lambda^{i_1,m_1},V_\lambda^{i_2,m_2})(x_{1:m})\,.
\end{equation}
Our notations are convenient enough to get the next lemma in a nice form. Its proof can be found in \cite[Lemma 14.10]{Duc22}, and the fact that $\|\dot G_\lambda\|_{L^1(\Lambda)}\lesssim1$, and $\|\dot G_\lambda\|_{L^\infty(\Lambda)}\lesssim\lambda^{-\frac{d}{\sigma}}$, following \eqref{eq:estimateLp}.
\begin{lemma}
  We have, for $U\in\mcV^n$ and $V\in\mcV^m$, 
  \begin{equation}
    \snnorm{B_\lambda(U,V)}_\lambda\lesssim \snnorm{U}_\lambda\snnorm{V}_\lambda\,,
  \end{equation}
  and for $V\in\mcV^m$, 
  \begin{equation}
    \snnorm{A_\lambda(V)}_\lambda\lesssim \lambda^{-\frac{d}{\sigma}}\snnorm{V}_\lambda\,.
  \end{equation}
\end{lemma}
We adapt the notion of degree to the present context. It is directly taken from Definition \ref{def:degree} and adapted by noticing that all the diagrams with $i$ internal vertices and $m$ legs all have the same degree. We have
\begin{equation}
  \deg(i,m)=i(2\sigma-d)+m(d-\sigma)/2-d\,.
\end{equation}
Before jumping to the core of this section, we need one last technical ingredient, namely a Taylor-type expansion for our Feynman diagrams. Its proof can be found in \cite[Theorem 8.8]{Duc22}.
\begin{lemma}
  Let $U\in\mcV^2$, we have
  \begin{equation}\label{eq:Taylor}
    U(x_{1:m})=(\I U)(x_1) \delta(x_2-x_1)+\sum_{|k|=2}\D^k(\L^kU)(x_1,x_2)\,,
  \end{equation}
  where
  \begin{equation}
    (\L^kU)(x_1,x_2)=\frac{|k|}{k!}\int_0^1(1-\tau)^{|k|-1}\tau^{-d}(X^kU)(x_1,x_1+(x_2-x_1)/\tau)\,\dint\tau\,.
  \end{equation}
\end{lemma}

We may now state and prove the following proposition that is the main object of this section. It is a bound on each of the terms of the perturbative expansion of the renormalised potential $V_\lambda$.
\begin{proposition}
  The following holds uniformly in $\eps$ and for every couple $(i,m)$
  \begin{equation}
    \snnorm{V_\lambda^{i,m}}_\lambda\lesssim \lambda^{\frac{\deg(i,m)}{\sigma}}\,.
  \end{equation}
\end{proposition}

\begin{remark}
  Before turning to the proof of the result, that works with an inductive procedure, we would like to mention that this result could be alternatively proved using Hepp sectors techniques as it is done in 
  \cite{BPHZ_theorem}. Taking the kernel $G-G_\lambda$ instead of $G$ only adds more constraints on the Hepp sectors but it is not hard to see that the proof therein can be modified to accommodate this difference.
\end{remark}

\begin{proof}
  The result is proved by induction on the couple $(i,m)$.\\
  \textit{Base case}. The induction is started by the couple $(1,4)$. We have, using that $K_\lambda$ integrates to $1$,
  \begin{equs}
    \snnorm{V_\lambda^{1,4}}_\lambda&=\sup_{x_1}\int_{\Lambda^4}\prod_{i=1}^4 K_\lambda(x_i-y)\,\dint y\dint x_2\dint x_3\dint x_4=1\,,
  \end{equs}
  which is the desired result.\\
  \textit{Induction}. Let us start with the \textit{irrelevant} cumulants, \ie the ones such that $\deg(i,m)>0$.
  We have, using the flow equation on the coefficients,
  \begin{equs}
    \nnorm{V_\lambda^{i,m}}_\lambda&\lesssim\int_0^\lambda\nnorm{\D_s V_s^{i,m}}_s\,\dint s\lesssim \int_0^\lambda s^{-\frac{d}{\sigma}}s^{\frac{\deg(i,m)+d-\sigma}{\sigma}}
    +\sum_{\scriptscriptstyle{\substack{i_1+i_2=i\\m_1+m_2=m+2}}}s^{\frac{\deg(i,m)+d-\sigma-d}{\sigma}}\,\dint s\\
    &\lesssim \int_0^\lambda s^{-1+\frac{\deg(i,m)}{\sigma}}\,\dint s\lesssim \lambda^{\frac{\deg(i,m)}{\sigma}}\,.
  \end{equs}
  We recall that the integrand of all the integrals in this proof vanish for large $s$ thanks to the support properties of $\dot G_\lambda$, so we do not check the integrability in $\infty$. We now deal with the \textit{relevant} coefficients, \ie the ones with $\deg(i,m)\leq0$. Note that this imposes $m=0$ or $m=2$. The terms with $m=0$, the vacuum diagrams, are easily tackled. Indeed we have
  \begin{equs}
    V_\lambda^{i,0}=\mfb_\eps^i-\int_0^\lambda\D_sV_s^{i,0}\,\dint s\,,
  \end{equs}
  so that choosing
  \begin{equs}\label{eq:choiceb}
    \mfb_\eps^i=\int_0^\infty \D_sV_s^{i,0}\,\dint s\,,
  \end{equs}
  we get
  \begin{equs}
    |V_\lambda^{i,0}|\leq \int_\lambda^\infty|\D_sV_s^{i,0}|\,\dint s\lesssim \int_\lambda^\infty s^{\frac{\deg(i,0)}{\sigma}-1}\,\dint s\lesssim \lambda^{\frac{\deg(i,0)}{\sigma}}\,.
  \end{equs}
  For $m=2$, the idea of proof is similar but somewhat more involved. We need to make use of the Taylor expansion \eqref{eq:Taylor} to localize the terms. We write
  \begin{multline*}
    V_\lambda^{i,2}=\left(\mfa_\eps^i-\int_0^\lambda \I(K_s^{\otimes2}*\D_s V_s^{i,2})\,\dint s\right)\delta +\int_0^\lambda (K_s^{\otimes2}-\mathrm{Id})*\D_s V_s^{i,2}\,\dint s\\-\int_0^\lambda \sum_{|k|=2} \D^k (\L^kK_s^{\otimes 2}*V_s^{i,2})\,\dint s\,.
  \end{multline*}
  We can choose the renormalisation constant. We set
  \begin{equs}\label{eq:choicea}
    \mfa_\eps^i=\int_0^\infty \I(K_s^{\otimes2}*\D_s V_s^{i,2})\,\dint s\,.
  \end{equs}
  We can then bound
  \begin{equs}
    \int_\lambda^\infty|\I(K_s^{\otimes2}*\D_s V_s^{i,2})|\,\dint s\leq\int_\lambda^\infty \snnorm{\D_s V_s^{i,2}}_s\,\dint s\leq \int_\lambda^\infty s^{\frac{\deg(i,2)}{\sigma}-1}\,\dint s\lesssim \lambda^{\frac{\deg(i,2)}{\sigma}}\,.
  \end{equs}
  For the second term, we need the following identity, whose proof can be found in \cite[Lemma 4.5]{Duc22}.
  \begin{equs}
    \|K_\lambda^{\otimes 2}*(K_s^{\otimes 2}-\mathrm{Id})*U\|_{\mcV^2}\lesssim (s/\lambda)^{2/\sigma}\|K_s^{\otimes 2}*U\|_{\mcV^2}\,.
  \end{equs}
  We thus have
  \begin{equs}
    \int_0^\lambda \nnorm{(K_s^{\otimes2}-\mathrm{Id})*\D_s V_s^{i,2}}\,\dint s&\lesssim \int_0^\lambda \|K_\lambda*(K_s^{\otimes2}-\mathrm{Id})*\D_s V_s^{i,2}\|_{\mcV^2}\,\dint s\\
    &\lesssim \lambda^{-2/\sigma}\int_0^\lambda s^{2/\sigma} \|K_s*\D_s V_s^{i,2}\|_{\mcV^2}\,\dint s\lesssim \lambda^{-2/\sigma}\int_0^\lambda s^{2/\sigma}\snnorm{V_s^{i,2}}_s\,\dint s\\
    &\lesssim \lambda^{-2/\sigma}\int_0^\lambda s^{(\deg(i,2)+2)/\sigma-1}\,\dint s\lesssim \lambda^{\frac{\deg(i,2)}{\sigma}}\,.
  \end{equs}
  Finally, for the last term we have, for some $k$ such that $|k|=2$,
  \begin{equs}
    \int_0^\lambda\snnorm{\D^k(\L^k K_s^{\otimes 2}*V_s^{i,2})}_s\,\dint s&=\int_0^\lambda \|K_\lambda*\D^k(\L^k K_s^{\otimes 2}*V_s^{i,2})\|_{\mcV^2}\,\dint s\\
    &\lesssim \lambda^{2/\sigma}\int_0^\lambda \|X^k(K_s^{\otimes 2}*V_s^{i,2})\|_{\mcV^2}\,\dint s\,.
  \end{equs}
  Using Lemma 5.6 in \cite{Duc23}, we have 
  \begin{equs}
    \|X^k(K_s^{\otimes 2}*V_s^{i,2})\|_{\mcV^2}\lesssim s^{-2}\|K_s^{\otimes 2}*V_s^{i,2}\|_{\mcV^2}=\snnorm{V_s^{i,2}}_s\,.
  \end{equs}
  Note that this latter identity uses crucially the support properties of $\dot G_\lambda$. Otherwise we would need to use other techniques, such as the generalised force coefficents or weighted norms. It allows to conclude in the same way as the previous case.
\end{proof}

\section{Equivalence proof}

This section is devoted to the proof of the main results in Theorem \ref{maintheorem} and Corollary \ref{corollary:renormmeasure} and is the core of the present paper. In Section \ref{sec:addnotions}, we introduce additional material on Feynman diagrams, especially all the coefficients that appear in the ansatz for the renormalised potential. In Section \ref{sec:algeop} we introduce several algebraic operations on Feynman diagrams, and we prove that they have nice morphism properties with the renormalisation procedure $\hat\Pi_\lambda$. We then introduce in Section \ref{sec:taylor} a combinatorial counterpart of the Taylor expansion of the legs of a Feynman diagram, and we show that it also has a morphism property with the renormalisation procedure. This allows to identify the character $g$ as the BPHZ character $g_{\scriptscriptstyle\mathrm{BPHZ}}$. In Section \ref{sec:forestduality}, we state and prove the duality relation that is one of the keys of the proof of the main theorem, that can be found in Section \ref{sec:lastsection}.

\subsection{Additional notions on Feynman diagrams}\label{sec:addnotions}

We define, for a Feynman diagram $\Gamma\in F$ its \textit{vacuum diagram} as well as its \textit{residue}. The first one, denoted $\mathrm{vac}(\Gamma)$ is simply the diagram whose set of vertices and edges are respectively $\mcV_\star$ and $\mcE_\star$. In other words, it is the biggest subdiagram that we can extract from $\Gamma$. The residue of $\Gamma$, denoted $\mathrm{res}(\Gamma)$ is the diagram obtained by extracting $\mathrm{vac}(\Gamma)$ and a contraction of the remaining edges (in that case the legs) to one point. Moreover, given a map $\ell\colon\Lab\to\N^d$, we define $\mathrm{res}_\ell(\Gamma)$ similarly as $\mathrm{res}(\Gamma)$ but with decoration $\mfe|_\Lab+\ell$. As a drawing is better than many words we give an example.
\begin{example}
  We have
  \begin{equs}
    \mathrm{vac}\Big(\begin{tikzpicture}[scale=0.15,baseline=-0.1cm, trim left=-0.4cm, trim right=1cm]
			\node at (0,0)  [dot,label= {[label distance=-0.2em]below: \scriptsize  $      $} ] (center) {};
			\node at (-3,3)  [label= {[label distance=-0.2em]right: \scriptsize  $     $} ] (topleft) {};
			\node at (-3,-3)  [label= {[label distance=-0.2em]right: \scriptsize  $     $} ] (botleft) {};
      \node at (4,0)  [dot,label= {[label distance=-0.2em]right: \scriptsize  $     $} ] (centerright) {};
      \node at (7,3)  [label= {[label distance=-0.2em]right: \scriptsize  $     $} ] (leg1) {};
      \node at (7,-3)  [label= {[label distance=-0.2em]right: \scriptsize  $     $} ] (leg2) {};
      \draw[arrow] (center) .. controls (1,1) and (3,1) .. (centerright);
      \draw[arrow] (center) .. controls (1,-1) and (3,-1) .. (centerright);
      \draw[kernel1,color=blue] (center) to
			node [sloped,below] {\small }     (topleft);
      \draw[kernel1,color=blue] (center) to
			node [sloped,below] {\small }     (botleft);
      \draw[kernel1,color=blue] (centerright) to
			node [sloped,below] {\small }     (leg1);
      \draw[kernel1,color=blue] (centerright) to
			node [sloped,below] {\small }     (leg2);
    \end{tikzpicture}\Big)=\begin{tikzpicture}[scale=0.15,baseline=-0.1cm, trim left=-0.4cm, trim right=1cm]
			\node at (0,0)  [dot,label= {[label distance=-0.2em]below: \scriptsize  $      $} ] (center) {};
      \node at (4,0)  [dot,label= {[label distance=-0.2em]right: \scriptsize  $     $} ] (centerright) {};
      \draw[arrow] (center) .. controls (1,1) and (3,1) .. (centerright);
      \draw[arrow] (center) .. controls (1,-1) and (3,-1) .. (centerright);
    \end{tikzpicture},\mathrm{and}\quad
    \mathrm{res}_\ell\Big(\begin{tikzpicture}[scale=0.15,baseline=-0.1cm, trim left=-0.4cm, trim right=1cm]
			\node at (0,0)  [dot,label= {[label distance=-0.2em]below: \scriptsize  $      $} ] (center) {};
			\node at (-3,3)  [label= {[label distance=-0.2em]right: \scriptsize  $     $} ] (topleft) {};
			\node at (-3,-3)  [label= {[label distance=-0.2em]right: \scriptsize  $     $} ] (botleft) {};
      \node at (4,0)  [dot,label= {[label distance=-0.2em]right: \scriptsize  $     $} ] (centerright) {};
      \node at (7,3)  [label= {[label distance=-0.2em]right: \scriptsize  $     $} ] (leg1) {};
      \node at (7,-3)  [label= {[label distance=-0.2em]right: \scriptsize  $     $} ] (leg2) {};
      \draw[arrow] (center) .. controls (1,1) and (3,1) .. (centerright);
      \draw[arrow] (center) .. controls (1,-1) and (3,-1) .. (centerright);
      \draw[kernel1,color=blue] (center) to
			node [sloped,below] {\small }     (topleft);
      \draw[kernel1,color=blue] (center) to
			node [sloped,below] {\small }     (botleft);
      \draw[kernel1,color=blue] (centerright) to
			node [sloped,below] {\small }     (leg1);
      \draw[kernel1,color=blue] (centerright) to
			node [sloped,below] {\small }     (leg2);
    \end{tikzpicture}\Big)=\begin{tikzpicture}[scale=0.2,baseline=-0.1cm, trim right=0.6cm, trim left=-0.6cm]
			\node at (0,0)  [dot,label= {[label distance=-0.2em]below: \scriptsize  $      $} ] (center) {};
			\node at (-3,3)  [label= {[label distance=-0.2em]right: \scriptsize  $     $} ] (topleft) {};
			\node at (-3,-3)  [label= {[label distance=-0.2em]right: \scriptsize  $     $} ] (botleft) {};
			\node at (3,3)  [label= {[label distance=-0.2em]right: \scriptsize  $     $} ] (topright) {};
      \node at (3,-3)  [label= {[label distance=-0.2em]right: \scriptsize  $     $} ] (botright) {};
			\draw[kernel1,color=blue] (center) to
			node [sloped,below] {\small }     (topright);
      \draw[kernel1,color=blue] (center) to
			node [sloped,below] {\small }     (topleft);
      \draw[kernel1,color=blue] (center) to
			node [sloped,below] {\small }     (botright);
      \draw[kernel1,color=blue] (center) to
			node [sloped,below] {\small }     (botleft);
      \node at (1.2,1.2)  [circle,fill=white, scale=0.4] {$\ell_3$};
      \node at (1.2,-1.2)  [circle,fill=white, scale=0.4] {$\ell_4$};
      \node at (-1.2,-1.2)  [circle,fill=white, scale=0.4] {$\ell_2$};
      \node at (-1.2,1.2)  [circle,fill=white, scale=0.4] {$\ell_1$};
    \end{tikzpicture}.
  \end{equs}
\end{example}
Additionally, we will need for our ansatz a notion of rule, \ie given a theory, which set of graphs we are allowed to consider. To do that, we consider a set of elementary graphs $E^\star\subset F_0$, that is a set of pure legs Feynman diagrams with only one vertex. For example for the $\phi^4$ it would contain only the graph with one vertex and four legs with null decoration. From this (generally small) set of diagrams, we build a bigger set $F^\star\subset F_0$ by taking all the diagrams we can get by pairing the legs (in the sense of the operations $\graft_{v_1,v_2}$ and $\circlearrowleft_{v_1,v_2}$ defined in the next section) of any number of elementary diagrams. Again with the $\phi^4$ theory, this would be the set of all diagrams with arity $4$ on every vertex.

For $\Gamma\in F$, we define its symmetry factor $S(\Gamma)$ as the cardinal of the group of its automorphisms $\mathrm{Aut}(\Gamma)$. More precisely, it is a permutation of its vertices and half-edges that transforms $\Gamma$ in a Feynman diagram isomorphic to $\Gamma$ and that keeps the endpoints of all its edges. 
\begin{example}
  For elementary diagrams, the symmetry factor is simply given by the permutations of the legs. We have
  \begin{equs}
    S\Big(\begin{tikzpicture}[scale=0.15,baseline=-0.1cm, trim right=0.4cm, trim left=-0.4cm]
			\node at (0,0)  [dot,label= {[label distance=-0.2em]below: \scriptsize  $      $} ] (center) {};
			\node at (-3,3)  [label= {[label distance=-0.2em]right: \scriptsize  $     $} ] (topleft) {};
			\node at (-3,-3)  [label= {[label distance=-0.2em]right: \scriptsize  $     $} ] (botleft) {};
			\node at (3,3)  [label= {[label distance=-0.2em]right: \scriptsize  $     $} ] (topright) {};
      \node at (3,-3)  [label= {[label distance=-0.2em]right: \scriptsize  $     $} ] (botright) {};
			\draw[kernel1,color=blue] (center) to
			node [sloped,below] {\small }     (topright);
      \draw[kernel1,color=blue] (center) to
			node [sloped,below] {\small }     (topleft);
      \draw[kernel1,color=blue] (center) to
			node [sloped,below] {\small }     (botright);
      \draw[kernel1,color=blue] (center) to
			node [sloped,below] {\small }     (botleft);
    \end{tikzpicture}\Big)=4!\,,\quad S\Big(\begin{tikzpicture}[scale=0.2,baseline=-0.1cm, trim right=0.6cm, trim left=-0.6cm]
			\node at (0,0)  [dot,label= {[label distance=-0.2em]below: \scriptsize  $      $} ] (center) {};
			\node at (-3,3)  [label= {[label distance=-0.2em]right: \scriptsize  $     $} ] (topleft) {};
			\node at (-3,-3)  [label= {[label distance=-0.2em]right: \scriptsize  $     $} ] (botleft) {};
			\node at (3,3)  [label= {[label distance=-0.2em]right: \scriptsize  $     $} ] (topright) {};
      \node at (3,-3)  [label= {[label distance=-0.2em]right: \scriptsize  $     $} ] (botright) {};
			\draw[kernel1,color=blue] (center) to
			node [sloped,below] {\small }     (topright);
      \draw[kernel1,color=blue] (center) to
			node [sloped,below] {\small }     (topleft);
      \draw[kernel1,color=blue] (center) to
			node [sloped,below] {\small }     (botright);
      \draw[kernel1,color=blue] (center) to
			node [sloped,below] {\small }     (botleft);
      \node at (1.2,1.2)  [circle,fill=white, scale=0.5] {$1$};
      \node at (1.2,-1.2)  [circle,fill=white, scale=0.5] {$1$};
    \end{tikzpicture}\Big)=(2!)^2\,.
  \end{equs}
  The symmetry factors take the decorations into account, like so
  \begin{equs}
    S\Big(\begin{tikzpicture}[scale=0.15,baseline=-0.1cm, trim left=-0.4cm, trim right=1cm]
			\node at (0,0)  [dot,label= {[label distance=-0.2em]below: \scriptsize  $      $} ] (center) {};
			\node at (-3,3)  [label= {[label distance=-0.2em]right: \scriptsize  $     $} ] (topleft) {};
			\node at (-3,-3)  [label= {[label distance=-0.2em]right: \scriptsize  $     $} ] (botleft) {};
      \node at (4,0)  [dot,label= {[label distance=-0.2em]right: \scriptsize  $     $} ] (centerright) {};
      \node at (7,3)  [label= {[label distance=-0.2em]right: \scriptsize  $     $} ] (leg1) {};
      \node at (7,-3)  [label= {[label distance=-0.2em]right: \scriptsize  $     $} ] (leg2) {};
      \draw[arrow, color=red] (center) .. controls (1,1) and (3,1) .. (centerright);
      \draw[arrow] (center) .. controls (1,-1) and (3,-1) .. (centerright);
      \draw[kernel1,color=blue] (center) to
			node [sloped,below] {\small }     (topleft);
      \draw[kernel1,color=blue] (center) to
			node [sloped,below] {\small }     (botleft);
      \draw[kernel1,color=blue] (centerright) to
			node [sloped,below] {\small }     (leg1);
      \draw[kernel1,color=blue] (centerright) to
			node [sloped,below] {\small }     (leg2);
    \end{tikzpicture}\Big)=2^2\,,\quad S\Big(\begin{tikzpicture}[scale=0.15,baseline=-0.1cm, trim left=-0.4cm, trim right=1cm]
			\node at (0,0)  [dot,label= {[label distance=-0.2em]below: \scriptsize  $      $} ] (center) {};
			\node at (-3,3)  [label= {[label distance=-0.2em]right: \scriptsize  $     $} ] (topleft) {};
			\node at (-3,-3)  [label= {[label distance=-0.2em]right: \scriptsize  $     $} ] (botleft) {};
      \node at (4,0)  [dot,label= {[label distance=-0.2em]right: \scriptsize  $     $} ] (centerright) {};
      \node at (7,3)  [label= {[label distance=-0.2em]right: \scriptsize  $     $} ] (leg1) {};
      \node at (7,-3)  [label= {[label distance=-0.2em]right: \scriptsize  $     $} ] (leg2) {};
      \draw[arrow] (center) .. controls (1,1) and (3,1) .. (centerright);
      \draw[arrow] (center) .. controls (1,-1) and (3,-1) .. (centerright);
      \draw[kernel1,color=blue] (center) to
			node [sloped,below] {\small }     (topleft);
      \draw[kernel1,color=blue] (center) to
			node [sloped,below] {\small }     (botleft);
      \draw[kernel1,color=blue] (centerright) to
			node [sloped,below] {\small }     (leg1);
      \draw[kernel1,color=blue] (centerright) to
			node [sloped,below] {\small }     (leg2);
    \end{tikzpicture}\Big)=2^3\,.
  \end{equs}
\end{example}
If $(\Gamma,\mfn)$ is a vacuum diagram, then we set $S(\Gamma,\mfn)=S(\Gamma)\mfn!$.

For a Feynman diagram and a given vertex $v\in\mcV_\star$, we define $\mathrm{elem}(v)\in E^\star$ as the elementary diagram whose leg decoration is given by the adjacent edge decorations of $v$, either incoming or outgoing.
\begin{example}
  Let us provide a simple example for the sake of clarity.
  \begin{equs}
    \mathrm{elem}(v_1)\Big(\begin{tikzpicture}[scale=0.15,baseline=-0.1cm, trim right=1.1cm, trim left=-0.5cm]
			\node at (0,0)  [dot,label= {[label distance=-0.2em]below: \scriptsize  $      $} ] (center) {};
			\node at (-4,0)  [label= {[label distance=-0.2em]right: \scriptsize  $     $} ] (topleft) {};
			\node at (0,4)  [label= {[label distance=-0.2em]right: \scriptsize  $     $} ] (botleft) {};
			\node at (0,-4)  [label= {[label distance=-0.2em]right: \scriptsize  $     $} ] (topright) {};
      \node at (4,0)  [dot,label= {[label distance=-0.2em]right: \scriptsize  $     $} ] (centerright) {};
      \node at (4,4)  [label= {[label distance=-0.2em]right: \scriptsize  $     $} ] (leg1) {};
      \node at (8,0)  [label= {[label distance=-0.2em]right: \scriptsize  $     $} ] (leg2) {};
      \node at (4,-4)  [label= {[label distance=-0.2em]right: \scriptsize  $     $} ] (leg3) {};
      \draw[arrow] (center) to
			node [sloped,below] {\small }     (centerright);
			\draw[kernel1,color=blue] (center) to
			node [sloped,below] {\small }     (topright);
      \draw[kernel1,color=blue] (center) to
			node [sloped,below] {\small }     (topleft);
      \draw[kernel1,color=blue] (center) to
			node [sloped,below] {\small }     (botleft);
      \draw[kernel1,color=blue] (centerright) to
			node [sloped,below] {\small }     (leg1);
      \draw[kernel1,color=blue] (centerright) to
			node [sloped,below] {\small }     (leg2);
      \draw[kernel1,color=blue] (centerright) to
			node [sloped,below] {\small }     (leg3);
      \node at (3,1.25)  [fill=white, scale=0.5] {$2$};
      \node at (1,1.25)  [fill=white, scale=0.5] {$1$};
    \end{tikzpicture}\Big)=\begin{tikzpicture}[scale=0.2,baseline=-0.1cm, trim right=0.5cm, trim left=-0.5cm]
			\node at (0,0)  [dot,label= {[label distance=-0.2em]below: \scriptsize  $      $} ] (center) {};
			\node at (-3,3)  [label= {[label distance=-0.2em]right: \scriptsize  $     $} ] (topleft) {};
			\node at (-3,-3)  [label= {[label distance=-0.2em]right: \scriptsize  $     $} ] (botleft) {};
			\node at (3,3)  [label= {[label distance=-0.2em]right: \scriptsize  $     $} ] (topright) {};
      \node at (3,-3)  [label= {[label distance=-0.2em]right: \scriptsize  $     $} ] (botright) {};
			\draw[kernel1,color=blue] (center) to
			node [sloped,below] {\small }     (topright);
      \draw[kernel1,color=blue] (center) to
			node [sloped,below] {\small }     (topleft);
      \draw[kernel1,color=blue] (center) to
			node [sloped,below] {\small }     (botright);
      \draw[kernel1,color=blue] (center) to
			node [sloped,below] {\small }     (botleft);
      \node at (1.2,1.2)  [circle,fill=white, scale=0.5] {$1$};
    \end{tikzpicture},\mathrm{and}\quad\mathrm{elem}(v_2)\Big(\begin{tikzpicture}[scale=0.15,baseline=-0.1cm, trim right=1.1cm, trim left=-0.5cm]
			\node at (0,0)  [dot,label= {[label distance=-0.2em]below: \scriptsize  $      $} ] (center) {};
			\node at (-4,0)  [label= {[label distance=-0.2em]right: \scriptsize  $     $} ] (topleft) {};
			\node at (0,4)  [label= {[label distance=-0.2em]right: \scriptsize  $     $} ] (botleft) {};
			\node at (0,-4)  [label= {[label distance=-0.2em]right: \scriptsize  $     $} ] (topright) {};
      \node at (4,0)  [dot,label= {[label distance=-0.2em]right: \scriptsize  $     $} ] (centerright) {};
      \node at (4,4)  [label= {[label distance=-0.2em]right: \scriptsize  $     $} ] (leg1) {};
      \node at (8,0)  [label= {[label distance=-0.2em]right: \scriptsize  $     $} ] (leg2) {};
      \node at (4,-4)  [label= {[label distance=-0.2em]right: \scriptsize  $     $} ] (leg3) {};
      \draw[arrow] (center) to
			node [sloped,below] {\small }     (centerright);
			\draw[kernel1,color=blue] (center) to
			node [sloped,below] {\small }     (topright);
      \draw[kernel1,color=blue] (center) to
			node [sloped,below] {\small }     (topleft);
      \draw[kernel1,color=blue] (center) to
			node [sloped,below] {\small }     (botleft);
      \draw[kernel1,color=blue] (centerright) to
			node [sloped,below] {\small }     (leg1);
      \draw[kernel1,color=blue] (centerright) to
			node [sloped,below] {\small }     (leg2);
      \draw[kernel1,color=blue] (centerright) to
			node [sloped,below] {\small }     (leg3);
      \node at (3,1.25)  [fill=white, scale=0.5] {$2$};
      \node at (1,1.25)  [fill=white, scale=0.5] {$1$};
    \end{tikzpicture}\Big)=\begin{tikzpicture}[scale=0.2,baseline=-0.1cm, trim right=0.5cm, trim left=-0.5cm]
			\node at (0,0)  [dot,label= {[label distance=-0.2em]below: \scriptsize  $      $} ] (center) {};
			\node at (-3,3)  [label= {[label distance=-0.2em]right: \scriptsize  $     $} ] (topleft) {};
			\node at (-3,-3)  [label= {[label distance=-0.2em]right: \scriptsize  $     $} ] (botleft) {};
			\node at (3,3)  [label= {[label distance=-0.2em]right: \scriptsize  $     $} ] (topright) {};
      \node at (3,-3)  [label= {[label distance=-0.2em]right: \scriptsize  $     $} ] (botright) {};
			\draw[kernel1,color=blue] (center) to
			node [sloped,below] {\small }     (topright);
      \draw[kernel1,color=blue] (center) to
			node [sloped,below] {\small }     (topleft);
      \draw[kernel1,color=blue] (center) to
			node [sloped,below] {\small }     (botright);
      \draw[kernel1,color=blue] (center) to
			node [sloped,below] {\small }     (botleft);
      \node at (1.2,1.2)  [circle,fill=white, scale=0.5] {$2$};
    \end{tikzpicture}.
  \end{equs}
  We have denoted by $v_1$ and $v_2$ the left and right vertices respectively.
\end{example}
Then, for $\mathbf\Gamma\in \mfF$, we define its graph factorial $\mathbf\Gamma!$ in the following way.
\begin{equation}
  \mathbf\Gamma!=\prod_{v\in\mcV_\star}S\big(\mathrm{elem}(v)\big)\,.
\end{equation}
In that case we see the forest as a not necessarily connected Feynman diagram. Another way to see this quantity is that it is the product of the symmetry factors of the elementary diagrams of which its composed. Note that for elementary diagrams, $S(\Gamma)=\Gamma!$.

Moreover, we define for any elementary diagram of $E^\star$ an associated coefficient $\alpha\colon E^\star\to\R$. This allows to write the \textit{unrenormalised} potential under the form
\begin{equation}
  V[\phi]=\sum_{\Gamma\in E^\star}\alpha(\Gamma)(\Pi_\lambda\Gamma)[\phi]\,.
\end{equation}
For a more general diagram $\Gamma\in F^\star$, we define 
\begin{equation}
  \alpha(\Gamma)=\prod_{v\in\mcV_\star}\alpha\big(\mathrm{elem}(v)\big)\,.
\end{equation}

\begin{example}
  We have
  \begin{equs}
    \begin{tikzpicture}[scale=0.15,baseline=-0.1cm, trim right=1.1cm, trim left=-0.5cm]
			\node at (0,0)  [dot,label= {[label distance=-0.2em]below: \scriptsize  $      $} ] (center) {};
			\node at (-4,0)  [label= {[label distance=-0.2em]right: \scriptsize  $     $} ] (topleft) {};
			\node at (0,4)  [label= {[label distance=-0.2em]right: \scriptsize  $     $} ] (botleft) {};
			\node at (0,-4)  [label= {[label distance=-0.2em]right: \scriptsize  $     $} ] (topright) {};
      \node at (4,0)  [dot,label= {[label distance=-0.2em]right: \scriptsize  $     $} ] (centerright) {};
      \node at (4,4)  [label= {[label distance=-0.2em]right: \scriptsize  $     $} ] (leg1) {};
      \node at (8,0)  [label= {[label distance=-0.2em]right: \scriptsize  $     $} ] (leg2) {};
      \node at (4,-4)  [label= {[label distance=-0.2em]right: \scriptsize  $     $} ] (leg3) {};
      \draw[arrow] (center) to
			node [sloped,below] {\small }     (centerright);
			\draw[kernel1,color=blue] (center) to
			node [sloped,below] {\small }     (topright);
      \draw[kernel1,color=blue] (center) to
			node [sloped,below] {\small }     (topleft);
      \draw[kernel1,color=blue] (center) to
			node [sloped,below] {\small }     (botleft);
      \draw[kernel1,color=blue] (centerright) to
			node [sloped,below] {\small }     (leg1);
      \draw[kernel1,color=blue] (centerright) to
			node [sloped,below] {\small }     (leg2);
      \draw[kernel1,color=blue] (centerright) to
			node [sloped,below] {\small }     (leg3);
    \end{tikzpicture}!=(4!)^2\,,\mathrm{and}\quad\alpha\Big(\begin{tikzpicture}[scale=0.15,baseline=-0.1cm, trim right=1.1cm, trim left=-0.5cm]
			\node at (0,0)  [dot,label= {[label distance=-0.2em]below: \scriptsize  $      $} ] (center) {};
			\node at (-4,0)  [label= {[label distance=-0.2em]right: \scriptsize  $     $} ] (topleft) {};
			\node at (0,4)  [label= {[label distance=-0.2em]right: \scriptsize  $     $} ] (botleft) {};
			\node at (0,-4)  [label= {[label distance=-0.2em]right: \scriptsize  $     $} ] (topright) {};
      \node at (4,0)  [dot,label= {[label distance=-0.2em]right: \scriptsize  $     $} ] (centerright) {};
      \node at (4,4)  [label= {[label distance=-0.2em]right: \scriptsize  $     $} ] (leg1) {};
      \node at (8,0)  [label= {[label distance=-0.2em]right: \scriptsize  $     $} ] (leg2) {};
      \node at (4,-4)  [label= {[label distance=-0.2em]right: \scriptsize  $     $} ] (leg3) {};
      \draw[arrow] (center) to
			node [sloped,below] {\small }     (centerright);
			\draw[kernel1,color=blue] (center) to
			node [sloped,below] {\small }     (topright);
      \draw[kernel1,color=blue] (center) to
			node [sloped,below] {\small }     (topleft);
      \draw[kernel1,color=blue] (center) to
			node [sloped,below] {\small }     (botleft);
      \draw[kernel1,color=blue] (centerright) to
			node [sloped,below] {\small }     (leg1);
      \draw[kernel1,color=blue] (centerright) to
			node [sloped,below] {\small }     (leg2);
      \draw[kernel1,color=blue] (centerright) to
			node [sloped,below] {\small }     (leg3);
    \end{tikzpicture}\Big)=\alpha\Big(\begin{tikzpicture}[scale=0.15,baseline=-0.1cm, trim right=0.4cm, trim left=-0.4cm]
			\node at (0,0)  [dot,label= {[label distance=-0.2em]below: \scriptsize  $      $} ] (center) {};
			\node at (-3,3)  [label= {[label distance=-0.2em]right: \scriptsize  $     $} ] (topleft) {};
			\node at (-3,-3)  [label= {[label distance=-0.2em]right: \scriptsize  $     $} ] (botleft) {};
			\node at (3,3)  [label= {[label distance=-0.2em]right: \scriptsize  $     $} ] (topright) {};
      \node at (3,-3)  [label= {[label distance=-0.2em]right: \scriptsize  $     $} ] (botright) {};
			\draw[kernel1,color=blue] (center) to
			node [sloped,below] {\small }     (topright);
      \draw[kernel1,color=blue] (center) to
			node [sloped,below] {\small }     (topleft);
      \draw[kernel1,color=blue] (center) to
			node [sloped,below] {\small }     (botright);
      \draw[kernel1,color=blue] (center) to
			node [sloped,below] {\small }     (botleft);
    \end{tikzpicture}\Big)^2\,.
  \end{equs}
\end{example}

Finally, we say that a diagram $\Gamma\in\ F_1$ is one-particle-irreducible (1PI) if removing the red edge yields a connected diagram. Otherwise we say that it is one-particle-reducible (1PR).

\subsection{Algebraic operations on diagrams and morphism properties}\label{sec:algeop}

We introduce in this section three algebraic operations on Feynman diagrams, $\uparrow$, $\rightarrow$, and $\circlearrowleft$ that are the combinatorial counterpart of analytic operations. We show that they have nice morphism properties with respect to the renormalised evaluation map $\hat\Pi_\lambda$.

\subsubsection{Commutation with the derivative}

The first operation is the counterpart of the derivative in $\lambda$. We show that it commutes with $\hat\Pi_\lambda$.

\begin{definition}
We define a map $\uparrow\colon F_0\to\langle F_1\rangle$ by
\begin{equation}
  \uparrow\Gamma=\sum_{e\in\CE_\star}\uparrow^e\Gamma\,,
\end{equation}
and we extend it by linearity to a map $\uparrow\colon \langle F_0\rangle\to\langle F_1\rangle$.
\end{definition}

\begin{example}
  We give a graphical example of how this map acts.
  \begin{equs}
    \uparrow\begin{tikzpicture}[scale=0.2,baseline=0.3cm,trim right=0.3cm]
			\node at (0,0)  [dot,label= {[label distance=-0.2em]below: \scriptsize  $      $} ] (root) {};
			\node at (-2.5,2)  [dot,label= {[label distance=-0.2em]right: \scriptsize  $     $} ] (center) {};
			\node at (0,4)  [dot,label= {[label distance=-0.2em]right: \scriptsize  $     $} ] (centerc) {};
			\node at (2,-1)  [,label= {[label distance=-0.2em]right: \scriptsize  $     $} ] (right) {};
      \node at (2,5)  [,label= {[label distance=-0.2em]right: \scriptsize  $     $} ] (topright) {};
      \node at (-5.5,2)  [dot,label= {[label distance=-0.2em]right: \scriptsize  $     $} ] (centerleft) {};
      \node at (-8,4)  [dot,label= {[label distance=-0.2em]right: \scriptsize  $     $} ] (topleft) {};
			\node at (-8,0)  [dot,label= {[label distance=-0.2em]right: \scriptsize  $     $} ] (botleft) {};
      \node at (-10,-1)  [label= {[label distance=-0.2em]right: \scriptsize  $     $} ] (leg1) {};
      \node at (-10,5)  [label= {[label distance=-0.2em]right: \scriptsize  $     $} ] (leg2) {};
      \draw[arrow] (center) to
			node [sloped,below] {\small }     (root);
			\draw[arrow] (center) to
			node [sloped,below] {\small }     (centerc);
			\draw[arrow] (centerc) to
			node [sloped,below] {\small }     (root);
      \draw[arrow] (centerleft) to
			node [sloped,below] {\small }     (center);
      \draw[arrow] (centerleft) to
			node [sloped,below] {\small }     (topleft);
      \draw[arrow] (centerleft) to
			node [sloped,below] {\small }     (botleft);
      \draw[arrow] (topleft) to
			node [sloped,below] {\small }     (botleft);
			\draw[kernel1,color=blue] (root) to
			node [sloped,below] {\small }     (right);
      \draw[kernel1,color=blue] (centerc) to
			node [sloped,below] {\small }     (topright);
      \draw[kernel1,color=blue] (topleft) to
			node [sloped,below] {\small }     (leg2);
      \draw[kernel1,color=blue] (botleft) to
			node [sloped,below] {\small }     (leg1);
    \end{tikzpicture}=\begin{tikzpicture}[scale=0.2,baseline=0.3cm,trim right=0.3cm]
			\node at (0,0)  [dot,label= {[label distance=-0.2em]below: \scriptsize  $      $} ] (root) {};
			\node at (-2.5,2)  [dot,label= {[label distance=-0.2em]right: \scriptsize  $     $} ] (center) {};
			\node at (0,4)  [dot,label= {[label distance=-0.2em]right: \scriptsize  $     $} ] (centerc) {};
			\node at (2,-1)  [,label= {[label distance=-0.2em]right: \scriptsize  $     $} ] (right) {};
      \node at (2,5)  [,label= {[label distance=-0.2em]right: \scriptsize  $     $} ] (topright) {};
      \node at (-5.5,2)  [dot,label= {[label distance=-0.2em]right: \scriptsize  $     $} ] (centerleft) {};
      \node at (-8,4)  [dot,label= {[label distance=-0.2em]right: \scriptsize  $     $} ] (topleft) {};
			\node at (-8,0)  [dot,label= {[label distance=-0.2em]right: \scriptsize  $     $} ] (botleft) {};
      \node at (-10,-1)  [label= {[label distance=-0.2em]right: \scriptsize  $     $} ] (leg1) {};
      \node at (-10,5)  [label= {[label distance=-0.2em]right: \scriptsize  $     $} ] (leg2) {};
      \draw[arrow] (center) to
			node [sloped,below] {\small }     (root);
			\draw[arrow] (center) to
			node [sloped,below] {\small }     (centerc);
			\draw[arrow] (centerc) to
			node [sloped,below] {\small }     (root);
      \draw[arrow, color=red] (centerleft) to
			node [sloped,below] {\small }     (center);
      \draw[arrow] (centerleft) to
			node [sloped,below] {\small }     (topleft);
      \draw[arrow] (centerleft) to
			node [sloped,below] {\small }     (botleft);
      \draw[arrow] (topleft) to
			node [sloped,below] {\small }     (botleft);
			\draw[kernel1,color=blue] (root) to
			node [sloped,below] {\small }     (right);
      \draw[kernel1,color=blue] (centerc) to
			node [sloped,below] {\small }     (topright);
      \draw[kernel1,color=blue] (topleft) to
			node [sloped,below] {\small }     (leg2);
      \draw[kernel1,color=blue] (botleft) to
			node [sloped,below] {\small }     (leg1);
    \end{tikzpicture}+2\begin{tikzpicture}[scale=0.2,baseline=0.3cm,trim right=0.3cm]
			\node at (0,0)  [dot,label= {[label distance=-0.2em]below: \scriptsize  $      $} ] (root) {};
			\node at (-2.5,2)  [dot,label= {[label distance=-0.2em]right: \scriptsize  $     $} ] (center) {};
			\node at (0,4)  [dot,label= {[label distance=-0.2em]right: \scriptsize  $     $} ] (centerc) {};
			\node at (2,-1)  [,label= {[label distance=-0.2em]right: \scriptsize  $     $} ] (right) {};
      \node at (2,5)  [,label= {[label distance=-0.2em]right: \scriptsize  $     $} ] (topright) {};
      \node at (-5.5,2)  [dot,label= {[label distance=-0.2em]right: \scriptsize  $     $} ] (centerleft) {};
      \node at (-8,4)  [dot,label= {[label distance=-0.2em]right: \scriptsize  $     $} ] (topleft) {};
			\node at (-8,0)  [dot,label= {[label distance=-0.2em]right: \scriptsize  $     $} ] (botleft) {};
      \node at (-10,-1)  [label= {[label distance=-0.2em]right: \scriptsize  $     $} ] (leg1) {};
      \node at (-10,5)  [label= {[label distance=-0.2em]right: \scriptsize  $     $} ] (leg2) {};
      \draw[arrow] (center) to
			node [sloped,below] {\small }     (root);
			\draw[arrow] (center) to
			node [sloped,below] {\small }     (centerc);
			\draw[arrow] (centerc) to
			node [sloped,below] {\small }     (root);
      \draw[arrow] (centerleft) to
			node [sloped,below] {\small }     (center);
      \draw[arrow] (centerleft) to
			node [sloped,below] {\small }     (topleft);
      \draw[arrow] (centerleft) to
			node [sloped,below] {\small }     (botleft);
      \draw[arrow, color=red] (topleft) to
			node [sloped,below] {\small }     (botleft);
			\draw[kernel1,color=blue] (root) to
			node [sloped,below] {\small }     (right);
      \draw[kernel1,color=blue] (centerc) to
			node [sloped,below] {\small }     (topright);
      \draw[kernel1,color=blue] (topleft) to
			node [sloped,below] {\small }     (leg2);
      \draw[kernel1,color=blue] (botleft) to
			node [sloped,below] {\small }     (leg1);
    \end{tikzpicture}+4\begin{tikzpicture}[scale=0.2,baseline=0.3cm,trim right=0.3cm]
			\node at (0,0)  [dot,label= {[label distance=-0.2em]below: \scriptsize  $      $} ] (root) {};
			\node at (-2.5,2)  [dot,label= {[label distance=-0.2em]right: \scriptsize  $     $} ] (center) {};
			\node at (0,4)  [dot,label= {[label distance=-0.2em]right: \scriptsize  $     $} ] (centerc) {};
			\node at (2,-1)  [,label= {[label distance=-0.2em]right: \scriptsize  $     $} ] (right) {};
      \node at (2,5)  [,label= {[label distance=-0.2em]right: \scriptsize  $     $} ] (topright) {};
      \node at (-5.5,2)  [dot,label= {[label distance=-0.2em]right: \scriptsize  $     $} ] (centerleft) {};
      \node at (-8,4)  [dot,label= {[label distance=-0.2em]right: \scriptsize  $     $} ] (topleft) {};
			\node at (-8,0)  [dot,label= {[label distance=-0.2em]right: \scriptsize  $     $} ] (botleft) {};
      \node at (-10,-1)  [label= {[label distance=-0.2em]right: \scriptsize  $     $} ] (leg1) {};
      \node at (-10,5)  [label= {[label distance=-0.2em]right: \scriptsize  $     $} ] (leg2) {};
      \draw[arrow] (center) to
			node [sloped,below] {\small }     (root);
			\draw[arrow] (center) to
			node [sloped,below] {\small }     (centerc);
			\draw[arrow] (centerc) to
			node [sloped,below] {\small }     (root);
      \draw[arrow] (centerleft) to
			node [sloped,below] {\small }     (center);
      \draw[arrow, color=red] (centerleft) to
			node [sloped,below] {\small }     (topleft);
      \draw[arrow] (centerleft) to
			node [sloped,below] {\small }     (botleft);
      \draw[arrow] (topleft) to
			node [sloped,below] {\small }     (botleft);
			\draw[kernel1,color=blue] (root) to
			node [sloped,below] {\small }     (right);
      \draw[kernel1,color=blue] (centerc) to
			node [sloped,below] {\small }     (topright);
      \draw[kernel1,color=blue] (topleft) to
			node [sloped,below] {\small }     (leg2);
      \draw[kernel1,color=blue] (botleft) to
			node [sloped,below] {\small }     (leg1);
    \end{tikzpicture}\,.
  \end{equs}
\end{example}
We have the following relation, that highlights that this red decoration represents on the combinatorial side the derivative $\D_\lambda$.
\begin{equation}
  \D_\lambda\Pi_\lambda=\Pi_\lambda\uparrow\,.
\end{equation}

\begin{lemma}\label{commut} One has on $\langle F_0\rangle$,
\begin{equation}
  \D_\lambda\hat\Pi_\lambda=\hat\Pi_\lambda\uparrow\,.
\end{equation}
\end{lemma}

\begin{proof}
We have, for $\Gamma\in F_0$,
\begin{equs}
  \D_\lambda\hat\Pi_\lambda\Gamma&=\D_\lambda(g\otimes\Pi_\lambda)\Delta\Gamma=(g\otimes\Pi_\lambda\uparrow)\Delta\Gamma\\
  &=\sum_{\overline{\Gamma}\subset\Gamma}\sum_{\ell\colon\D\overline{\Gamma}\to\N^d}\frac{1}{\ell!}g(\overline{\Gamma},\pi\ell)\Pi_\lambda\sum_{e\notin\CE_\star(\overline{\Gamma})}\uparrow^e\big(\Gamma/(\overline{\Gamma},\ell)\big)\\
  &=\sum_{\overline{\Gamma}\subset\Gamma}\sum_{\ell\colon\D\overline{\Gamma}\to\N^d}\frac{1}{\ell!}g(\overline{\Gamma},\pi\ell)\Pi_\lambda\sum_{e\notin\CE_\star(\overline{\Gamma})}(\uparrow^e\Gamma)/(\overline{\Gamma},\ell)\\
  &=\sum_{e\in\CE_\star}\sum_{\overline{\Gamma}\subset\uparrow^e\Gamma}\sum_{\ell\colon\D\overline{\Gamma}\to\N^d}\frac{1}{\ell!}g(\overline{\Gamma},\pi\ell)\Pi_\lambda(\uparrow^e\Gamma)/(\overline{\Gamma},\ell)\\
  &=(g\otimes\Pi_\lambda)\Delta\uparrow\Gamma.
\end{equs}
Note that we have crucially used the fact that the coproduct forbids the extraction of the edges with decoration $\mfd(e)=1$ to go from line 2 to line 3.
\end{proof}

\begin{remark}
  A similar commutation Lemma \ref{commut} appears also in \cite[Proposition 4.3]{BM25}. This idea is coming from the convergence via spectral gap in regularity structures initiated in \cite{LOTT}. Indeed, in \cite[Proposition 4.1]{BN23}, the authors show that there is a commutation property between the Malliavin derivative and a combinatorial counterpart with the model. 
\end{remark}

\subsubsection{Graftings and morphism properties}

We define two grafting-type operations that will represent the two operations on the right-hand side of the Polchinski equation. We then prove that they have good morphism properties with respect to the renormalisation procedure.

\begin{definition}
Let $\Gamma_1,\Gamma_2\in\mfF_0$. We define the bilinear operation $\graft\colon F_0\times F_0\to \langle F_1\rangle$
\begin{equation}
  \Gamma_1\graft\Gamma_2=\sum_{\substack{\scriptscriptstyle v_1\in\Lab_1\\\scriptscriptstyle v_2\in\Lab_2}}\Gamma_1\graft_{v_1,v_2}\Gamma_2\,,
\end{equation}
where $\graft_{v_1,v_2}\colon F_0\times F_0\to F_1$ creates a new Feynman diagram by replacing the two legs $v_1$ and $v_2$ by an internal edge going from $v_1$ to $v_2$ with $\mfd$ decoration equal to $1$ and $\mfe$ decoration equal to $\big(\mfe(v_2),\mfe(v_1)\big)$ only if $\mft(v_1)=\mft(v_2)$, and is $0$ otherwise. We extend $\graft$ as an application $\langle F_0\rangle\otimes\langle F_0\rangle\to\langle F_1\rangle$ by bilinearity. We will also require an operation $\circlearrowleft_{v_1,v_2}\colon F_0\to\langle F_1\rangle$, for $v_1$ and $v_2$ in the set $\Lab$ of legs of some diagram $\Gamma\in F_0$. We define $\circlearrowleft_{v_1,v_2}\Gamma$ as the diagram where the two legs have been merged into one internal edge going from $v_1$ to $v_2$ with the same decoration as the previous bilinear operation and the same condition on the $\mft$ decoration. We can then define an operation $\circlearrowleft\colon F_0\to\langle F_1\rangle$.
\begin{equation}
  \circlearrowleft\Gamma=\sum_{v_1,v_2\in\Lab}\circlearrowleft_{v_1,v_2}\Gamma\,.
\end{equation}
We naturally extend it to a map $\circlearrowleft\colon \langle F_0\rangle\to\langle F_1\rangle$ by linearity.
\end{definition}

\begin{example}
  We give some examples of the operations we have just introduced. We start with the bilinear operation. If we take the elementary diagram of the $\phi^4$ theory, this would give
  \begin{equs}
.
  \end{equs}
  With some non-trivial edge decorations, we would have for instance
    \begin{equs}
    %
.
  \end{equs}
  Again with the $\phi^4$ theory diagrams, the linear operation looks like
  \begin{equs}
    \circlearrowleft%
.
  \end{equs}
  This also works with renormalisd diagrams as follows
  \begin{equs}
    (%
.
  \end{equs}
  $\mfc$ would be a renormalisation constant.
\end{example}

With these definitions, we get the two natural properties
\begin{equation}\label{vanillaprelie}
  \Pi_\lambda(\Gamma_1\graft\Gamma_2)[\phi]=\big(\mathrm{D}(\Pi_\lambda\Gamma_1)\dot G_\lambda\mathrm{D}(\Pi_\lambda\Gamma_2)\big)[\phi]\,,
\end{equation}
and
\begin{equation}
  \Pi_\lambda(\circlearrowleft\Gamma)[\phi]=\mathrm{Tr}\big(\dot G_\lambda \mathrm{D}^2(\Pi_\lambda\Gamma)\big)[\phi].
\end{equation}

\begin{lemma}\label{prelie1}
The following holds for $\Gamma_1,\Gamma_2\in\langle F_0\rangle$
\begin{equation}
  \hat\Pi_\lambda(\Gamma_1\graft\Gamma_2)[\phi]=\big(\mathrm{D}(\hat\Pi_\lambda\Gamma_1)\dot G_\lambda\mathrm{D}(\hat\Pi_\lambda\Gamma_2)\big)[\phi]\,.
\end{equation}
\end{lemma}

\begin{proof}
We have, using that edges with derivatives in $\lambda$ cannot be extracted,
\begin{equs}
\Delta(\Gamma_1\graft\Gamma_2)&=\sum_{\substack{\scriptscriptstyle v_1\in\Lab_1\\\scriptscriptstyle v_2\in\Lab_2}}\sum_{\scriptscriptstyle\overline{\Gamma}\subset\Gamma_1\graft_{v_1,v_2}\Gamma_2}\sum_{\scriptscriptstyle\ell\colon\D\overline{\Gamma}\to\N^d}\frac{1}{\ell!}(\overline{\Gamma},\pi\ell)\otimes(\Gamma_1\graft_{v_1,v_2}\Gamma_2)/(\overline{\Gamma},\ell)\\
&=\sum_{\substack{\scriptscriptstyle v_1\in\Lab_1\\\scriptscriptstyle v_2\in\Lab_2}}\sum_{\substack{\scriptscriptstyle\overline{\Gamma}_1\subset\Gamma_1\\\scriptscriptstyle \overline{\Gamma}_2\subset\Gamma_2}}\sum_{\scriptscriptstyle\ell\colon\D\overline{\Gamma}_1\bullet\overline{\Gamma}_2\to\N^d}\frac{1}{\ell!}(\overline{\Gamma},\pi\ell)\otimes(\Gamma_1\graft_{v_1,v_2}\Gamma_2)/(\overline{\Gamma},\ell)\\
&=\sum_{\substack{\scriptscriptstyle\overline{\Gamma}_1\subset\Gamma_1\\\scriptscriptstyle \overline{\Gamma}_2\subset\Gamma_2}}\sum_{\scriptscriptstyle\ell\colon\D\overline{\Gamma}_1\bullet\overline{\Gamma}_2\to\N^d}\frac{1}{\ell!}(\overline{\Gamma},\pi\ell)\otimes(\Gamma_1\graft\Gamma_2)/(\overline{\Gamma},\ell)\\
&=\sum_{\substack{\scriptscriptstyle\overline{\Gamma}_1\subset\Gamma_1\\\scriptscriptstyle \overline{\Gamma}_2\subset\Gamma_2}}\sum_{\substack{\scriptscriptstyle\ell_1\colon\D\overline{\Gamma}_1\to\N^d\\\scriptscriptstyle \ell_2\colon\D\overline{\Gamma}_2\to\N^d}}\frac{1}{\ell_1!}\frac{1}{\ell_2!}\big((\overline{\Gamma}_1,\pi\ell_1)\bullet(\overline{\Gamma}_2,\pi\ell_2)\big)\otimes\\
&\qquad\qquad\qquad\qquad\qquad\qquad\qquad\big((\Gamma_1/(\overline{\Gamma}_1,\ell_1))\graft(\Gamma_1/(\overline{\Gamma}_2,\ell_2))\big)\\
&=\Delta\Gamma_1(\mathrm{id}\otimes\graft)\Delta\Gamma_2.
\end{equs}
We have used the fact that for indices with disjoint support, we have $(\ell_1+\ell_2)!=\ell_1!\ell_2!$ and a similar argument for $\out$. We directly conclude using \ref{vanillaprelie} as well as the fact that $g$ is a character with respect to $\bullet$.
\end{proof}
We can obtain the following lemma with a similar proof.

\begin{lemma}\label{prelie2}
  We have, for $\Gamma_1\in\langle F_0\rangle$,
  \begin{equation}
    \hat\Pi_\lambda(\circlearrowleft\Gamma)[\phi]=\mathrm{Tr}\big(\dot G_\lambda \mathrm{D}^2(\hat\Pi_\lambda\Gamma)\big)[\phi]\,.
  \end{equation}
\end{lemma}

\begin{remark}
  These two operations are the counterpart on graphs of the grafting map $\mathcolor{red}{\curvearrowright}_a$ on decorated trees used in \cite{BM25}. This map also has a morphism property, but that is harder to prove since elementary differentials are used therein.
\end{remark}

\subsection{Taylor expansion and BPHZ renormalisation}\label{sec:taylor}

In this section, we define an abstract counterpart of the Taylor expansion of the legs of a Feynman diagram. This allows us to identify the character $g$ as the BPHZ character.
\begin{definition}
  We set, for $\Gamma\in F_0$, 
  \begin{equation}
    \I\Gamma=\sum_{\ell\colon\D\mathrm{vac}(\Gamma)\to\N^d}\frac{1}{\ell!}\big(\mathrm{vac}(\Gamma),\pi\ell\big)\bullet\mathrm{res}_\ell(\Gamma)\,.
  \end{equation}
We extend it on $\langle F_0\rangle$ by linearity.
\end{definition}

\begin{example}
  We provide below an example of what this map does on a rather simple diagram.
  \begin{equs}\label{eq:exampletaylor}
    \I\begin{tikzpicture}[scale=0.3,baseline=-0.1cm,trim right=0cm]
			\node at (0,0)  [label= {[label distance=-0.2em]below: \scriptsize  $      $} ] (root) {};
      \node at (-2.5,0)  [dot,label= {[label distance=-0.2em]right: \scriptsize  $ $} ] (centerleft) {};
      \node at (-5,2)  [dot,label= {[label distance=-0.2em]right: \scriptsize  $ $} ] (topleft) {};
			\node at (-5,-2)  [dot,label= {[label distance=-0.2em]right: \scriptsize  $ $} ] (botleft) {};
      \node at (-7,-4)  [label= {[label distance=-0.2em]right: \scriptsize  $     $} ] (leg1) {};
      \node at (-7,4)  [label= {[label distance=-0.2em]right: \scriptsize  $     $} ] (leg2) {};
      \draw[kernel1,color=blue] (root) to
			node [sloped,below] {\small }     (centerleft);
      \draw[arrow] (centerleft) to
			node [sloped,below] {\small }     (topleft);
      \draw[arrow] (centerleft) to
			node [sloped,below] {\small }     (botleft);
      \draw[arrow] (topleft) to
			node [sloped,below] {\small }     (botleft);
      \draw[kernel1,color=blue] (topleft) to
			node [sloped,below] {\small }     (leg2);
      \draw[kernel1,color=blue] (botleft) to
			node [sloped,below] {\small }     (leg1);
    \end{tikzpicture}=\sum_{i,j,k\in\N^d}\frac{1}{i!j!k!}\begin{tikzpicture}[scale=0.3,baseline=-0.1cm,trim right=0.1cm]
      \node at (-2.5,0)  [dot,label= {[label distance=-0.2em]right: \scriptsize  $k$} ] (centerleft) {};
      \node at (-5,2)  [dot,label= {[label distance=-0.2em]left: \scriptsize  $j$} ] (topleft) {};
			\node at (-5,-2)  [dot,label= {[label distance=-0.2em]left: \scriptsize  $i$} ] (botleft) {};
      \draw[arrow] (centerleft) to
			node [sloped,below] {\small }     (topleft);
      \draw[arrow] (centerleft) to
			node [sloped,below] {\small }     (botleft);
      \draw[arrow] (topleft) to
			node [sloped,below] {\small }     (botleft);
      \node at (0.5,0)  [dot,label= {[label distance=-0.2em]below: \scriptsize  $      $} ] (center) {};
			\node at (-2,2.5)  [label= {[label distance=-0.2em]right: \scriptsize  $     $} ] (topleft) {};
      \node at (-2,-2.5)  [label= {[label distance=-0.2em]right: \scriptsize  $     $} ] (botleft) {};
			\node at (3.25,0)  [label= {[label distance=-0.2em]right: \scriptsize  $     $} ] (right) {};
			\draw[kernel1,color=blue] (center) to
			node [sloped,below] {\small }     (right);
      \draw[kernel1,color=blue] (center) to
			node [sloped,below] {\small }     (botleft);
      \draw[kernel1,color=blue] (center) to
			node [sloped,below] {\small }     (topleft);
      \node at (-0.4,1)  [circle, fill=white, label= {center: \scriptsize  $j$} ] {};
      \node at (-0.4,-1)  [circle, fill=white, label= {center: \scriptsize  $i$} ] {};
      \node at (1.75,0)  [fill=white, label= {center: \scriptsize  $k$} ] {};
    \end{tikzpicture}
  \end{equs}
  This maps represents on the combinatorial side a formal expansion on the legs of the Feynman diagrams, in the sense that the following formal equality holds
  \begin{equs}\label{eq:taylornormal}
    (\Pi_\lambda\Gamma)[\phi]=(\Pi_\lambda \I\Gamma)[\phi]\,.
  \end{equs}
  Let us demonstrate this with the example above. We have
  \begin{equs}
    \Pi_\lambda\Big(\begin{tikzpicture}[scale=0.15,baseline=-0.1cm,trim right=0cm]
			\node at (0,0)  [label= {[label distance=-0.2em]below: \scriptsize  $      $} ] (root) {};
      \node at (-2.5,0)  [dot,label= {[label distance=-0.2em]right: \scriptsize  $ $} ] (centerleft) {};
      \node at (-5,2)  [dot,label= {[label distance=-0.2em]right: \scriptsize  $ $} ] (topleft) {};
			\node at (-5,-2)  [dot,label= {[label distance=-0.2em]right: \scriptsize  $ $} ] (botleft) {};
      \node at (-7,-4)  [label= {[label distance=-0.2em]right: \scriptsize  $     $} ] (leg1) {};
      \node at (-7,4)  [label= {[label distance=-0.2em]right: \scriptsize  $     $} ] (leg2) {};
      \draw[kernel1,color=blue] (root) to
			node [sloped,below] {\small }     (centerleft);
      \draw[arrow] (centerleft) to
			node [sloped,below] {\small }     (topleft);
      \draw[arrow] (centerleft) to
			node [sloped,below] {\small }     (botleft);
      \draw[arrow] (topleft) to
			node [sloped,below] {\small }     (botleft);
      \draw[kernel1,color=blue] (topleft) to
			node [sloped,below] {\small }     (leg2);
      \draw[kernel1,color=blue] (botleft) to
			node [sloped,below] {\small }     (leg1);
    \end{tikzpicture}\Big)[\phi]=\int_{\Lambda^3}(G-G_\lambda)(x,y)(G-G_\lambda)(x,z)(G-G_\lambda)(y,z)\phi(x)\phi(y)\phi(z)\,\dint x \dint y \dint z\,. 
  \end{equs}
  We then perform a formal Taylor expansion, writing $\phi(y)=\sum_{j\in\N^d}\frac{\D^j\phi(x)}{k!}(y-x)^j$, and similarly for $\phi(z)$. We then get
  \begin{multline*}
    \Pi_\lambda\Big(\begin{tikzpicture}[scale=0.15,baseline=-0.1cm,trim right=0cm]
			\node at (0,0)  [label= {[label distance=-0.2em]below: \scriptsize  $      $} ] (root) {};
      \node at (-2.5,0)  [dot,label= {[label distance=-0.2em]right: \scriptsize  $ $} ] (centerleft) {};
      \node at (-5,2)  [dot,label= {[label distance=-0.2em]right: \scriptsize  $ $} ] (topleft) {};
			\node at (-5,-2)  [dot,label= {[label distance=-0.2em]right: \scriptsize  $ $} ] (botleft) {};
      \node at (-7,-4)  [label= {[label distance=-0.2em]right: \scriptsize  $     $} ] (leg1) {};
      \node at (-7,4)  [label= {[label distance=-0.2em]right: \scriptsize  $     $} ] (leg2) {};
      \draw[kernel1,color=blue] (root) to
			node [sloped,below] {\small }     (centerleft);
      \draw[arrow] (centerleft) to
			node [sloped,below] {\small }     (topleft);
      \draw[arrow] (centerleft) to
			node [sloped,below] {\small }     (botleft);
      \draw[arrow] (topleft) to
			node [sloped,below] {\small }     (botleft);
      \draw[kernel1,color=blue] (topleft) to
			node [sloped,below] {\small }     (leg2);
      \draw[kernel1,color=blue] (botleft) to
			node [sloped,below] {\small }     (leg1);
    \end{tikzpicture}\Big)[\phi]=\sum_{j,k\in\N^d}\frac{1}{j!k!}\int_{\Lambda^3}(G-G_\lambda)(x,y)\\\times(G-G_\lambda)(x,z)(G-G_\lambda)(y,z)(y-x)^j(z-x)^k\phi(x)\D^j\phi(x)\D^k\phi(x)\,\dint x \dint y \dint z\,.
  \end{multline*}
  Performing the changes of variables $y\to y+x$ and $z\to z+x$, we get that it is equal to (assuming that $G$ is symmetric for simplicity),
  \begin{equs}
    \sum_{j,k\in\N^d}\frac{1}{j!k!}\int_{\Lambda^2}(G-G_\lambda)(y)(G-G_\lambda)(z)(G-G_\lambda)(y,z)y^jz^k\,\dint y \dint z\int_\Lambda\phi(x)\D^j\phi(x)\D^k\phi(x)\,\dint x\,.
  \end{equs}
  Then the remark we made in Example \ref{ex:vacuum} about the change of base point still applies here, and one can see that \eqref{eq:exampletaylor} is exactly the combinatorial counterpart of this latter equation.
\end{example}
We show that the relation \eqref{eq:taylornormal} still holds when the renormalisation comes into play.
\begin{lemma}\label{lemma:renormtaylor}
  We have the formal equality for any $\Gamma\in\langle F_0\rangle$,
  \begin{equation}\label{eq:renormtaylor}
    (\hat\Pi_\lambda\Gamma)[\phi]=(\hat\Pi_\lambda \I\Gamma)[\phi]\,.
  \end{equation}
\end{lemma}

\begin{proof}
  Let us first take a look at the left-hand side. We have, using \eqref{eq:taylornormal},
  \begin{equs}
    (\hat\Pi_\lambda\Gamma)[\phi]&=\sum_{\overline{\Gamma}\subset\Gamma}\sum_{\overline{\ell}\colon\D\overline{\Gamma}\to\N^d}\frac{1}{\overline{\ell}!}g(\overline{\Gamma},\pi\overline{\ell})\big(\Pi_\lambda(\Gamma/(\overline{\Gamma},\overline{\ell}))\big)[\phi]\\
    &=\sum_{\overline{\Gamma}\subset\Gamma}\sum_{\overline{\ell}\colon\D\overline{\Gamma}\to\N^d}\frac{1}{\overline{\ell}!}g(\overline{\Gamma},\pi\overline{\ell})\big(\Pi_\lambda(\I(\Gamma/(\overline{\Gamma},\overline{\ell})))\big)[\phi]
  \end{equs}
    From this point, we write $\overline{\ell}=\overline{\ell}_1+\overline{\ell}_2$, where $\overline{\ell}_1$ is supported on the half-legs that touch legs of $\Gamma$ and $\overline{\ell}_2$ is supported on half-legs that touch internal edges of $\Gamma$. We thus get
  \begin{equs}
    (\hat\Pi_\lambda\Gamma)[\phi]
    &=\sum_{\overline{\Gamma}\subset\Gamma}\sum_{\overline{\ell}_1,\overline{\ell}_2}\sum_{\scriptscriptstyle{\ell:\D\mathrm{vac}(\Gamma)\to\N^d}}\frac{1}{\overline{\ell}_1!\overline{\ell}_2!}\frac{1}{\ell!}g(\overline{\Gamma},\pi\overline{\ell}_1+\pi\overline{\ell}_2)\\ & \qquad \times\Pi_\lambda\big(\mathrm{vac}(\Gamma/(\overline{\Gamma},\overline{\ell}_2)),\pi\ell\big)\big(\Pi_\lambda\mathrm{res}_{\ell+\overline{\ell}_1}(\Gamma)\big)[\phi]\\
    &=\sum_{\overline{\Gamma}\subset\Gamma}\sum_{\overline{\ell}_1,\overline{\ell}_2}\sum_{\scriptscriptstyle{\ell:\D\mathrm{vac}(\Gamma)\to\N^d}}\frac{1}{\overline{\ell}_1!\overline{\ell}_2!}\frac{1}{(\ell-\overline{\ell}_1)!}g(\overline{\Gamma},\pi\overline{\ell}_1+\pi\overline{\ell}_2)\\  &  \qquad   \times\Pi_\lambda\big(\mathrm{vac}(\Gamma/(\overline{\Gamma},\overline{\ell}_2)),\pi\ell-\pi\overline{\ell}_1\big)\big(\Pi_\lambda\mathrm{res}_\ell(\Gamma)\big)[\phi]\\
    &=\sum_{\overline{\Gamma}\subset\Gamma}\sum_{\overline{\ell}_1,\overline{\ell}_2}\sum_{\scriptscriptstyle{\ell:\D\mathrm{vac}(\Gamma)\to\N^d}}\frac{1}{\ell!\overline{\ell}_2!}\binom{\ell}{\overline{\ell}_1}g(\overline{\Gamma},\pi\overline{\ell}_1+\pi\overline{\ell}_2)\\  &  \qquad \times\Pi_\lambda\big(\mathrm{vac}(\Gamma/(\overline{\Gamma},\overline{\ell}_2)),\pi\ell-\pi\overline{\ell}_1\big)\big(\Pi_\lambda\mathrm{res}_\ell(\Gamma)\big)[\phi].
  \end{equs}
  We have also used support properties to get that $(\overline{\ell}_1+\overline{\ell}_2)!=\overline{\ell}_1!\overline{\ell}_2!$. On the other side, we have
  \begin{equs}
    (\hat\Pi_\lambda \I\Gamma)[\phi]
    &=\sum_{\scriptscriptstyle\ell:\D\mathrm{vac}(\Gamma)\to\N^d}\frac{1}{\ell!}\hat\Pi_\lambda(\mathrm{vac}(\Gamma),\pi\ell)\big(\Pi_\lambda\mathrm{res}_\ell(\Gamma)\big)[\phi]\\
    &=\sum_{\scriptscriptstyle\ell:\D\mathrm{vac}(\Gamma)\to\N^d}\frac{1}{\ell!}\sum_{\overline{\Gamma}\subset\mathrm{vac}(\Gamma)}\sum_{\scriptscriptstyle{\substack{\overline{\ell}\colon\D\overline{\Gamma}\to\mathbf{N}^d\\\overline{\mfn}\colon\overline{\mcV}\to\mathbf{N}^d}}}\frac{1}{\overline{\ell}!}\binom{\pi\ell}{\overline\mfn}g(\overline{\Gamma},\overline{\mfn}+\pi\overline{\ell}) \\ & \qquad \times\Pi_\lambda\big((\mathrm{vac}(\Gamma),\pi\ell-\overline{\mfn})/(\overline{\Gamma},\overline{\ell})\big)\big(\Pi_\lambda\mathrm{res}_\ell(\Gamma)\big)[\phi]\\
    &=\sum_{\overline{\Gamma}\subset\Gamma}\sum_{\substack{\scriptscriptstyle\overline{\ell}\colon\D\overline{\Gamma}\to\mathbf{N}^d\\\scriptscriptstyle\overline{\mfn}\colon\overline{\mcV}\to\mathbf{N}^d}}\sum_{\scriptscriptstyle\ell:\D\mathrm{vac}(\Gamma)\to\N^d}\frac{1}{\ell!\overline{\ell}!}\binom{\pi\ell}{\overline\mfn}g(\overline{\Gamma},\overline{\mfn}+\pi\overline{\ell})\\ & \qquad \times\Pi_\lambda\big((\mathrm{vac}(\Gamma),\pi\ell-\overline{\mfn})/(\overline{\Gamma},\overline{\ell})\big)\big(\Pi_\lambda\mathrm{res}_\ell(\Gamma)\big)[\phi].
  \end{equs}
  It is now easy to see that the two quantities match by using the Chu-Vandermonde identity, and by noticing that we can commute the quotient and the polynomial decoration in the two vacuum diagrams. Note that in the second series of equalities, the subdiagram extraction is only on $\mathrm{vac}(\Gamma)$, so, on the contrary of the extractions on $\Gamma$, the half-legs \textit{cannot} touch the legs. This is compensated by the summation on $\overline{\mfn}$. That is why in the first series of equalities, we take the care to separate the half-legs that touch the legs and the ones that do not.
\end{proof}
\begin{proposition}\label{prop:gbphz}
  The unique character $g$ satisfying the BPHZ boundary conditions at $\lambda=\infty$ for the Polchinski equation is the character $g_{{\scriptscriptstyle\mathrm{BPHZ}}}$.
\end{proposition}
\begin{proof}
  This proof follows the one of \cite[Theorem 5.14]{BM25}. However, the fact that we do not use elementary differentials makes it substantially easier. We recall from \cite[Proposition 2.22]{BPHZ_theorem} that the BPHZ character is uniquely characterised by the fact that
  \begin{equs}
    \hat\Pi_\infty(\Gamma,\mfn)=0\,,
  \end{equs}
  for every vacuum diagram such that $\deg(\Gamma,\mfn)\leq0$. This is equivalent to say that
  \begin{equs}\label{eq:characBPHZ}
    \hat\Pi_\infty\big(\mathrm{vac}(\Gamma),\mfn\big)=0\,,
  \end{equs}
  for every $\Gamma\in F_0$ such that $\deg(\Gamma)+|\mfn|\leq0$. We recall that $G_\infty=0$. We can now use Lemma \ref{lemma:renormtaylor} to see that it allows us to conclude the proof of the proposition. One can notice that, when treating (inductively) the relevant diagrams in the context of the control of the coefficients of the renormalised potential, the Taylor expansion and then the choice of the renormalisation constant exactly imply \eqref{eq:characBPHZ} for all the relevant diagrams. For clarification, one can notice that, in the context of Section \ref{sec:analytic}, the choice of constants \eqref{eq:choiceb},\eqref{eq:choicea} implies that $V_\infty^{i,0}=0$ and $V_\infty^{i,2}=0$ for every $i$. In that case, the Taylor expansion yields only one order $0$ term (the second one disappears by invariance by translation), but more complex models would lead to cancel more terms of the Taylor expansion.
\end{proof}

\subsection{Forests and duality}\label{sec:forestduality}

  We introduce in this section two operations on forests, that can be thought respectively as a grafting and a cutting operation. The main statement is the novel duality relation in Lemma \ref{lemma:duality}, that will be essential to the proof of the main theorem in the next section. 
  \begin{definition}
    Let $\mathbf{\Gamma}=\Gamma_1\dots\Gamma_n$ a forest. We define a map $L\colon\mfF_0\to\mfF_1$ by
    \begin{multline}
      L\mathbf{\Gamma}=\sum_{k=1}^n\Gamma_1\dots\Gamma_{k-1}\big(\circlearrowleft\Gamma_k\big)\Gamma_{k+1}\dots\Gamma_n\\
      +\sum_{1\leq i\neq j\leq n}\big(\Gamma_i\graft\Gamma_j\big)\Gamma_1\dots\Gamma_{i-1}\Gamma_{i+1}\dots\Gamma_{j-1}\Gamma_{j+1}\dots\Gamma_n\,.
    \end{multline}
    We define another map $D\colon F_1\to\mfF_0$ by
    \begin{equation}
      D\Gamma=\sum_{e\in\mcE_\star}\sum_{k=1}^N D_{e,k}\Gamma\,,
    \end{equation}
    where this latter operation cuts the red edge that has decoration $\mfd(e)=1$ and creates two legs with $\mfe$ decoration respectively $\mfe_+(e)$ and $\mfe_-(e)$, and $\mft$ decoration $k$ on both sides. We extend this map to forests using Leibniz' rule.
    \begin{equation}
      D\mathbf{\Gamma}=\sum_{k=1}^n\Gamma_1\dots\Gamma_{k-1}\big(D\Gamma_k\big)\Gamma_{k+1}\dots\Gamma_n\,.
    \end{equation}
    Finally we extend it as a linear map $D\colon \langle\mfF_1\rangle\to\langle\mfF_0\rangle$.
  \end{definition}
    \begin{example}
    Let us demonstrate this on a simple example
    \begin{equs}
      D\begin{tikzpicture}[scale=0.15,baseline=-0.1cm, trim right=1.1cm, trim left=-0.5cm]
			\node at (0,0)  [dot,label= {[label distance=-0.2em]below: \scriptsize  $      $} ] (center) {};
			\node at (-4,0)  [label= {[label distance=-0.2em]right: \scriptsize  $     $} ] (topleft) {};
			\node at (0,4)  [label= {[label distance=-0.2em]right: \scriptsize  $     $} ] (botleft) {};
			\node at (0,-4)  [label= {[label distance=-0.2em]right: \scriptsize  $     $} ] (topright) {};
      \node at (4,0)  [dot,label= {[label distance=-0.2em]right: \scriptsize  $     $} ] (centerright) {};
      \node at (4,4)  [label= {[label distance=-0.2em]right: \scriptsize  $     $} ] (leg1) {};
      \node at (8,0)  [label= {[label distance=-0.2em]right: \scriptsize  $     $} ] (leg2) {};
      \node at (4,-4)  [label= {[label distance=-0.2em]right: \scriptsize  $     $} ] (leg3) {};
      \draw[arrow,color=red] (center) to
			node [sloped,below] {\small }     (centerright);
			\draw[kernel1,color=blue] (center) to
			node [sloped,below] {\small }     (topright);
      \draw[kernel1,color=blue] (center) to
			node [sloped,below] {\small }     (topleft);
      \draw[kernel1,color=blue] (center) to
			node [sloped,below] {\small }     (botleft);
      \draw[kernel1,color=blue] (centerright) to
			node [sloped,below] {\small }     (leg1);
      \draw[kernel1,color=blue] (centerright) to
			node [sloped,below] {\small }     (leg2);
      \draw[kernel1,color=blue] (centerright) to
			node [sloped,below] {\small }     (leg3);
      \node at (3,1.25)  [fill=white, scale=0.5] {$1$};
      \node at (1,1.25)  [fill=white, scale=0.5] {$1$};
    \end{tikzpicture}\begin{tikzpicture}[scale=0.15,baseline=-0.1cm, trim left=-0.4cm, trim right=0.8cm]
			\node at (0,0)  [dot,label= {[label distance=-0.2em]below: \scriptsize  $      $} ] (center) {};
			\node at (-3,3)  [label= {[label distance=-0.2em]right: \scriptsize  $     $} ] (topleft) {};
			\node at (-3,-3)  [label= {[label distance=-0.2em]right: \scriptsize  $     $} ] (botleft) {};
      \node at (4,0)  [dot,label= {[label distance=-0.2em]right: \scriptsize  $     $} ] (centerright) {};
      \node at (7,3)  [label= {[label distance=-0.2em]right: \scriptsize  $     $} ] (leg1) {};
      \node at (7,-3)  [label= {[label distance=-0.2em]right: \scriptsize  $     $} ] (leg2) {};
      \draw[arrow, color=red] (centerright) .. controls (3,1) and (1,1) .. (center);
      \draw[arrow] (center) .. controls (1,-1) and (3,-1) .. (centerright);
      \draw[kernel1,color=blue] (center) to
			node [sloped,below] {\small }     (topleft);
      \draw[kernel1,color=blue] (center) to
			node [sloped,below] {\small }     (botleft);
      \draw[kernel1,color=blue] (centerright) to
			node [sloped,below] {\small }     (leg1);
      \draw[kernel1,color=blue] (centerright) to
			node [sloped,below] {\small }     (leg2);
    \end{tikzpicture}=\begin{tikzpicture}[scale=0.2,baseline=-0.1cm, trim right=0.6cm, trim left=-0.6cm]
			\node at (0,0)  [dot,label= {[label distance=-0.2em]below: \scriptsize  $      $} ] (center) {};
			\node at (-3,3)  [label= {[label distance=-0.2em]right: \scriptsize  $     $} ] (topleft) {};
			\node at (-3,-3)  [label= {[label distance=-0.2em]right: \scriptsize  $     $} ] (botleft) {};
			\node at (3,3)  [label= {[label distance=-0.2em]right: \scriptsize  $     $} ] (topright) {};
      \node at (3,-3)  [label= {[label distance=-0.2em]right: \scriptsize  $     $} ] (botright) {};
			\draw[kernel1,color=blue] (center) to
			node [sloped,below] {\small }     (topright);
      \draw[kernel1,color=blue] (center) to
			node [sloped,below] {\small }     (topleft);
      \draw[kernel1,color=blue] (center) to
			node [sloped,below] {\small }     (botright);
      \draw[kernel1,color=blue] (center) to
			node [sloped,below] {\small }     (botleft);
      \node at (1.2,1.2)  [circle,fill=white, scale=0.5] {$1$};
    \end{tikzpicture}\begin{tikzpicture}[scale=0.2,baseline=-0.1cm, trim right=0.6cm, trim left=-0.6cm]
			\node at (0,0)  [dot,label= {[label distance=-0.2em]below: \scriptsize  $      $} ] (center) {};
			\node at (-3,3)  [label= {[label distance=-0.2em]right: \scriptsize  $     $} ] (topleft) {};
			\node at (-3,-3)  [label= {[label distance=-0.2em]right: \scriptsize  $     $} ] (botleft) {};
			\node at (3,3)  [label= {[label distance=-0.2em]right: \scriptsize  $     $} ] (topright) {};
      \node at (3,-3)  [label= {[label distance=-0.2em]right: \scriptsize  $     $} ] (botright) {};
			\draw[kernel1,color=blue] (center) to
			node [sloped,below] {\small }     (topright);
      \draw[kernel1,color=blue] (center) to
			node [sloped,below] {\small }     (topleft);
      \draw[kernel1,color=blue] (center) to
			node [sloped,below] {\small }     (botright);
      \draw[kernel1,color=blue] (center) to
			node [sloped,below] {\small }     (botleft);
      \node at (1.2,1.2)  [circle,fill=white, scale=0.5] {$1$};
    \end{tikzpicture}\begin{tikzpicture}[scale=0.15,baseline=-0.1cm, trim left=-0.4cm, trim right=0.8cm]
			\node at (0,0)  [dot,label= {[label distance=-0.2em]below: \scriptsize  $      $} ] (center) {};
			\node at (-3,3)  [label= {[label distance=-0.2em]right: \scriptsize  $     $} ] (topleft) {};
			\node at (-3,-3)  [label= {[label distance=-0.2em]right: \scriptsize  $     $} ] (botleft) {};
      \node at (4,0)  [dot,label= {[label distance=-0.2em]right: \scriptsize  $     $} ] (centerright) {};
      \node at (7,3)  [label= {[label distance=-0.2em]right: \scriptsize  $     $} ] (leg1) {};
      \node at (7,-3)  [label= {[label distance=-0.2em]right: \scriptsize  $     $} ] (leg2) {};
      \draw[arrow, color=red] (centerright) .. controls (3,1) and (1,1) .. (center);
      \draw[arrow] (center) .. controls (1,-1) and (3,-1) .. (centerright);
      \draw[kernel1,color=blue] (center) to
			node [sloped,below] {\small }     (topleft);
      \draw[kernel1,color=blue] (center) to
			node [sloped,below] {\small }     (botleft);
      \draw[kernel1,color=blue] (centerright) to
			node [sloped,below] {\small }     (leg1);
      \draw[kernel1,color=blue] (centerright) to
			node [sloped,below] {\small }     (leg2);
    \end{tikzpicture}+\begin{tikzpicture}[scale=0.15,baseline=-0.1cm, trim right=1.1cm, trim left=-0.5cm]
			\node at (0,0)  [dot,label= {[label distance=-0.2em]below: \scriptsize  $      $} ] (center) {};
			\node at (-4,0)  [label= {[label distance=-0.2em]right: \scriptsize  $     $} ] (topleft) {};
			\node at (0,4)  [label= {[label distance=-0.2em]right: \scriptsize  $     $} ] (botleft) {};
			\node at (0,-4)  [label= {[label distance=-0.2em]right: \scriptsize  $     $} ] (topright) {};
      \node at (4,0)  [dot,label= {[label distance=-0.2em]right: \scriptsize  $     $} ] (centerright) {};
      \node at (4,4)  [label= {[label distance=-0.2em]right: \scriptsize  $     $} ] (leg1) {};
      \node at (8,0)  [label= {[label distance=-0.2em]right: \scriptsize  $     $} ] (leg2) {};
      \node at (4,-4)  [label= {[label distance=-0.2em]right: \scriptsize  $     $} ] (leg3) {};
      \draw[arrow,color=red] (center) to
			node [sloped,below] {\small }     (centerright);
			\draw[kernel1,color=blue] (center) to
			node [sloped,below] {\small }     (topright);
      \draw[kernel1,color=blue] (center) to
			node [sloped,below] {\small }     (topleft);
      \draw[kernel1,color=blue] (center) to
			node [sloped,below] {\small }     (botleft);
      \draw[kernel1,color=blue] (centerright) to
			node [sloped,below] {\small }     (leg1);
      \draw[kernel1,color=blue] (centerright) to
			node [sloped,below] {\small }     (leg2);
      \draw[kernel1,color=blue] (centerright) to
			node [sloped,below] {\small }     (leg3);
      \node at (3,1.25)  [fill=white, scale=0.5] {$1$};
      \node at (1,1.25)  [fill=white, scale=0.5] {$1$};
    \end{tikzpicture}\begin{tikzpicture}[scale=0.15,baseline=-0.1cm,trim right=1.1cm, trim left=-0.5cm]
			\node at (0,0)  [dot,label= {[label distance=-0.2em]below: \scriptsize  $      $} ] (center) {};
			\node at (-4,0)  [label= {[label distance=-0.2em]right: \scriptsize  $     $} ] (topleft) {};
			\node at (0,4)  [label= {[label distance=-0.2em]right: \scriptsize  $     $} ] (botleft) {};
			\node at (0,-4)  [label= {[label distance=-0.2em]right: \scriptsize  $     $} ] (topright) {};
      \node at (4,0)  [dot,label= {[label distance=-0.2em]right: \scriptsize  $     $} ] (centerright) {};
      \node at (4,4)  [label= {[label distance=-0.2em]right: \scriptsize  $     $} ] (leg1) {};
      \node at (8,0)  [label= {[label distance=-0.2em]right: \scriptsize  $     $} ] (leg2) {};
      \node at (4,-4)  [label= {[label distance=-0.2em]right: \scriptsize  $     $} ] (leg3) {};
      \draw[arrow] (center) to
			node [sloped,below] {\small }     (centerright);
			\draw[kernel1,color=blue] (center) to
			node [sloped,below] {\small }     (topright);
      \draw[kernel1,color=blue] (center) to
			node [sloped,below] {\small }     (topleft);
      \draw[kernel1,color=blue] (center) to
			node [sloped,below] {\small }     (botleft);
      \draw[kernel1,color=blue] (centerright) to
			node [sloped,below] {\small }     (leg1);
      \draw[kernel1,color=blue] (centerright) to
			node [sloped,below] {\small }     (leg2);
      \draw[kernel1,color=blue] (centerright) to
			node [sloped,below] {\small }     (leg3);
    \end{tikzpicture}\,.
    \end{equs}
    For vector-valued theories, this operation would read, for example
    \begin{equs}
      D\begin{tikzpicture}[scale=0.2,baseline=-0.1cm, trim right=1.4cm, trim left=-0.8cm]
        \node at (0,0)  [dot,label= {[label distance=-0.2em]below: \scriptsize  $      $} ] (center) {};
        \node at (-4,0)  [label= {[label distance=-0.2em]right: \scriptsize  $     $} ] (topleft) {};
        \node at (0,4)  [label= {[label distance=-0.2em]right: \scriptsize  $     $} ] (botleft) {};
        \node at (0,-4)  [label= {[label distance=-0.2em]right: \scriptsize  $     $} ] (topright) {};
        \node at (4,0)  [dot,label= {[label distance=-0.2em]right: \scriptsize  $     $} ] (centerright) {};
        \node at (4,4)  [label= {[label distance=-0.2em]right: \scriptsize  $     $} ] (leg1) {};
        \node at (8,0)  [label= {[label distance=-0.2em]right: \scriptsize  $     $} ] (leg2) {};
        \node at (4,-4)  [label= {[label distance=-0.2em]right: \scriptsize  $     $} ] (leg3) {};
        \draw[arrow,color=red] (center) to
        node [sloped,below] {\small }     (centerright);
        \draw[kernel1,color=blue] (center) to
        node [sloped,below] {\small }     (topright);
        \draw[kernel1,color=blue] (center) to
        node [sloped,below] {\small }     (topleft);
        \draw[kernel1,color=blue] (center) to
        node [sloped,below] {\small }     (botleft);
        \draw[kernel1,color=blue] (centerright) to
        node [sloped,below] {\small }     (leg1);
        \draw[kernel1,color=blue] (centerright) to
        node [sloped,below] {\small }     (leg2);
        \draw[kernel1,color=blue] (centerright) to
        node [sloped,below] {\small }     (leg3);
        \node at (0,1.7)  [circle,fill=white, scale=0.5] {$\mathcolor{purple}{a}$};
        \node at (-1.7,0)  [circle,fill=white, scale=0.5] {$\mathcolor{purple}{b}$};
        \node at (0,-1.7)  [circle,fill=white, scale=0.5] {$\mathcolor{purple}{c}$};
        \node at (4,-1.7)  [circle,fill=white, scale=0.5] {$\mathcolor{purple}{d}$};
        \node at (5.7,0)  [circle,fill=white, scale=0.5] {$\mathcolor{purple}{e}$};
        \node at (4,1.7)  [circle,fill=white, scale=0.5] {$\mathcolor{purple}{f}$};
      \end{tikzpicture}=\sum_{k=1}^N\begin{tikzpicture}[scale=0.2,baseline=-0.1cm, trim right=0.6cm, trim left=-0.6cm]
			\node at (0,0)  [dot,label= {[label distance=-0.2em]below: \scriptsize  $      $} ] (center) {};
			\node at (-3,3)  [label= {[label distance=-0.2em]right: \scriptsize  $     $} ] (topleft) {};
			\node at (-3,-3)  [label= {[label distance=-0.2em]right: \scriptsize  $     $} ] (botleft) {};
			\node at (3,3)  [label= {[label distance=-0.2em]right: \scriptsize  $     $} ] (topright) {};
      \node at (3,-3)  [label= {[label distance=-0.2em]right: \scriptsize  $     $} ] (botright) {};
			\draw[kernel1,color=blue] (center) to
			node [sloped,below] {\small }     (topright);
      \draw[kernel1,color=blue] (center) to
			node [sloped,below] {\small }     (topleft);
      \draw[kernel1,color=blue] (center) to
			node [sloped,below] {\small }     (botright);
      \draw[kernel1,color=blue] (center) to
			node [sloped,below] {\small }     (botleft);
      \node at (1.2,1.2)  [circle,fill=white, scale=0.5] {$\mathcolor{purple}{k}$};
      \node at (-1.2,1.2)  [circle,fill=white, scale=0.5] {$\mathcolor{purple}{a}$};
      \node at (-1.2,-1.2)  [circle,fill=white, scale=0.5] {$\mathcolor{purple}{b}$};
      \node at (1.2,-1.2)  [circle,fill=white, scale=0.5] {$\mathcolor{purple}{c}$};
    \end{tikzpicture}\begin{tikzpicture}[scale=0.2,baseline=-0.1cm, trim right=0.6cm, trim left=-0.6cm]
			\node at (0,0)  [dot,label= {[label distance=-0.2em]below: \scriptsize  $      $} ] (center) {};
			\node at (-3,3)  [label= {[label distance=-0.2em]right: \scriptsize  $     $} ] (topleft) {};
			\node at (-3,-3)  [label= {[label distance=-0.2em]right: \scriptsize  $     $} ] (botleft) {};
			\node at (3,3)  [label= {[label distance=-0.2em]right: \scriptsize  $     $} ] (topright) {};
      \node at (3,-3)  [label= {[label distance=-0.2em]right: \scriptsize  $     $} ] (botright) {};
			\draw[kernel1,color=blue] (center) to
			node [sloped,below] {\small }     (topright);
      \draw[kernel1,color=blue] (center) to
			node [sloped,below] {\small }     (topleft);
      \draw[kernel1,color=blue] (center) to
			node [sloped,below] {\small }     (botright);
      \draw[kernel1,color=blue] (center) to
			node [sloped,below] {\small }     (botleft);
      \node at (1.2,1.2)  [circle,fill=white, scale=0.5] {$\mathcolor{purple}{f}$};
      \node at (-1.2,1.2)  [circle,fill=white, scale=0.5] {$\mathcolor{purple}{k}$};
      \node at (-1.2,-1.2)  [circle,fill=white, scale=0.5] {$\mathcolor{purple}{d}$};
      \node at (1.2,-1.2)  [circle,fill=white, scale=0.5] {$\mathcolor{purple}{e}$};
    \end{tikzpicture},
    \end{equs}
    where we have this time denoted the $\mft$ decoration in purple.
  \end{example}
  Furthermore we introduce a notion of inner product for forest that reads on $\mfF\otimes\mfF$
  \begin{equation}\label{eq:innerprodforest}
    \langle\Gamma_1^1\dots\Gamma_n^1,\Gamma_1^2\dots\Gamma_m^2\rangle=\one_{n=m}\sum_{\pi\in\mfS_n}\prod_{k=1}^n\langle\Gamma_k^1,\Gamma_{\pi(k)}^2\rangle\,,
  \end{equation}
  and is extended to $\langle\mfF\rangle\otimes\langle\mfF\rangle$ by bilinearity. Note that this definition respects the symmetry of the forest product. Moreover, it corresponds to the usual scalar product on forests seen as not necessarily connected diagrams.
  \begin{remark}
    At this step, the reader might wonder why we require to use forests of diagrams. In the work \cite{BM25}, the authors simply use tensor product of trees, but this is possible only because a cut in a tree always produces two trees. This is not the case for graphs. Indeed, if one takes a 1PR diagram, we have, for instance
    \begin{equs}
      D
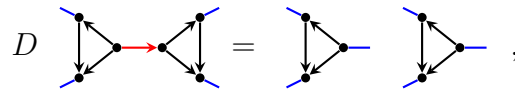
\begin{tikzpicture}[scale=0.2,baseline=0.3cm,trim right=0.3cm]
			\node at (0,0)  [dot,label= {[label distance=-0.2em]below: \scriptsize  $      $} ] (root) {};
			\node at (-2.5,2)  [dot,label= {[label distance=-0.2em]right: \scriptsize  $     $} ] (center) {};
			\node at (0,4)  [dot,label= {[label distance=-0.2em]right: \scriptsize  $     $} ] (centerc) {};
			\node at (2,-1)  [,label= {[label distance=-0.2em]right: \scriptsize  $     $} ] (right) {};
      \node at (2,5)  [,label= {[label distance=-0.2em]right: \scriptsize  $     $} ] (topright) {};
      \node at (-5.5,2)  [dot,label= {[label distance=-0.2em]right: \scriptsize  $     $} ] (centerleft) {};
      \node at (-8,4)  [dot,label= {[label distance=-0.2em]right: \scriptsize  $     $} ] (topleft) {};
			\node at (-8,0)  [dot,label= {[label distance=-0.2em]right: \scriptsize  $     $} ] (botleft) {};
      \node at (-10,-1)  [label= {[label distance=-0.2em]right: \scriptsize  $     $} ] (leg1) {};
      \node at (-10,5)  [label= {[label distance=-0.2em]right: \scriptsize  $     $} ] (leg2) {};
      \draw[arrow] (center) to
			node [sloped,below] {\small }     (root);
			\draw[arrow] (center) to
			node [sloped,below] {\small }     (centerc);
			\draw[arrow] (centerc) to
			node [sloped,below] {\small }     (root);
      \draw[arrow,color=red] (centerleft) to
			node [sloped,below] {\small }     (center);
      \draw[arrow] (centerleft) to
			node [sloped,below] {\small }     (topleft);
      \draw[arrow] (centerleft) to
			node [sloped,below] {\small }     (botleft);
      \draw[arrow] (topleft) to
			node [sloped,below] {\small }     (botleft);
			\draw[kernel1,color=blue] (root) to
			node [sloped,below] {\small }     (right);
      \draw[kernel1,color=blue] (centerc) to
			node [sloped,below] {\small }     (topright);
      \draw[kernel1,color=blue] (topleft) to
			node [sloped,below] {\small }     (leg2);
      \draw[kernel1,color=blue] (botleft) to
			node [sloped,below] {\small }     (leg1);
    \end{tikzpicture}=\begin{tikzpicture}[scale=0.2,baseline=0.3cm,trim right=-0.6cm]
			\node at (-3,2)  [label= {[label distance=-0.2em]right: \scriptsize  $     $} ] (center) {};
      \node at (-5.5,2)  [dot,label= {[label distance=-0.2em]right: \scriptsize  $     $} ] (centerleft) {};
      \node at (-8,4)  [dot,label= {[label distance=-0.2em]right: \scriptsize  $     $} ] (topleft) {};
			\node at (-8,0)  [dot,label= {[label distance=-0.2em]right: \scriptsize  $     $} ] (botleft) {};
      \node at (-10,-1)  [label= {[label distance=-0.2em]right: \scriptsize  $     $} ] (leg1) {};
      \node at (-10,5)  [label= {[label distance=-0.2em]right: \scriptsize  $     $} ] (leg2) {};
      \draw[kernel1,color=blue] (center) to
			node [sloped,below] {\small }     (centerleft);
      \draw[arrow] (centerleft) to
			node [sloped,below] {\small }     (topleft);
      \draw[arrow] (centerleft) to
			node [sloped,below] {\small }     (botleft);
      \draw[arrow] (topleft) to
			node [sloped,below] {\small }     (botleft);
      \draw[kernel1,color=blue] (topleft) to
			node [sloped,below] {\small }     (leg2);
      \draw[kernel1,color=blue] (botleft) to
			node [sloped,below] {\small }     (leg1);
    \end{tikzpicture}\begin{tikzpicture}[scale=0.2,baseline=0.3cm,trim right=-0.5cm]
			\node at (-3,2)  [label= {[label distance=-0.2em]right: \scriptsize  $     $} ] (center) {};
      \node at (-5.5,2)  [dot,label= {[label distance=-0.2em]right: \scriptsize  $     $} ] (centerleft) {};
      \node at (-8,4)  [dot,label= {[label distance=-0.2em]right: \scriptsize  $     $} ] (topleft) {};
			\node at (-8,0)  [dot,label= {[label distance=-0.2em]right: \scriptsize  $     $} ] (botleft) {};
      \node at (-10,-1)  [label= {[label distance=-0.2em]right: \scriptsize  $     $} ] (leg1) {};
      \node at (-10,5)  [label= {[label distance=-0.2em]right: \scriptsize  $     $} ] (leg2) {};
      \draw[kernel1,color=blue] (center) to
			node [sloped,below] {\small }     (centerleft);
      \draw[arrow] (centerleft) to
			node [sloped,below] {\small }     (topleft);
      \draw[arrow] (centerleft) to
			node [sloped,below] {\small }     (botleft);
      \draw[arrow] (topleft) to
			node [sloped,below] {\small }     (botleft);
      \draw[kernel1,color=blue] (topleft) to
			node [sloped,below] {\small }     (leg2);
      \draw[kernel1,color=blue] (botleft) to
			node [sloped,below] {\small }     (leg1);
    \end{tikzpicture}\,,
    \end{equs}
    where for a 1PI diagram, the graph remains connected
    \begin{equs}
      D
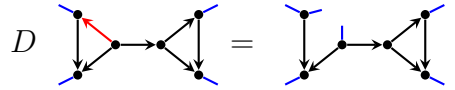
\begin{tikzpicture}[scale=0.2,baseline=0.3cm,trim right=0.3cm]
			\node at (0,0)  [dot,label= {[label distance=-0.2em]below: \scriptsize  $      $} ] (root) {};
			\node at (-2.5,2)  [dot,label= {[label distance=-0.2em]right: \scriptsize  $     $} ] (center) {};
			\node at (0,4)  [dot,label= {[label distance=-0.2em]right: \scriptsize  $     $} ] (centerc) {};
			\node at (2,-1)  [,label= {[label distance=-0.2em]right: \scriptsize  $     $} ] (right) {};
      \node at (2,5)  [,label= {[label distance=-0.2em]right: \scriptsize  $     $} ] (topright) {};
      \node at (-5.5,2)  [dot,label= {[label distance=-0.2em]right: \scriptsize  $     $} ] (centerleft) {};
      \node at (-8,4)  [dot,label= {[label distance=-0.2em]right: \scriptsize  $     $} ] (topleft) {};
			\node at (-8,0)  [dot,label= {[label distance=-0.2em]right: \scriptsize  $     $} ] (botleft) {};
      \node at (-10,-1)  [label= {[label distance=-0.2em]right: \scriptsize  $     $} ] (leg1) {};
      \node at (-10,5)  [label= {[label distance=-0.2em]right: \scriptsize  $     $} ] (leg2) {};
      \draw[arrow] (center) to
			node [sloped,below] {\small }     (root);
			\draw[arrow] (center) to
			node [sloped,below] {\small }     (centerc);
			\draw[arrow] (centerc) to
			node [sloped,below] {\small }     (root);
      \draw[arrow] (centerleft) to
			node [sloped,below] {\small }     (center);
      \draw[arrow,color=red] (centerleft) to
			node [sloped,below] {\small }     (topleft);
      \draw[arrow] (centerleft) to
			node [sloped,below] {\small }     (botleft);
      \draw[arrow] (topleft) to
			node [sloped,below] {\small }     (botleft);
			\draw[kernel1,color=blue] (root) to
			node [sloped,below] {\small }     (right);
      \draw[kernel1,color=blue] (centerc) to
			node [sloped,below] {\small }     (topright);
      \draw[kernel1,color=blue] (topleft) to
			node [sloped,below] {\small }     (leg2);
      \draw[kernel1,color=blue] (botleft) to
			node [sloped,below] {\small }     (leg1);
    \end{tikzpicture}=\begin{tikzpicture}[scale=0.2,baseline=0.3cm,trim right=0.3cm]
			\node at (0,0)  [dot,label= {[label distance=-0.2em]below: \scriptsize  $      $} ] (root) {};
			\node at (-2.5,2)  [dot,label= {[label distance=-0.2em]right: \scriptsize  $     $} ] (center) {};
			\node at (0,4)  [dot,label= {[label distance=-0.2em]right: \scriptsize  $     $} ] (centerc) {};
			\node at (2,-1)  [,label= {[label distance=-0.2em]right: \scriptsize  $     $} ] (right) {};
      \node at (2,5)  [,label= {[label distance=-0.2em]right: \scriptsize  $     $} ] (topright) {};
      \node at (-5.5,2)  [dot,label= {[label distance=-0.2em]right: \scriptsize  $     $} ] (centerleft) {};
      \node at (-8,4)  [dot,label= {[label distance=-0.2em]right: \scriptsize  $     $} ] (topleft) {};
			\node at (-8,0)  [dot,label= {[label distance=-0.2em]right: \scriptsize  $     $} ] (botleft) {};
      \node at (-10,-1)  [label= {[label distance=-0.2em]right: \scriptsize  $     $} ] (leg1) {};
      \node at (-10,5)  [label= {[label distance=-0.2em]right: \scriptsize  $     $} ] (leg2) {};
      \node at (-6.1,4.5)  [label= {[label distance=-0.2em]right: \scriptsize  $     $} ] (leg3) {};
      \node at (-5.5,4)  [label= {[label distance=-0.2em]right: \scriptsize  $     $} ] (leg4) {};
      \draw[arrow] (center) to
			node [sloped,below] {\small }     (root);
			\draw[arrow] (center) to
			node [sloped,below] {\small }     (centerc);
			\draw[arrow] (centerc) to
			node [sloped,below] {\small }     (root);
      \draw[arrow] (centerleft) to
			node [sloped,below] {\small }     (center);
      \draw[arrow] (centerleft) to
			node [sloped,below] {\small }     (botleft);
      \draw[arrow] (topleft) to
			node [sloped,below] {\small }     (botleft);
			\draw[kernel1,color=blue] (root) to
			node [sloped,below] {\small }     (right);
      \draw[kernel1,color=blue] (centerc) to
			node [sloped,below] {\small }     (topright);
      \draw[kernel1,color=blue] (topleft) to
			node [sloped,below] {\small }     (leg2);
      \draw[kernel1,color=blue] (botleft) to
			node [sloped,below] {\small }     (leg1);
      \draw[kernel1,color=blue] (topleft) to
			node [sloped,below] {\small }     (leg3);
      \draw[kernel1,color=blue] (centerleft) to
			node [sloped,below] {\small }     (leg4);
    \end{tikzpicture}\,.
    \end{equs}
  \end{remark}
  
  We are now ready to state and prove a duality relation between these two operators. It will be the cornerstone, on the combinatorial side, of the proof of the main theorem \ref{maintheorem}. The introduction of forests allows us to write it under a quite elegant form.
  \begin{lemma}\label{lemma:duality}
    The following duality relation holds on $\langle\mfF\rangle\otimes\langle\mfF\rangle$
    \begin{equation}
      \langle L\mathbf{\Gamma}^1,\mathbf{\Gamma}^2\rangle=\langle\mathbf{\Gamma}^1,D\mathbf{\Gamma}^2\rangle\,.
    \end{equation}
  \end{lemma}

\begin{proof}
  Let $n,m\in\N$. We set $\mcG(n,m)$ the set of (not necessarily connected) labelled diagrams with $n$ vertices and $m$ half-edges. We give two examples of graphs in $\mcG(2,6)$ 
  \begin{equs}
.
  \end{equs}
  In the picture above, the half-legs are represented in orange. The linear maps $L$ and $D$ have a natural counterpart on this set that respects the labelling. For example, we have
  \begin{equs}
    D%
  \end{equs}
  $\langle\mcG(n,m)\rangle$ is endowed with an inner product that reads on $\mcG(n,m)\otimes\mcG(n,m)$
  \begin{equs}
    \langle\mathbf\Gamma_1,\mathbf\Gamma_2\rangle=\mathbf{1}_{\mathbf\Gamma_1=\mathbf\Gamma_2}\,,
  \end{equs}
  and is extended by bilinearity. In addition, $\mcG(n,m)$ is also endowed with a group action of $\mfS_n\times\mfS_m$. At this step, it is clear that the duality relation
  \begin{equs} \label{Adjoint_L_D}
    \langle L\mathbf{\Gamma}_1,\mathbf{\Gamma}_2\rangle=\langle\mathbf{\Gamma}_1,D\mathbf{\Gamma}_2\rangle
  \end{equs}
  holds on $\langle\mcG(n,m)\rangle$, and that $L$ and $D$ are equivariant, \ie, for $\sigma\in\mfS_n\times\mfS_m$,
  \begin{equs} \label{commutation_permutation}
    \sigma\cdot L\mathbf{\Gamma}=L(\sigma\cdot\mathbf{\Gamma}),\quad\mathrm{and}\quad\sigma\cdot D\mathbf{\Gamma}=D(\sigma\cdot\mathbf{\Gamma})\,.
  \end{equs}
  Indeed, the maps $ L $ and $D$ do not change the labels on the nodes and the half-edges. They connect or disconnect two half-edges which is compatible with the action of a permutation $\sigma\in\mfS_n\times\mfS_m$. 
  Let us define $\mfG(n,m)\subset\mfF$ the subset of (unlabelled) forests with $n$ vertices and $m$ half-edges. There is a naturally covariant map $[\bigcdot]\colon\mcG(n,m)\to\mfG(n,m)$ that associates to $\mathbf{\Gamma}$ its orbit in $\mfG(n,m)$. It can also be seen as an unlabelling map. The inner product defined in \eqref{eq:innerprodforest} also reads
  \begin{equs}
    \langle  [\mathbf{\Gamma}_1],[\mathbf{\Gamma}_2]\rangle=\sum_{\sigma\in\mfS_n\times\mfS_m}\langle\sigma\cdot\mathbf{\Gamma}_1,\mathbf{\Gamma}_2\rangle=\sum_{\sigma\in\mfS_n\times\mfS_m}\langle\mathbf{\Gamma}_1,\sigma\cdot\mathbf{\Gamma}_2\rangle\,.
  \end{equs}
  It is also easy to see that
  \begin{equs}
    L([\mathbf\Gamma])=[L\mathbf\Gamma]\quad\mathrm{and}\quad D([\mathbf\Gamma])=[D\mathbf\Gamma]\,.
  \end{equs}
 for the same reasons as for \eqref{commutation_permutation}.
  We are now ready to conclude the proof. We have
  \begin{equs}
    \langle L([\mathbf\Gamma_1]),[\mathbf\Gamma_2]\rangle&=\langle [L\mathbf\Gamma_1],[\mathbf\Gamma_2]\rangle=\sum_{\sigma\in\mfS_n\times\mfS_m}\langle L\mathbf\Gamma_1,\sigma\cdot\mathbf\Gamma_2\rangle\\
    &=\sum_{\sigma\in\mfS_n\times\mfS_m}\langle \mathbf\Gamma_1,D(\sigma\cdot\mathbf\Gamma_2)\rangle=\sum_{\sigma\in\mfS_n\times\mfS_m}\langle \mathbf\Gamma_1,\sigma\cdot D\mathbf\Gamma_2\rangle\\
    &=\sum_{\sigma\in\mfS_n\times\mfS_m}\langle \sigma\cdot\mathbf\Gamma_1, D\mathbf\Gamma_2\rangle=\langle[\mathbf\Gamma_1],[D\mathbf\Gamma_2]\rangle\\
    &=\langle[\mathbf\Gamma_1],D[\mathbf\Gamma_2]\rangle\,.
  \end{equs}
\end{proof}

\begin{remark} The previous proof is inspired by the species theory developed by Joyal \cite{J81}. Species were already considered in the context of QFT in \cite{Faris} where the authors seems to consider  similar operations to $ \circlearrowleft $ and $ \graft $ directly at the level of species (see Section $4$ therein).
	\end{remark}

\subsection{Putting everything together}\label{sec:lastsection}

We can eventually prove in this last section the main theorem \ref{maintheorem}. We first start with an elementary lemma.

\begin{lemma}\label{lemma:coeffuparrow}
  Let $\Gamma\in F_0$. There is a collection of different $\Gamma_i$, all of the form $\uparrow^e\Gamma$ for $e\in\CE_\star$ such that
  \begin{equation}
    \uparrow\Gamma=\sum_i\frac{S(\Gamma)}{S(\Gamma_i)}\Gamma_i.
  \end{equation}
\end{lemma}

\begin{proof}
We denote for this proof $\mathrm{eAut}(\Gamma)$ the edge automorphisms of $\Gamma$. This is a subgroup of $\mathrm{Aut}(\Gamma)$. For an edge $e\in\mcE_\star(\Gamma)$, we denote its orbit with respect to this group $\mathrm{Orb}(e)$. Moreover, we denote $\mathrm{Stab}(e,\Gamma)$ and $\mathrm{eStab}(e,\Gamma)$ the stabilizers of $e$ in $\Gamma$ with respect to $\mathrm{Aut}(\Gamma)$ and $\mathrm{eAut}(\Gamma)$ respectively. We have, using the orbit-stabilizer theorem,
\begin{equs}
\uparrow\Gamma&=\sum_{e\in\CE_\star(\Gamma)}\uparrow^e\Gamma=\sum_{e\in\CE_\star(\Gamma)/\mathrm{eAut}(\Gamma)}|\mathrm{Orb}(e)|\uparrow^e\Gamma\\
&=\sum_{e\in\CE_\star(\Gamma)/\mathrm{eAut}(\Gamma)}\frac{|\mathrm{eAut}(\Gamma)|}{|\mathrm{eStab}(e,\Gamma)|}\uparrow^e\Gamma\\
&=\sum_{e\in\CE_\star(\Gamma)/\mathrm{eAut}(\Gamma)}\frac{|\mathrm{Aut}(\Gamma)|}{|\mathrm{Stab}(e,\Gamma)|}\uparrow^e\Gamma\\
&=\sum_{e\in\CE_\star(\Gamma)/\mathrm{eAut}(\Gamma)}\frac{S(\Gamma)}{S(\uparrow^e\Gamma)}\uparrow^e\Gamma.
\end{equs}
In all the sums above, $\uparrow^e\Gamma$ should be understood as a representative in the equivalence class of edge automorphisms. 
\end{proof}

\begin{proof}[of Theorem \ref{maintheorem}]
Using the Lemmas \ref{commut}, \ref{prelie1}, \ref{prelie2}, we have
\begin{multline*}
  \D_\lambda V_\lambda[\phi]+\frac{1}{2}\mathrm{Tr}(\dot G_\lambda \mathrm{D}^2 V_\lambda)[\phi]+\frac{1}{2}(\mathrm{D} V_\lambda\dot G_\lambda\mathrm{D}V_\lambda)[\phi]=\sum_{\Gamma\in F^\star}\frac{\alpha(\Gamma)\Gamma!}{S(\Gamma)}(\hat\Pi_\lambda\uparrow\Gamma)[\phi]\\
  +\frac{1}{2}\sum_{\Gamma\in F^\star}\frac{\alpha(\Gamma)\Gamma!}{S(\Gamma)}(\hat\Pi_\lambda\circlearrowleft\Gamma)[\phi]+\frac{1}{2}\sum_{\Gamma_1,\Gamma_2\in F^\star}\frac{\alpha(\Gamma_1)\Gamma_1!\alpha(\Gamma_2)\Gamma_2!}{S(\Gamma_1)S(\Gamma_2)}\hat\Pi_\lambda(\Gamma_1\rightarrow\Gamma_2)[\phi]\,.
\end{multline*}
It suffices to show the formal equality
\begin{equation}\label{eq:formaleq}
  \begin{split}
  \sum_{\Gamma\in F^\star}\frac{\alpha(\Gamma)\Gamma!}{S(\Gamma)}\uparrow\Gamma&=\frac{1}{2}\sum_{\Gamma\in F^\star}\frac{\alpha(\Gamma)\Gamma!}{S(\Gamma)}\circlearrowleft\Gamma+\frac{1}{2}\sum_{\Gamma_1,\Gamma_2\in F^\star}\frac{\alpha(\Gamma_1)\Gamma_1!\alpha(\Gamma_2)\Gamma_2!}{S(\Gamma_1)S(\Gamma_2)}\Gamma_1\rightarrow\Gamma_2\\
  &=\frac{1}{2}\sum_{\Gamma\in F^\star}\frac{\alpha(\Gamma)\Gamma!}{S(\Gamma)}L\Gamma+\frac{1}{2}\sum_{\Gamma_1,\Gamma_2\in F^\star}\frac{\alpha(\Gamma_1)\Gamma_1!\alpha(\Gamma_2)\Gamma_2!}{S(\Gamma_1)S(\Gamma_2)}L(\Gamma_1\Gamma_2)
  \end{split}
\end{equation}
Let $e\in \mcE_\star(\Gamma)$. We have on the left-hand side, using Lemma \ref{lemma:coeffuparrow}
\begin{equs}
  \frac{\alpha(\Gamma)\Gamma!}{S(\Gamma)}\langle\uparrow\Gamma,\uparrow^e\Gamma\rangle&=\left\langle\frac{\alpha(\Gamma)\Gamma!}{S(\Gamma)}\sum_{e\in\mcE_\star}\frac{S(\Gamma)}{S(\uparrow^e\Gamma)}\uparrow^e\Gamma,\uparrow^e\Gamma\right\rangle\\
  &=\frac{\alpha(\Gamma)\Gamma!}{S(\uparrow^e\Gamma)}\langle\uparrow^e\Gamma,\uparrow^e\Gamma\rangle\\
  &=\alpha(\Gamma)\Gamma!\,.
\end{equs}
For the first term on the right-hand side, we have
\begin{equs}
  \frac{\alpha(\Gamma)\Gamma!}{S(\Gamma)}\left\langle L\Gamma,\uparrow^e\Gamma\right\rangle=\frac{\alpha(\Gamma)\Gamma!}{S(\Gamma)}\left\langle\Gamma,D\uparrow^e\Gamma\right\rangle\,.
\end{equs}
The scalar product does not vanish only if $D$ cuts $\uparrow^e\Gamma$ in a connected diagram (\ie $\Gamma$ is 1PI), in which case the quantity above equals $\alpha(\Gamma)\Gamma!$. Moreover, there are two terms in the sum for which the inner product with the graph given by $\uparrow^e\Gamma$ does not vanish since the cut of $D$ does not see the direction of the edges, thus cancelling the $1/2$ factor. For the second term, we carry out a similar computation.
\begin{equs}
  \frac{1}{S(\Gamma_1)S(\Gamma_2)}\langle L(\Gamma_1\Gamma_2),\uparrow^e\Gamma\rangle=\frac{1}{S(\Gamma_1)S(\Gamma_2)}\langle \Gamma_1\Gamma_2,D\uparrow^e\Gamma\rangle\,.
\end{equs}
The scalar product does not vanish only if $D$ cuts $\Gamma$ into two disconnected components $\gamma_1$ and $\gamma_2$ (\ie $\Gamma$ is 1PR). From there, two cases are possible. The first one is that the two diagrams are equal, in which case the quantity above is $2$, because $\langle\gamma\gamma,\gamma\gamma\rangle=2S(\gamma)^2$. The second one is that the two diagrams are not equal, in which case the result is $1$. In that case, the factor $1/2$  is cancelled by the fact that $\gamma_1\gamma_2$ appears once as $\gamma_1\gamma_2$ and once as $\gamma_2\gamma_1$ in the double sum. We can conclude the proof by noticing that for such a decomposition, we have 
\begin{equs}
  \Gamma_1 !  \Gamma_2 !  =  \prod_{v_1\in\mcV_{\star,1}}S\big(\mathrm{elem}(v_1)\big) \prod_{v_2\in\mcV_{\star,2}}S\big(\mathrm{elem}(v_2)\big)=\prod_{v\in\mcV_\star}S\big(\mathrm{elem}(v)\big)\,.
\end{equs}
and
\begin{equs}
  \alpha(\Gamma_1)\alpha(\Gamma_2)  =  \prod_{v_1\in\mcV_{\star,1}}\alpha\big(\mathrm{elem}(v_1)\big) \prod_{v_2\in\mcV_{\star,2}}\alpha\big(\mathrm{elem}(v_2)\big)=\prod_{v\in\mcV_\star}\alpha\big(\mathrm{elem}(v)\big)\,.
\end{equs}
\end{proof}

\begin{remark}
  If one follows this proof naively, it would be easy to think that the coefficients in the ansatz are only $1/S(\Gamma)$. However, this is trivially false if one runs the induction with the first few terms. The way to correct that is to add an additional factor to the elementary diagrams and propagate this to all the diagrams.
\end{remark}

\begin{proof}[of Corollary \ref{corollary:renormmeasure}]
  The proof follows the one of Theorem 5.5 in \cite{BM25}. Let $\Gamma\in F_0$. Let us write in Sweedler's notation $\Delta\Gamma=\sum_{(\Gamma)}\Gamma_1\otimes\Gamma_2$. The backbone of the proof is to identify the couples $(\Gamma_1,\Gamma_2)$ in $\Delta\Gamma$ such that $\Pi_0\Gamma_2\neq0$. Since $G_0=G$, the only diagrams that can satisfy this are the pure legs diagrams, that are either the elementary diagrams, or the renormalising diagrams that come at each step of the induction. More precisely, we have
  \begin{equs}
    (\hat\Pi_0\Gamma)[\phi]=\sum_{\ell}\frac{1}{\ell!}g_{\scriptscriptstyle\mathrm{BPHZ}}(\mathrm{vac}(\Gamma),\pi\ell)\big(\Pi_0\mathrm{res}_\ell(\Gamma)\big)[\phi]\,.
  \end{equs}
  This directly gives
  \begin{equs}
    V_0[\phi]=\sum_{\deg(\Gamma)+|\pi\ell|\leq0}\frac{\alpha(\Gamma)\Gamma!g_{\scriptscriptstyle\mathrm{BPHZ}}(\mathrm{vac}(\Gamma),\pi\ell)}{S(\Gamma)}\frac{1}{\ell!}\big(\Pi_0\mathrm{res}_\ell(\Gamma)\big)[\phi]\,,
  \end{equs}
  which in turn directly yields the result with a simple resummation.
\end{proof}


\begin{thebibliography}{99}
	\expandafter\ifx\csname url\endcsname\relax
	\def\url#1{\texttt{#1}}\fi
	\expandafter\ifx\csname urlprefix\endcsname\relax\def\urlprefix{URL }\fi
	\expandafter\ifx\csname href\endcsname\relax
	\def\href#1#2{#2}\fi
	\expandafter\ifx\csname burlalt\endcsname\relax
	\def\burlalt#1#2{\href{#2}{\texttt{#1}}}\fi
	
	
	\bibitem{pre_FD}
	A.~Abdesselam. {\em Feynman diagrams in algebraic combinatorics.} Semin. Lothar. Comb.
	\textbf{B49c}, (2002), 45 p., electronic only.
	\burlalt{http://eudml.org/doc/123420}{http://eudml.org/doc/123420}.

  \bibitem{BB21}
  R.~{Bauerschmidt}, T.~{Bodineau}. \newblock {\em Log-Sobolev inequality for the continuum sine-Gordon model}. 
  \newblock Comm. Pure Appl. Math., \textbf{74}, no.~10, (2021), 2064--2113.
  \burlalt{doi:10.1002/cpa.21926}{https://www.doi.org/10.1002/cpa.21926}. 

  \bibitem{BBD24}
  R.~{Bauerschmidt}, T.~{Bodineau}, B.~{Dagallier}. \newblock {\em Stochastic dynamics and the Polchinski equation: An introduction}. 
  \newblock Probab. Surveys, \textbf{21}, (2024), 200--290.
  \burlalt{doi: 10.1214/24-PS27}{https://www.doi.org/10.1214/24-PS27}.
	
	\bibitem{BP57}
	N.~N. Bogoliubow, O.~S. Parasiuk.
	\newblock { \em \"{U}ber die {M}ultiplikation der {K}ausalfunktionen in der
		{Q}uantentheorie der {F}elder.}
	\newblock Acta Math. \textbf{97}, (1957), 227--266.
	\newblock
	\burlalt{doi:10.1007/BF02392399}{http://dx.doi.org/10.1007/BF02392399}.
	
	\bibitem{BCCH}
	{ \rm Y. Bruned, A. Chandra, I. Chevyrev,
		M. Hairer},
	\newblock {\em Renormalising SPDEs in regularity structures}.
	\newblock J. Eur. Math. Soc. (JEMS), \textbf{23}, no.~3, (2021), 869-947.
	\newblock
	\burlalt{doi:10.4171/JEMS/1025}{http://dx.doi.org/10.4171/JEMS/1025}.
	
		\bibitem{BH25}
	Y.~Bruned, Y.~Hou.
	\newblock {\textsl{	Renormalising Feynman diagrams with multi-indices.}}
	\newblock \burlalt{arXiv:2501.08151}{http://arxiv.org/abs/2501.08151}.
	

	
	\bibitem{BHZ}
	{\rm Y. Bruned, M. Hairer, L. Zambotti}.
	\newblock {\em Algebraic renormalisation of regularity structures.}
	\newblock Invent. Math. \textbf{215}, no.~3, (2019), 1039--1156.
	\newblock
	\burlalt{doi:10.1007/s00222-018-0841-x}{https://dx.doi.org/10.1007/s00222-018-0841-x}.
	
	\bibitem{CH16}
	A.~Chandra, M.~Hairer.
	\newblock {\textsl{An analytic {BPHZ} theorem for regularity structures.}}
	\newblock \burlalt{arXiv:1612.08138}{http://arxiv.org/abs/1612.08138}.
	

	
		\bibitem{BM25}
	Y.~{Bruned}, A.~{Minguella}.
	\newblock {\em  Renormalisation in the flow approach for singular SPDEs.} To appear in Annals of Probability. \burlalt{arXiv:2504.04885}{https://arxiv.org/abs/2504.04885}.
	
	
	\bibitem{BN23}
	Y.~{Bruned}, U.~{Nadeem}, \newblock {\em Diagram-free approach for convergence of tree-based models in Regularity Structures}. 
	\newblock J. Math. Soc. Japan, \textbf{76}, no.~4, (2024), 1139-1169. 
	\burlalt{doi: 10.2969/jmsj/91129112}{https://www.doi.org/10.2969/jmsj/91129112}. 
	
	
	\bibitem{CF24a}
	A.~Chandra, L.~Ferdinand.
	\newblock{\textsl{A flow approach to the generalized KPZ equation}}
	\newblock\burlalt{arXiv:2402.03101}{https://arxiv.org/abs/2402.03101}.
	
	
	\bibitem{CK1}
	A.~Connes, D.~Kreimer.
	\newblock { \em Hopf algebras, renormalization and noncommutative geometry.}
	\newblock Comm. Math. Phys. \textbf{199}, no.~1, (1998), 203--242.
	\newblock
	\burlalt{doi:10.1007/s002200050499}{http://dx.doi.org/10.1007/s002200050499}.
	
	\bibitem{CK2}
	A.~Connes, D.~Kreimer.
	\newblock { \em Renormalization in quantum field theory and the {R}iemann-{H}ilbert
		problem {I}: the {H}opf algebra structure of graphs and the main theorem.}
	\newblock  Comm. Math. Phys. \textbf{210}, (2000), 249--73.
	\newblock
	\burlalt{doi:10.1007/s002200050779}{http://dx.doi.org/10.1007/s002200050779}.
	
	
		\bibitem{Costello}
	K.~Costello.
	\newblock { \em Renormalization and Effective Field Theory.}
	\newblock  Mathematical Surveys and Monograph \textbf{170}, AMS, (2011), 1--251.
	\newblock
	\burlalt{doi:10.1090/surv/170}{http://dx.doi.org/10.1090/surv/170}.
	
	
	
	\bibitem{Duc21}
	P.~Duch.
	\newblock{\textsl{Flow equation approach to singular stochastic PDEs}}.  Probab. and Math. Phys. \textbf{6}, no.~2, (2025), 327–437. 
	\newblock
	\burlalt{doi:10.2140/pmp.2025.6.327}{https://dx.doi.org/10.2140/pmp.2025.6.327}.
	
	\bibitem{Duc22}
	P.~Duch.
	\newblock{\textsl{Renormalization of singular elliptic stochastic PDEs using flow equation}}. Probab. and Math. Phys. \textbf{6}, no.~1, (2025), 111--138.
	\newblock
	\burlalt{doi:10.2140/pmp.2025.6.111}{https://dx.doi.org/10.2140/pmp.2025.6.111}.

  \bibitem{Duc23}
  P.~Duch.
  \newblock{\textsl{Lecture notes on flow equation approach to singular
  stochastic PDEs}}
  \newblock\burlalt{https://pawelduch.github.io/spde_flow_notes.pdf}.
	
	\bibitem{Faris}
	W.D.~Faris.
	\newblock {\em Combinatorial species and Feynman diagrams.}
	Semin. Lothar. Comb.
	\textbf{61A}, (2011), Article B61An, electronic only.
	\burlalt{http://www.kurims.kyoto-u.ac.jp/EMIS/journals/SLC/wpapers/s61Afaris.pdf}{http://www.kurims.kyoto-u.ac.jp/EMIS/journals/SLC/wpapers/s61Afaris.pdf}.

  \bibitem{FMRV85}
  J.~Feldman, J.~Magnen, V.~Rivasseau, R.~Sénéor. {\em Bounds on renormalized Feynman graphs.}
  \newblock Commun.Math. Phys. \textbf{100}, 23–55 (1985).
  \burlalt{doi:10.1007/BF01212686}{https://dx.doi.org/10.1007/BF01212686}.
	
	\bibitem{GM24}
  M.~Gubinelli, S-J.~Meyer.
  \newblock{\textsl{The FBSDE approach to sine-Gordon up to $6\pi$}}
  \newblock\burlalt{arXiv:2401.13648}{https://arxiv.org/abs/2401.13648}.

	\bibitem{reg}
	{\rm M. Hairer}.
	\newblock {\em A theory of regularity structures.}
	\newblock Invent. Math. \textbf{198}, no.~2, (2014), 269--504.
	\newblock
	\burlalt{doi:10.1007/s00222-014-0505-4}{https://dx.doi.org/10.1007/s00222-014-0505-4}.
	
	\bibitem{BPHZ_theorem}
	{\rm M. Hairer}.
	\newblock {\em An Analyst’s Take on the BPHZ Theorem.} Computation and Combinatorics in Dynamics, Stochastics and Control. Abelsymposium 2016, \textbf{13}, Springer Cham., (2018).
	\newblock
	\burlalt{doi:10.1007/978-3-030-01593-0_16}{https://dx.doi.org/10.1007/978-3-030-01593-0_16}.
	
	\bibitem{KH69}
	K. ~Hepp.
	\newblock {\em On the equivalence of additive and analytic renormalization}.
	\newblock Comm. Math. Phys. \textbf{14}, (1969), 67--69.
	\newblock \burlalt{doi:10.1007/BF01645456}{http://dx.doi.org/10.1007/BF01645456}.
	
  \bibitem{J81}
  A.~Joyal.
  \newblock {\em Une théorie combinatoire des séries formelles.} 
  \newblock Adv. in Math. \textbf{42.1}, (1981), 1--82.
  \newblock \burlalt{doi:10.1016/0001-8708(81)90052-9}{https://doi.org/10.1016/0001-8708(81)90052-9}

  \bibitem{Kop07}
  C.~Kopper. {\em Renormalization theory based on flow equations.}
  \newblock Rigorous Quantum Field Theory: A Festschrift for Jacques Bros. Basel : Birkhäuser Basel, (2007). p. 161-174.
  \newblock \burlalt{doi:10.1007/978-3-7643-7434-1_12}{https://doi.org/10.1007/978-3-7643-7434-1_12}.

  %\bibitem{KM00}
  %C. Kopper, V.F. Müller. {\em Renormalization Proof for Spontaneously Broken Yang-Mills Theory with Flow Equations}
  %\newblock Commun. Math. Phys. \textbf{209} (2000) 477 - 516.
  %\newblock \burlalt{doi:10.1007/s002200050028}{https://doi.org/10.1007/s002200050028}.
	
	\bibitem{K16}
	A.~Kupiainen.
	\newblock {\em Renormalization group and stochastic PDEs}. Ann. Henri.
	Poincaré, \textbf{17}, no.~3, (2016), 497--535.
	\burlalt{doi:10.1007/s00023-015-0408-y}{http://dx.doi.org/10.1007/s00023-015-0408-y}.
	
	\bibitem{KM17}
	A.~Kupiainen, M.~Marcozzi,
	\newblock {\em Renormalization of generalized KPZ equation}. J. Stat. Phys. \textbf{166}, no.~3, (2017), 876--902.
	\burlalt{doi:10.1007/s10955-016-1636-3}{http://dx.doi.org/10.1007/s10955-016-1636-3}.

  \bibitem{Mey26}
  S-J.~Meyer.
  \newblock{\em An FBSDE Construction of the Sine-Gordon EQFT for $\beta^2<\frac{6}{7}8\pi$ and Perturbative Renormalization in the Full Subcritical Regime.}
  \newblock\burlalt{arXiv:2607.20632}{https://arxiv.org/abs/2607.20632}.
	
	\bibitem{M03}
	V. F. Müller, \newblock {\em Perturbative Renormalization by Flow Equation}. 
	\newblock Rev. Math. Phys. \textbf{15}, no.~5, (2003) 491--558.
	\burlalt{doi:10.1142/S0129055X03001692}{https://www.doi.org/10.1142/S0129055X03001692}. 
	
	\bibitem{OSSW}
	F.~Otto, J.~Sauer, S.~Smith, H.~Weber.
	\newblock {\em A priori bounds for quasi-linear SPDEs in the full sub-critical regime}. J. Eur. Math. Soc. (JEMS)  \textbf{27}, no. 1, (2025),  71--118.
	\burlalt{doi:10.4171/JEMS/1574}{http://dx.doi.org/10.4171/JEMS/1574}.
	
	\bibitem{LOTT}
	{\rm P.~Linares, F.~Otto, M.~Tempelmayr, P.~Tsatsoulis}
	\newblock {\em A diagram-free approach to the stochastic estimates in regularity structures.} Invent. Math. \textbf{237}, (2024), 1469--1565.
	\burlalt{doi:10.1007/s00222-024-01275-z}{https://dx.doi.org/10.1007/s00222-024-01275-z}.
	
	\bibitem{P84}
	J. Polchinski. \newblock {\em Renormalization and effective lagrangians}. 
	\newblock Nucl. Phys. B, \textbf{231}, no.~2, (1984) 269--295.
	\burlalt{doi:10.1016/0550-3213(84)90287-6}{https://www.doi.org/10.1016/0550-3213(84)90287-6}.

  \bibitem{Riv91}
  V.~Rivasseau. {\em Rivasseau, Vincent. From Perturbative to Constructive Renormalization.}
  \newblock Princeton University Press, (1991).
  \burlalt{doi:10.1515_9781400862085.fm}{https://www.doi.org/10.1515-9781400862085}.
	
	\bibitem{VS2007}
	W. D.~van Suijlekom.
	\newblock {\em Renormalization of Gauge Fields: A Hopf Algebra Approach.} Comm. Math. Phys. \textbf{276}, (2007), 773--798.
	\burlalt{doi:10.1007/s00220-007-0353-9}{https://dx.doi.org/10.1007/s00220-007-0353-9}.

  \bibitem{Wil71}
  K.~Wilson. \newblock {\em Renormalization Group and Critical Phenomena. I. Renormalization Group and the Kadanoff Scaling Picture}
  Phys. Rev. B, \textbf{4}, (1971), 3174–3183,.
  \burlalt{doi:10.1103/PhysRevB.4.3174}{https://dx.doi.org/10.1103/PhysRevB.4.3174}.
	
	\bibitem{WZ69}
	W. ~Zimmermann.
	\newblock{\em Convergence of Bogoliubov’s method of renormalization in momentum space.}
	\newblock {Comm. Math. Phys. \textbf{15}, (1969), 208--234.}
	\newblock \burlalt{doi:10.1007/BF01645676}{http://dx.doi.org/10.1007/BF01645676}.


\end{thebibliography}
\end{document}